%% file: paper.tex
\pdfoutput = 1
\documentclass[a4paper,11pt]{article}
\usepackage[margin=2cm]{geometry}

\usepackage{commands}

\definecolor[named]{urlblue}{cmyk}{1,0.58,0,0.21}

\hypersetup{
	breaklinks=true,
	colorlinks=true,
	citecolor=purple!70!blue!60!black,
	linkcolor=purple!70!blue!90!black,
	urlcolor=urlblue,
	pdflang={en},
	pdfborder={0 0 0},
	pdftitle={Approximate Monotone Local Search for Weighted Problems},
	pdfauthor={Bar\i\c{s} Can Esmer, Ariel Kulik, D{\'{a}}niel Marx, Daniel Neuen, and Roohani Sharma},
	pdfcreator={},
}

\title{Approximate Monotone Local Search for Weighted Problems}

\author[1]{Bar\i\c{s} Can Esmer\thanks{The author is part of Saarbrücken Graduate School of Computer Science, Germany.}}
\author[1]{Ariel Kulik\thanks{This project has received funding from the European Union’s Horizon 2020 research and innovation programme under grant agreement No 852780-ERC (SUBMODULAR).}}
\author[1]{D{\'{a}}niel Marx\thanks{Research supported by the European Research Council (ERC) consolidator grant No.~725978 SYSTEMATICGRAPH.}}
\author[2]{Daniel Neuen}
\author[3]{Roohani Sharma}
\affil[1]{CISPA Helmholtz Center for Information Security, Saarbr\"ucken, Germany. \{\texttt{baris-can.esmer|ariel.kulik|marx}\}\texttt{@cispa.de}}
\affil[2]{University of Bremen, Bremen, Germany. \texttt{dneuen@uni-bremen.de}}
\affil[3]{Max Planck Institute for Informatics, Saarland Informatics Campus, Saarbr\"ucken, Germany. \texttt{rsharma@mpi-inf.mpg.de}}

\begin{document}

\maketitle

\begin{abstract}
 In a recent work, Esmer et al.\ describe a simple method -- Approximate Monotone Local Search -- to obtain exponential approximation algorithms from existing parameterized exact algorithms, polynomial-time approximation algorithms and, more generally, parameterized approximation algorithms.
 In this work, we generalize those results to the weighted setting.

 More formally, we consider monotone subset minimization problems over a weighted universe of size $n$ (e.g., \textsc{Vertex Cover}, \textsc{$d$-Hitting Set} and \textsc{Feedback Vertex Set}).
 We consider a model where the algorithm is only given access to a subroutine that finds a solution of weight at most $\alpha \cdot W$ (and of arbitrary cardinality) in time $c^k \cdot n^{\OO(1)}$  where~$W$ is the minimum  weight of a solution of cardinality  at most $k$.
 In the unweighted setting, Esmer et al.\ determine the smallest value $d$
 for which a $\beta$-approximation algorithm running in time $d^n \cdot n^{\OO(1)}$ can be obtained in this model. 
 We show that the same dependencies also hold in a weighted setting in this model: 
 for every fixed $\eps>0$ 
 we obtain a $\beta$-approximation algorithm running in time $\OO\left((d+\eps)^{n}\right)$, for the same $d$ as in the unweighted setting. 
 
 Similarly, we also extend a $\beta$-approximate brute-force search
 (in a model which only provides access to a membership oracle) to the weighted setting.
 Using existing approximation algorithms and exact parameterized algorithms for weighted problems, we obtain the first exponential-time  $\beta$-approximation algorithms that are better than brute force for a variety of problems including \textsc{Weighted Vertex Cover}, \textsc{Weighted $d$-Hitting Set}, \textsc{Weighted Feedback Vertex Set} and \textsc{Weighted Multicut}.
\end{abstract}

\section{Introduction}
\input{introduction}

\section{Our Results}
\label{sec:our_results}
\input{our_results}

\section{Applications}
\label{sec:applications}
\input{applications}

\section{Weighted Approximate Brute Force}
\label{sec:brute_force}
\input{membership_oracle}

\section{Weighted Monotone Local Search}
\label{sec:weighted_monotone_local_search}
\input{extension_oracles}

\section{Discussion}
\label{sec:discussion}
\input{discussion}

\bibliographystyle{plainurl}
\bibliography{refs}

\newpage
\appendix 

\section{Basics on O-Notation}
\label{sec:o_notation}
\input{o_notation}

\section{The Definition of $\amlsbound$}
\label{sec:amls_def}
\input{amls}

\section{Continuity properties}
\label{sec:eps_subexponential}
\input{eps_subexponential}

\newpage

\section{Problem Definitions}
\label{sec:problem_def}
\input{problem_definitions}

\newpage

\section{Running Times}
\label{sec:running_times}
\input{running_times}

\end{document}

%% file: introduction.tex
In this work, we are interested in \emph{subset problems}, where the goal is to find a subset of a given $n$-sized universe $U$ that satisfies some property $\Pi$ (e.g., \textsc{Vertex Cover}, \textsc{Hitting Set}, \textsc{Feedback Vertex Set}, \textsc{Mulitcut}).
Such problems can trivially be solved in time $\OO^*(2^n)$\footnote{The $\OO^*$ notation hides polynomial factors in the expression.}, and in the past decades there has been great interest in designing algorithms that beat this exhaustive search and run in time $\OO^*\left(d^n\right)$ for as small $1 < d < 2$ as possible (see, e.g., \cite{FominK10}).
On the other hand, many of the considered problems admit polynomial-time $\alpha$-approximation algorithms for some constant $\alpha > 1$ (e.g., \textsc{Vertex Cover} admits a polynomial-time $2$-approximation \cite{Bar-YehudaE81}).
To bridge the gap between exact exponential-time algorithm and polynomial-time $\alpha$-approximation algorithm for some possibly large constant $\alpha$, there has been recent interest in \emph{exponential-time approximation algorithms} \cite{EsmerKMNS23,EsmerKMNS22,AroraBS15,BansalCLNN19,BourgeoisEP11,CyganKW09,EscoffierPT16,ManurangsiT18} to obtain approximation ratios that are better than what is considered possible in polynomial time.

In a recent work, Esmer et al.\ \cite{EsmerKMNS23} describe a simple method -- Approximate Monotone Local Search -- to obtain exponential-time approximation algorithms for certain subset problems from existing parameterized exact algorithms, polynomial-time approximation algorithms and, more generally, parameterized approximation algorithms.
More precisely, the focus in \cite{EsmerKMNS23} lies on \emph{subset minimization problems} where, given a universe $U$ with $n$ elements, we are aiming to find a set $S \subseteq U$ of minimum cardinality satisfying some property $\Pi$.
To allow for approximation algorithms, we restrict to \emph{monotone} problems, i.e., the family $\CF$ of solution sets is closed under taking supersets.
In this setting, a \emph{$\beta$-approximation algorithm} is asked to return a solution set $S \in \CF$ such that $|S| \leq \beta \cdot |\OPT|$ where $\OPT$ denotes a solution of minimum size.
Given access to a parameterized $\alpha$-approximation algorithm running in time $\OO^*(c^k)$ (where the parameter $k$ denotes the size of the desired optimal solution), Esmer et al.\ \cite{EsmerKMNS23} determine the best possible value $d = \amlsbound(\alpha,c,\beta)$ 
such that a $\beta$-approximation algorithm running in time $\OO^*(d^n)$ can be obtained.
Using existing parameterized approximation algorithms, which in particular includes polynomial-time approximation and exact parameterized algorithms, this leads to the fastest exponential-time approximation algorithms for a variety of problems including \textsc{Vertex Cover}, \textsc{$d$-Hitting Set}, \textsc{Feedback Vertex Set} and \textsc{Odd Cycle Transversal}.

In this work, we are interested in \emph{weighted} monotone subset minimization problems.
Here, the universe $U$ is additionally equipped with a weight function $\w\colon U \to \NN$ and we are asking for a solution set $S \in \CF$ of minimum weight $\w(S) \coloneqq \sum_{u \in S} \w(u)$.
Accordingly, in a $\beta$-approximation algorithm, we are seeking a solution set $S \in \CF$ such that $\w(S) \leq \beta \cdot \w(\OPT)$ where $\OPT$ denotes a solution of minimum weight.
Looking at \cite{EsmerKMNS23}, it can be observed that obtained algorithms only extend to the weighted setting for the special case $\alpha = \beta = 1$.
Indeed, in this particular case Fomin, Gaspers, Lokshtanov and Saurabh \cite{FominGLS19} already show in an earlier work that an exact parameterized algorithm running in time $\OO^*(c^k)$ can be turned into an exact exponential-time algorithm running time $\OO^*((2-\frac{1}{c})^n)$.
As already pointed out in \cite{FominGLS19}, the result also holds in weighted setting and implies exact exponential-time algorithms for, e.g., \textsc{$d$-Hitting Set} and \textsc{Feedback Vertex Set}.
On the other hand, for $\beta > 1$, the algorithm presented in~\cite{EsmerKMNS23} does not work in a weighted setting.
In a nutshell, the main idea in \cite{EsmerKMNS23} is to randomly sample a set of vertices $X$ and to bound the probability that it contain a sufficiently large portion of an optimum solution $\OPT$.
However, in a weighted setting, such an approach is destined to fail since even adding a single element of large weight to an optimum solution may lead to an unbounded approximation factor.

The main contribution of this work is to adapt the tools from \cite{EsmerKMNS23} to the weighted setting.
Given a parameterized $\alpha$-approximation algorithm running in time $\OO^*(c^k)$ for weighted subset minimization (where the parameter $k$ still denotes the \emph{size} of the desired optimal solution,
we give additional details in the next paragraph), for every fixed  $\eps>0$ we obtain a $\beta$-approximation algorithm running in time $\OO\left((\amlsbound(\alpha,c,\beta) +\eps)^{n}\right)$.
Note that this matches the corresponding bound in the unweighted setting up to the additive $\eps$ in the base of the exponent in the 	running time.

To state our main result more precisely, let us formalize the requirements for the parameterized $\alpha$-approximation algorithm.
Similar to \cite{EsmerKMNS23}, we  require an \emph{$\alpha$-approximate extension algorithm}.
Such an algorithm receives as input a set $X \subseteq U$ and a number $k \geq 0$.
If there is an extension $S \subseteq U$ to a solution set (i.e., $X \cup S \in \CF$) of size at most $k$, then the algorithm outputs a set $T \subseteq U$ such that $X \cup T \in \CF$ and $\w(T) \leq \alpha \cdot \w(S^*)$ where $S^*$ is a minimum-weight extension of $X$ of size at most $k$.
Note that the size of $T$ does not need to be bounded in $k$; the parameter $k$ only restricts the size of an ``optimum solution'' which we compare against.
For example, the polynomial-time $2$-approximation algorithm for \textsc{Weighted Vertex Cover} \cite{Bar-YehudaE81} immediately results in a $2$-approximate extension algorithm:
given a graph $G$, $X \subseteq V(G)$ and $k \geq 0$, we apply the $2$-approximation algorithm to $G - X$ and return the output $T \subseteq V(G)$ (this algorithm behaves independently of $k$).
With this, our main result can informally be stated as follows.

\medskip
\begin{mdframed}[backgroundcolor=gray!20]
 Suppose a weighted monotone subset minimization problem admits an $\alpha$-approximate extension algorithm running in time $\OO^*(c^k)$.
 Then there is a $\beta$-approximation algorithm running in time $\OO^*((\amlsbound(\alpha,c,\beta) + \eps)^{n})$ for every $\eps > 0$. 
\end{mdframed}
\medskip

The basic idea to achieve this result is to partition the universe $U$ into subsets $U_i$ of elements of roughly the same weight.
We apply the results from \cite{EsmerKMNS23} to each of the sets $U_i$ separately which results in a ``query set'' for $U_i$, i.e., a set of  queries made by the algorithm from \cite{EsmerKMNS23} to the $\alpha$-approximate extension algorithm.
The crucial observation is that these queries are made in a non-adaptive way, i.e., the ``query set'' only depends on the set $U_i$.
We then combine the ``query set'' for the blocks $U_i$ into a query set for the whole weighted set $U$.
In particular, the results from \cite{EsmerKMNS23} are only used in a black-box manner.

For many problems such as \textsc{Vertex Cover} \cite{Bar-YehudaE81}, \textsc{$d$-Hitting Set} \cite{Bar-YehudaE81} and \textsc{Feedback Vertex Set} \cite{BafnaBF99}, polynomial-time approximation algorithms directly extend to the weighted setting, and provide $\alpha$-approximate extension algorithms as discussed above.
On the other hand, while many parameterized exact algorithms do not directly extend to the weighted setting, several problems have been studied in the weighted setting.
For example, for \textsc{Weighted Vertex Cover} \cite{ShachnaiZ17} provides a $\OO^*(1.363^k)$ algorithm that, given a vertex-weighted graph and a number $k \geq 0$, returns a vertex cover of weight at most $W$ where $W$ is the minimum weight of a vertex cover of size at most $k$ (if there is no vertex cover of size at most $k$, the algorithm reports failure).
As before, to obtain a $1$-approximate extension subroutine, we apply the algorithm to the graph $G - X$ (where $X$ is the input set for the extension subroutine).
Similar results are also available for \textsc{Weighted $d$-Hitting Set} \cite{ShachnaiZ17,FominGKLS10} and \textsc{Weighted Feedback Vertex Set} \cite{AgrawalKLS16}.
Also, some simple branching algorithms immediately extend to the weighted setting (e.g., deletion to $\mathcal{H}$-free graphs where $\mathcal{H}$ is a finite set of forbidden induced subgraphs).
Finally, we can also rely on parameterized approximation algorithms. 
For example, \cite{LokshtanovMRSZ21} provides such algorithms for several problems including \textsc{Weighted Directed FVS} and \textsc{Weighted Multicut}.

We remark that there are also FPT algorithms for other parameterizations in the weighted setting.
For example, Niedermeier and Rossmanith \cite{NiedermeierR03} show that \textsc{Weighted Vertex Cover} is fixed-parameter tractable parameterized by the weight $W$ of a minimum-weight vertex cover.
However, such subroutines are not useful in our setting since we aim to bound the running time of our approximation algorithms with respect to the number of vertices.

Similar to \cite{EsmerKMNS23}, we compare our algorithms to a brute-force search. %
Here, we consider a setting where the weighted monotone subset minimization problem can only be accessed via a membership oracle, i.e., given a set $X \subseteq U$, we can test (in polynomial time) whether $X$ is a solution set.
For every $\beta \geq 1$ we define $\brute(\beta) \coloneqq 1 + \exp\left(-\beta \cdot  \HH\left(\frac{1}{\beta}\right) \right)$ where $\HH(\beta) \coloneqq -\beta \ln \beta - (1-\beta) \ln (1-\beta)$ denotes the entropy function.
In \cite{EsmerKMNS22} it has been shown that, in the unweighted setting, there is a $\beta$-approximation algorithm running in time $\OO^*((\brute(\beta))^n)$ that only exploits a  membership test.
We also extend this result to the weighted setting, i.e., we show that any weighted monotone subset minimization problem can be solved in time $(\brute(\beta))^{n + o(n)}$ given only a membership oracle.

Esmer et al.\ \cite{EsmerKMNS23} show that $\amls(\alpha,c,\beta) < \brute(\beta)$ for all $\alpha,c \geq 1$ and $\beta > 1$.
Since the same bounds are achieved in the weighted setting, all algorithms obtained above are strictly faster than the brute-force $\beta$-approximation algorithm.
In particular, we obtain exponential-time approximation algorithms that are faster than the approximate brute-force search for the weighted versions of
\textsc{Vertex Cover}, \textsc{$d$-Hitting Set}, \textsc{Feedback Vertex Set}, \textsc{Tournament FVS}, \textsc{Subset FVS}, \textsc{Cluster Graph Vertex Deletion}, \textsc{Cograph Vertex Deletion}, \textsc{Split Vertex Deletion}, \textsc{Partial Vertex Cover},
{\sc Directed Feedback Vertex Set}, {\sc Directed Subset FVS}, {\sc Directed Odd Cycle Transversal} and {\sc Multicut} (all problems are defined in Appendix \ref{sec:problem_def}).

\subsection{Organization}

The paper is organized as follows.
In \Cref{sec:our_results} we state the problems we want to address, provide the necessary definitions and notation, and formally state our main results.
In \Cref{sec:applications}, we demonstrate how our methods can be applied to specific problems to obtain exponential-time approximation algorithms.
\Cref{sec:brute_force} contains the proof of \Cref{thm:brute_force}, our result on exponential-time approximation algorithms for the $\WSM$ problem in the membership model.
Similarly, \Cref{sec:weighted_monotone_local_search} contains the proof of \Cref{thm:weighted_amls}, our result on exponential-time approximation algorithms for the $\WSM$ problem in the extension model.
Finally, in \Cref{sec:discussion} we conclude the paper by summarizing our main contributions and key findings.

%% file: our_results.tex
To formally state our results we use an abstract notion of a problem and oracle-based computational models.
Let $U$ be a finite set of elements.
We use $n$ to denote the cardinality of $U$.
A set system $\CF$ of $U$ is a family $\CF\subseteq 2^U$ of subsets of $U$.
We say the set system $\CF$ is \emph{monotone} if (i) $U\in \CF$ and (ii) for all $T\subseteq S\subseteq U$, if $T\in \CF$ then $S\in \CF$ as well.

An instance of the \emph{Weighted Monotone Subset Minimization problem} ($\WSM$) is a triplet $(U,\w,\CF)$ where  $U$ is a finite set, $\w\colon U\to \NN$ is a weight function over the elements of $U$, and~$\CF$ is a monotone set system of $U$.
The set of solutions is $\CF$ and the objective is to find $S\in \CF$ with minimum total weight $\w(S) \coloneqq \sum_{e\in S} \w(e)$.
We use $\opt(U,\w,\CF) \coloneqq \min\{\w(S) \mid S\in \CF\}$ to denote the optimum value of a solution to the $\WSM$ instance $(U,\w,\CF)$.
We refer to the special case in which  $\w(u)=1$ for all $u\in U$ as the  \emph{Unweighted Monotone Subset Minimization problem} (\USM).

The weighted monotone subset minimization problem is a meta-problem which captures multiple well studied problems as special cases, e.g., \textsc{Weighted Vertex Cover}, \textsc{Weighted Feedback Vertex Set} and \textsc{Weighted Multicut}.
We study the problem using two computational models.
In both models the set $U$ and the weight function $\w$ are given as part of the input.
The set $\CF$ can only be accessed using an oracle, and the models differ in the type of supported oracle queries.

\subsection{Membership Oracles and Weighted Approximate Brute Force}

In the \emph{membership model} the input to the algorithm is a universe $U$ and a weight function~$\w\colon U \to \NN$.
Additionally, the algorithm has access to a membership oracle for a monotone set system $\CF$ of $U$, that is, the algorithm can check if a subset $S\subseteq U$ satisfies $S\in \CF$ in a single step.
For every $\alpha\geq 1$, we say an algorithm is an $\alpha$-approximation for $\WSM$ ($\USM$) in the membership model if for every $\WSM$ ($\USM$) instance $(U,\w,\CF)$ the algorithm returns a set $S\in \CF$ such that $\w(S)\leq \alpha \cdot\opt(U,\w,\CF)$.

One can easily attain a $1$-approximation algorithm for $\WSM$ in the membership model by iterating over all subsets of $U$ and querying the oracle for each.
This leads to an algorithm with running time $\OO(2^{n})$.
Moreover, it can be easily shown there is no $1$-approximation algorithm for $\WSM$ (or for $\USM$) in the membership model which runs in time $\OO\left( (2-\eps)^{n}\right)$. We refer to this algorithm as the (exact) brute force.

Intuitively, for every $\alpha>1$, it should be possible to design an $\alpha$-approximation algorithm for $\WSM$ and $\USM$ in the membership   model which runs in time $\OO(c^n)$ for some $c<2$. However, the value of the optimal $c$ in this setting is not obvious.
In \cite{EsmerKMNS22} the authors studied $\USM$ in the membership model and pinpointed the right  value of $c$.
For every $\alpha \geq 1$ we define
\begin{equation}\label{eq:brute_defn}
	\brute(\alpha) = 1+\exp\left(-\alpha \cdot \entropy \left(\frac{1}{\alpha } \right)\right),
\end{equation}
where $\entropy(x) = -x\ln(x) -(1-x) \ln(1-x)$ is the entropy function.

\begin{lemma}[{\cite[Theorem 5.1]{EsmerKMNS22}}]
	\label{lem:EKMNS22}
	For every $\alpha \geq 1$ the following holds.
	\begin{enumerate}
		\item There is a deterministic $\alpha$-approximation algorithm for $\USM$ in the membership model which runs in time $\left( \brute(\alpha)\right)^n\cdot n^{\OO(1)}$.
		\item Let $\eps > 0$. There is no $\alpha$-approximation algorithm for $\USM$ in the membership model which runs in time $\left( \brute(\alpha)-\eps\right)^n\cdot n^{\OO(1)}$.
	\end{enumerate}
\end{lemma}

As the algorithmic result in \Cref{lem:EKMNS22} can be viewed as an approximate analogue of the brute-force algorithm, it is commonly referred as {\em $\alpha$-approximate brute force}.
The lower bound given in \cite{EsmerKMNS22} also holds if the algorithm is allowed to use randomization. 
As $\WSM$ is a generalization of $\USM$, the following corollary is an immediate consequence of \Cref{lem:EKMNS22}.

\begin{corollary}
	\label{cor:brute_lowerbound}
	For every $\alpha \geq 1$ and $\eps>0$
	there is no deterministic $\alpha$-approximation algorithm for $\WSM$ in the membership model which runs in time $\left( \brute(\alpha)-\eps\right)^n\cdot n^{\OO(1)}$.
\end{corollary}

As the bound in \Cref{lem:EKMNS22} also holds if randomization is allowed, the same holds true for the bound in \Cref{cor:brute_lowerbound}.
Our first result is a generalization of the approximate brute-force algorithm of \cite{EsmerKMNS22} for the weighted setting.
That is, we provide an algorithm which matches the lower bound in \Cref{cor:brute_lowerbound} up to a sub-exponential factor.

\begin{theorem}[Weighted Approximate Brute Force]
	\label{thm:brute_force}
	For every $\alpha> 1$ there is an $\alpha$-approximation algorithm for $\WSM$ in the membership model which runs in time $\left(\brute(\alpha)\right)^{n+o(n)}$.
\end{theorem}

The proof of \Cref{thm:brute_force} is based on a rounding of the weight function $\w$ and utilizes a construction from \cite{EsmerKMNS22} which is applied to each of the rounded weight classes.

\subsection{Extension Oracles and Weighted Approximate Monotone Local Search}

Our second computational model deals with \emph{extension oracles}.
The input for these oracles is a set $S\subseteq U$  and a number $\ell\in \NN_{\geq 0}$ and the output is an \emph{extension} of $S$, that is, a set  $X\subseteq U$ such that $S\cup X \in \CF$.
Furthermore, the returned set $X$ is guaranteed to have a small weight in comparison to the minimum weight extension of $S$ which contains at most $\ell$ elements.
For multiple problems, such as \textsc{Vertex Cover} and \textsc{Feedback Vertex Set}, these oracles can be implemented using existing parameterized algorithms which have a running time of the form $c^{\ell} \cdot n^{\OO(1)}$.
We therefore associate a running time of $c^{\ell }$  with the query $(S,\ell)$. The formal definition of extension oracles is as follows.

\begin{definition}[Extension Oracle]
	\label{def:extension}
	Let $(U,\w,\CF)$ be a $\WSM$ instance and let $\alpha\geq 1$. An $\alpha$-extension oracle for $(U,\w,\CF)$ is a function $\oracle:2^U\times \NN_{\geq 0} \rightarrow 2^U$ such that for every $S\subseteq U$ and $\ell \in \NN_{\geq 0}$ the following holds:
	\begin{enumerate}
		\item $\oracle(S,\ell) \cup S\in \CF$.
		\item $\w\left(\oracle (S,\ell)\right) \leq \alpha \cdot \min\{ \w(X) \mid X\subseteq U,~\abs{X}\leq \ell,~X\cup S\in \CF \}$ (we set $\min\emptyset =\infty$).
	\end{enumerate}
\end{definition}

In the {\em $(\alpha,c)$-extension model} the input for the algorithm is a finite set $U$ and a weight function $\w\colon U\rightarrow \NN$.
Furthermore, the algorithm is given oracle access to an  $\alpha$-extension oracle $\oracle$ of $(U,\w,\CF)$ for some monotone set system $\CF$ of $U$.
For every $\beta \geq 1$, we say an algorithm is a $\beta$-approximation algorithm  for $\WSM$ ($\USM$) in the  $(\alpha,c)$-extension model if for every $\WSM$ ($\USM$) instance $(U,\w,\CF)$  and $\alpha$-extension oracle $\oracle$ of the instance, the algorithm returns $T\in \CF$ such that $\w(T)\leq \beta \cdot \opt(U,\w,\CF)$.
The running time of an algorithm in this model is the number of computations steps plus $c^\ell$ for every query $(S,\ell)$ issued to the oracle during the execution.
Following the standard worst case analysis convention, we say an algorithm runs in time $f(n)$ if for every $\WSM$ instance $(U,\w,\CF)$ and $\alpha$-extension oracle $\oracle$ of the instance the algorithm runs in time at most $f\left(\abs{U}\right)$.

The $(\alpha,c)$-extension  model has been studied in \cite{EsmerKMNS23} for the  special case of $\USM$.
The authors of \cite{EsmerKMNS23} provided a deterministic $\beta$-approximation algorithm in the $(\alpha,c)$-extension model, known as \emph{deterministic approximate monotone local search}, which runs in time $\left(\amlsbound(\alpha,c,\beta)\right)^{n+o(n)}$, where $\amlsbound(\alpha,c,\beta)$ is defined as the optimal value of a continuous optimization problem.
Throughout the paper we use $\amlsbound$ to denote this function.
We provide the formal definition of $\amlsbound$ in \Cref{sec:amls_def} for completeness.
We note this formal definition is not required for the understanding of the results in this paper.
It was also shown in \cite{EsmerKMNS23} that the value of $\amlsbound(\alpha,c,\beta)$ can be computed up to precision of $\eps$ in time polynomial in the encoding length of $\alpha$, $c$, $\beta$ and $\eps$.

This algorithmic result was complemented in \cite{EsmerKMNS23} with a matching lower bound.

\begin{lemma}[\cite{EsmerKMNS23}]
	\label{lem:alpha_neq_beta_lowerbound}
	 For every $\alpha,\beta,c\geq 1$ and $\eps>0$ there is no deterministic $\beta$-approximation for $\USM$ in the $(\alpha,c)$-extension model which runs in time $n^{\OO(1)}\cdot \left( \amlsbound(\alpha,c,\beta)-\eps\right)^{n}$.
\end{lemma}

Furthermore, it was shown that the running time of deterministic approximate monotone local search is better than brute force for $\beta>1$. 

\begin{lemma}[\cite{EsmerKMNS23}]
	\label{lem:better_than_brute}
	For all $\alpha,c\geq 1$ and $\beta>1$ it holds that $\amlsbound(\alpha,c,\beta) <\brute(\beta)$. 
\end{lemma}

The results of \cite{EsmerKMNS23} also include a randomized algorithm which omits the subexponential factor in the running time and the lower bound also holds for randomized algorithms.
In \cite{EsmerKMNS23} the authors used the deterministic approximate monotone local search algorithm to obtain exponential-time approximation algorithms for multiple {\em unweighted} problems such as \textsc{Vertex Cover} and \textsc{Feedback Vertex Set}.
We generalize the algorithmic results of \cite{EsmerKMNS23} to the weighted setting and similarly use it to obtain exponential-time approximation algorithms  for {\em weighted} variants of the mentioned problems (see \Cref{sec:applications}).

Since $\USM$ is a special case of $\WSM$, \Cref{lem:alpha_neq_beta_lowerbound} immediately implies the following
\begin{corollary}
	\label{cor:weighted_amls_lowerbound}
	For every $\alpha,c,\beta\geq 1$  and $\eps>0$ there is no deterministic $\beta$-approximation for $\WSM$ in the $(\alpha,c)$-extension model which runs in time $n^{\OO(1)}\cdot \left( \amlsbound(\alpha,c,\beta)-\eps\right)^{n}$.
\end{corollary}

Our second result is an algorithm which matches the running time in \Cref{cor:weighted_amls_lowerbound} to a factor of $\eps$  in the running time.
\begin{theorem}[Weighted Approximate Montone Local Search]
	\label{thm:weighted_amls}
	For every $\alpha,c\geq 1$, $\beta >1$ and $\eps>0$ there is a deterministic $\beta$-approximation for $\WSM$ in the $(\alpha,c)$-extension model which runs in time $\OO\left(\left( \amlsbound(\alpha,c,\beta)+\eps\right)^{n}\right)$.
\end{theorem}

The proof of \Cref{thm:weighted_amls} is similar to the one of \Cref{thm:brute_force}.
It applies rounding to the weights on the elements, and then utilizes a construction from \cite{EsmerKMNS23} for each of the weight classes.

The result of \Cref{thm:weighted_amls} only applies for $\beta>1$.
If $\alpha=\beta=1$ the result of \cite{FominGLS19} can be used to obtain an exact algorithm of running time $n^{\OO(1)}\cdot \left(2-\frac{1}{c}\right)^n$.
For $\alpha>\beta=1$ \Cref{cor:weighted_amls_lowerbound} implies that the best possible running time is $\OO^*(2^{n})$ which can be attained by brute force.

%% file: applications.tex
Our results provide exponential-time approximation algorithms for a variety of weighted vertex-deletions problems.
For illustration purposes, let us first focus on the  \textsc{Weighted Vertex Cover} problem where we are given a graph $G$ with vertex weights $\w \colon V(G) \to \NN$, and we ask for a vertex cover $S \subseteq V(G)$ of minimum weight $\w(S) = \sum_{v \in S} \w(v)$.
It is well-known that \textsc{Weighted Vertex Cover} admits a polynomial-time $2$-approximation algorithm \cite{Bar-YehudaE81}.
Also, the problem can be solved exactly in time $\OO^*(1.238^n)$ \cite{Wahlstrom08}.

For the unweighted version \textsc{Unweighted Vertex Cover}, Bourgeois, Escoffier and Paschos~\cite{BourgeoisEP11} designed several exponential-time approximation algorithms for approximation ratios in the range $(1,2)$.
For example, they obtain a $1.1$-approximation algorithm running in time $\CO^*(1.127^n)$ where $n$ denotes the number of vertices of the input graph.
These running times are further improved in \cite{EsmerKMNS22,EsmerKMNS23} using the framework of Approximate Monotone Local Search.
Indeed, the fastest known $1.1$-approximation algorithm for \textsc{Unweighted Vertex Cover} runs in time $\CO^*(1.113^n)$ \cite{EsmerKMNS23}.

For the weighted version, no such results have been obtained so far.
We use Theorem \ref{thm:weighted_amls} to design the first exponential $\beta$-approximation algorithms for \textsc{Weighted Vertex Cover} for all $\beta \in (1,2)$.
For the extension oracle, we can rely on the well-known polynomial-time $2$-approximation algorithm.
Given a set $S \subseteq V(G)$, we delete all vertices in $S$ and apply the $2$-approximation algorithm to the graph $G - S$ which outputs a vertex cover $X$ such that $\w(X) \leq 2 \cdot \w(\OPT)$ where $\OPT$ denotes a minimum vertex cover of $G - S$.
As a result, we can implement a $2$-extension oracle in polynomial time which corresponds (up to polynomial factors) to cost $c = 1$.
So Theorem \ref{thm:weighted_amls} results in a $\beta$-approximation algorithm for \textsc{Weighted Vertex Cover} which runs in time $\OO^*(\left(\amlsbound(\alpha,c,\beta) + \eps\right)^{n})$ for every $\beta > 1$, where $\alpha = 2$ and $c = 1$.
A visualization is given in Figure \ref{fig:runtimes_vc}.

\begin{figure}
	\centering
	\begin{subfigure}{.5\textwidth}
		\centering
		\input{plots/figure_vc.tex}
		\label{fig:vc_results}
	\end{subfigure}%
	\begin{subfigure}{.5\textwidth}
		\centering
		\input{plots/figure_vc_comp.tex}
		\label{fig:vc_comp_results}
	\end{subfigure}%
	\caption{The left figure shows running times for \textsc{Weighted Vertex Cover} and right side provides a comparison to \textsc{Unweighted Vertex Cover}.
		A dot at $(\beta,d)$ means that the respective algorithm outputs an $\beta$-approximation in time $\CO^*(d^n)$.}
	\label{fig:runtimes_vc}
\end{figure}
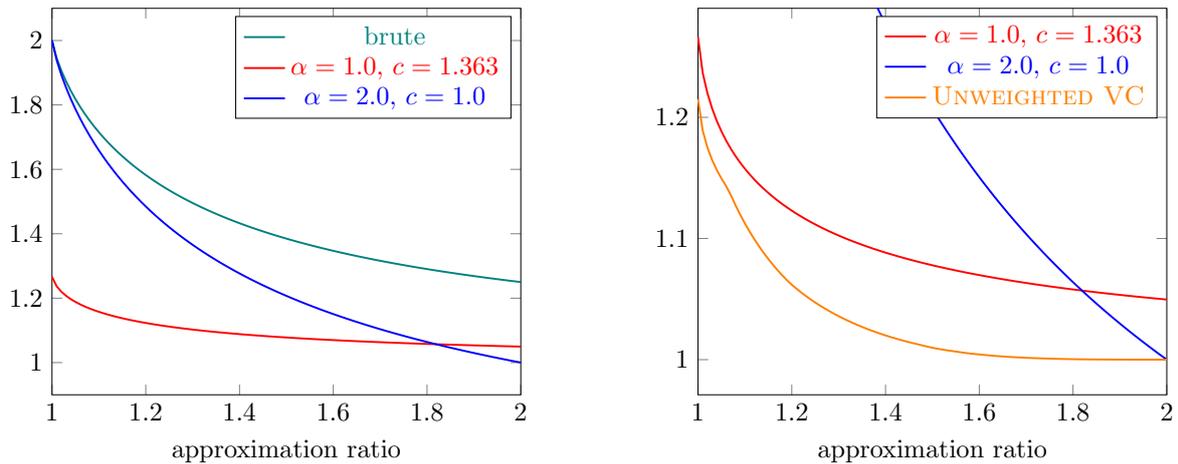

Instead of using a polynomial-time $2$-approximation algorithm, we can also rely on existing FPT algorithms for \textsc{Weighted Vertex Cover} to simulate the extension oracle.
Let us point out that different parameterizations have been considered for \textsc{Weighted Vertex Cover} in the literature.
For example, Niedermeier and Rossmanith \cite{NiedermeierR03} give FPT algorithms for \textsc{Weighted Vertex Cover} parameterized by the weight of the optimal solution.
However, in light of the computational model introduced above, we require FPT algorithms parameterized by the \emph{size} of the solution.
Given a graph $G$ with vertex weights $\w \colon V(G) \to \NN$ and an integer $k \geq 0$, we ask for a vertex cover of weight at most $W$ where $W$ is the minimum weight of a vertex cover of size at most $k$.
The best known FPT algorithm (parameterized by $k$) for this problem runs in time $1.363^k \cdot n^{\CO(1)}$ \cite{ShachnaiZ17}.
This algorithm provides an extension algorithm with parameters $\alpha = 1$ and $c = 1.363$.
Using Theorem \ref{thm:weighted_amls}, we obtain a $\beta$-approximation algorithm for \textsc{Weighted Vertex Cover} running in time $\OO^*(\left(\amlsbound(\alpha,c,\beta) + \eps\right)^{n})$.
For $\beta = 1.1$, we obtain a running time of $\OO^*(1.158^{n})$.

\begin{table}
	\centering
	\input{plots/table_vc_extended.tex}
	\caption{The table shows the running times for \textsc{Weighted Vertex Cover} (second and third row) and \textsc{Unweighted Vertex Cover} (last row).
		An entry $d$ in column $\beta$ means that the respective algorithm outputs a $\beta$-approximation in time $\CO^*(d^n)$.}
	\label{tab:table_vc}
\end{table}

We provide running times for selected approximation ratios for both algorithms in Table \ref{tab:table_vc} and a graphical comparison in Figure \ref{fig:runtimes_vc}.
We also compare the running times to the approximate brute-force search and the best algorithms in the unweighted setting \cite{EsmerKMNS23}.
It can be observed that the second algoritm (using the FPT algorithm as an extension subroutine) is faster for $\beta \lesssim 1.82$.
Also, there is still a noticeable gap to the unweighted setting.
This can be mainly explained by the fact that the approximation algorithm for the unweighted setting \cite{EsmerKMNS23} relies on parameterized approximation algorithms for \textsc{Unweighted Vertex Cover} \cite{KulikS20} which are currently unavailable in the weighted setting.

\begin{table}
	\centering
	\begin{tabular}{l|lr|c|lr|c|}
		Problem & $c_1$ & & det. & $\alpha_2$ & & det.\\
		\hline
		\textsc{Vertex Cover} & $1.363$ & \cite{ShachnaiZ17} & \cmark & $2$ & \cite{Bar-YehudaE81} & \cmark \\
		\hline
		\textsc{FVS} & $3.618$ & \cite{AgrawalKLS16} & \cmark & $2$ & \cite{BafnaBF99} & \cmark \\
		\hline
		\textsc{Tournament FVS} & $2.0$ & & \cmark & $3$ & & \cmark \\
		\hline
		\textsc{Subset FVS} & - & & \cmark & $8$ & \cite{EvenNZ00} & \cmark \\
		\hline
		\textsc{$3$-Hitting Set} & $2.168$ & \cite{ShachnaiZ17} & \cmark & $3$ & \cite{Bar-YehudaE81} & \cmark \\
		\hline
		\textsc{$d$-Hitting Set} ($d \geq 4$) & $d - 0.832$ & \cite{FominGKLS10} & \cmark & $d$ & \cite{Bar-YehudaE81} & \cmark \\
		\hline
		\textsc{Cluster Graph Vertex Deletion} & $2.168$ & & \cmark & $3$ & & \cmark \\
		\hline
		\textsc{Cograph Vertex Deletion} & $3.168$ & & \cmark & $4$ & & \cmark \\
		\hline
		\textsc{Split Vertex Deletion} & $4.168$ & & \cmark & $5$ & & \cmark\\
		\hline
		\textsc{Partial Vertex Cover} & - & & \cmark & $2$ & \cite{BshoutyB98} & \cmark \\
		\hline
	\end{tabular}
	\caption{List of weighted deletion problems admitting an single-exponential parameterized algorithm running in time $O^*(c_1^k)$ and/or a polynomial-time $\alpha_2$-approximation algorithm.
		The problems \textsc{Tournament FVS}, \textsc{Cluster Graph Vertex Deletion}, \textsc{Cograph Vertex Deletion} and \textsc{Split Vertex Deletion} can be easily reduced to \textsc{$d$-Hitting Set} for appropriate values of $d$ by exploiting known characterizations in terms of forbidden induced subgraphs.
		Additionally, for \textsc{Tournament FVS} we can rely on iterative compression to obtain a $\OO^*(2^k)$ algorithm (see, e.g., \cite{CyganFKLMPPS15}).}
	\label{tab:overview_problems}
\end{table}

We stress that our results are not limited to \textsc{Weighted Vertex Cover}, but they are applicable to various vertex-deletion problems including \textsc{Weighted Feedback Vertex Set} and \textsc{Weighted $d$-Hitting Set} (see Figure \ref{fig:runtimes_fvs_3hs}).
Table \ref{tab:overview_problems} gives an overview on problems for which we obtain exponential approximation algorithms by simulating the extension oracle by an FPT algorithm (parameterized by solution size) running in time $c_1^k \cdot n^{\CO(1)}$ or a polynomial-time $\alpha_2$-approximation algorithm.
For all the problems listed in Table \ref{tab:overview_problems}, we obtain the first exponential $\beta$-approximation algorithms for all $\beta \in (1,\alpha_2)$.
Observe that these algorithms always outperform the approximate brute-force search by Lemma \ref{lem:better_than_brute}.
We provide data on the running times of these algorithms in Appendix \ref{sec:running_times}.

\begin{figure}
	\centering
	\begin{subfigure}{.5\textwidth}
		\centering
		\caption{\textsc{Feedback Vertex Set}}
		\input{plots/figure_fvs.tex}
		\label{fig:fvs_results}
	\end{subfigure}%
	\begin{subfigure}{.5\textwidth}
		\centering
		\caption{\textsc{$3$-Hitting Set}}
		\input{plots/figure_3hs.tex}
		\label{fig:3hs_results}
	\end{subfigure}%
	\caption{The figure shows running times for \textsc{Feedback Vertex Set} and \textsc{$3$-Hitting Set}.
		A dot at $(\beta,d)$ means that the respective algorithm outputs an $\beta$-approximation in time $O^*(d^n)$.}
	\label{fig:runtimes_fvs_3hs}
\end{figure}
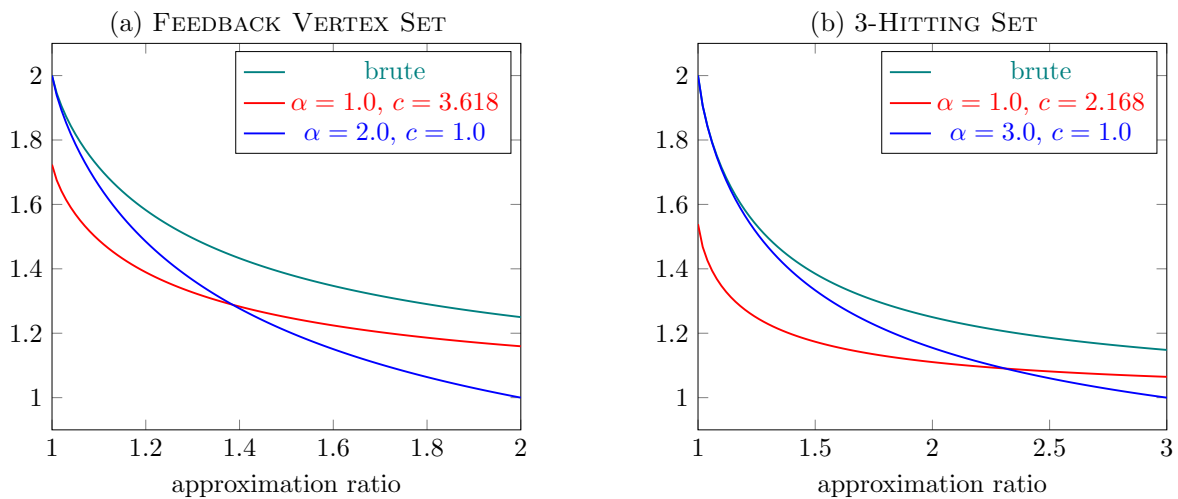

Finally, let us point out that our methods are also applicable if there is a parameterized approximation algorithm for a weighted vertex deletion problem, i.e., an $\alpha$-approximation algorithm running time $c^k \cdot n^{\CO(1)}$.
In \cite{LokshtanovMRSZ21}, such algorithms have been obtained for \textsc{Weighted Directed Feedback Vertex Set}, \textsc{Weighted Directed Subset FVS}, \textsc{Weighted Directed Odd Cycle Transversal} and \textsc{Weighted Multicut}.
As a result, we also obtain exponential $\beta$-approximation algorithms for these problems that outperform the approximate brute-force search.

%% file: plots/figure_vc.tex
\begin{tikzpicture}[scale = 0.9]
	\begin{axis}[xmin = 1, xmax = 2, ymin = 0.9, ymax = 2.1, xlabel = {approximation ratio}]

	\addplot[teal, thick] coordinates {
		(1.0, 2.0)
		(1.01, 1.9454431345)
		(1.02, 1.9062508454)
		(1.03, 1.8731549013)
		(1.04, 1.8440491304)
		(1.05, 1.8178991111)
		(1.06, 1.794081097)
		(1.07, 1.7721758604)
		(1.08, 1.7518817491)
		(1.09, 1.7329712548)
		(1.1, 1.7152667656)
		(1.11, 1.698625911)
		(1.12, 1.6829321593)
		(1.13, 1.6680884977)
		(1.14, 1.6540130203)
		(1.15, 1.6406357531)
		(1.16, 1.6278963084)
		(1.17, 1.615742114)
		(1.18, 1.6041270506)
		(1.19, 1.5930103866)
		(1.2, 1.5823559323)
		(1.21, 1.5721313608)
		(1.22, 1.5623076557)
		(1.23, 1.552858658)
		(1.24, 1.5437606905)
		(1.25, 1.534992244)
		(1.26, 1.5265337131)
		(1.27, 1.5183671728)
		(1.28, 1.510476187)
		(1.29, 1.5028456449)
		(1.3, 1.4954616201)
		(1.31, 1.4883112476)
		(1.32, 1.4813826173)
		(1.33, 1.4746646809)
		(1.34, 1.4681471696)
		(1.35, 1.4618205222)
		(1.36, 1.4556758212)
		(1.37, 1.4497047356)
		(1.38, 1.443899471)
		(1.39, 1.4382527242)
		(1.4, 1.4327576428)
		(1.41, 1.427407789)
		(1.42, 1.4221971069)
		(1.43, 1.4171198929)
		(1.44, 1.4121707695)
		(1.45, 1.4073446605)
		(1.46, 1.4026367697)
		(1.47, 1.3980425604)
		(1.48, 1.3935577374)
		(1.49, 1.3891782307)
		(1.5, 1.3849001795)
		(1.51, 1.380719919)
		(1.52, 1.3766339674)
		(1.53, 1.3726390137)
		(1.54, 1.3687319076)
		(1.55, 1.3649096489)
		(1.56, 1.3611693784)
		(1.57, 1.3575083696)
		(1.58, 1.3539240206)
		(1.59, 1.3504138468)
		(1.6, 1.3469754742)
		(1.61, 1.343606633)
		(1.62, 1.3403051519)
		(1.63, 1.3370689522)
		(1.64, 1.3338960432)
		(1.65, 1.3307845174)
		(1.66, 1.3277325456)
		(1.67, 1.3247383734)
		(1.68, 1.3218003168)
		(1.69, 1.3189167589)
		(1.7, 1.3160861463)
		(1.71, 1.3133069861)
		(1.72, 1.3105778425)
		(1.73, 1.3078973347)
		(1.74, 1.3052641335)
		(1.75, 1.3026769593)
		(1.76, 1.3001345795)
		(1.77, 1.2976358065)
		(1.78, 1.2951794954)
		(1.79, 1.2927645423)
		(1.8, 1.2903898821)
		(1.81, 1.2880544871)
		(1.82, 1.2857573652)
		(1.83, 1.2834975581)
		(1.84, 1.2812741403)
		(1.85, 1.2790862173)
		(1.86, 1.2769329244)
		(1.87, 1.2748134255)
		(1.88, 1.2727269118)
		(1.89, 1.2706726008)
		(1.9, 1.2686497351)
		(1.91, 1.2666575813)
		(1.92, 1.2646954293)
		(1.93, 1.2627625911)
		(1.94, 1.2608584001)
		(1.95, 1.2589822102)
		(1.96, 1.2571333949)
		(1.97, 1.2553113469)
		(1.98, 1.2535154767)
		(1.99, 1.2517452127)
		(2.0, 1.25)
	};

	\addplot[ red , thick] coordinates {
		(1.0, 1.2663242847)
		(1.01, 1.2367469348)
		(1.02, 1.2200190205)
		(1.03, 1.2073852094)
		(1.04, 1.1970937089)
		(1.05, 1.1883673416)
		(1.06, 1.1807772732)
		(1.07, 1.1740569098)
		(1.08, 1.1680272319)
		(1.09, 1.1625612928)
		(1.1, 1.1575653247)
		(1.11, 1.1529678327)
		(1.12, 1.1487128983)
		(1.13, 1.1447558616)
		(1.14, 1.1410604257)
		(1.15, 1.1375966491)
		(1.16, 1.1343395157)
		(1.17, 1.131267891)
		(1.18, 1.128363745)
		(1.19, 1.1256115629)
		(1.2, 1.1229978911)
		(1.21, 1.1205109827)
		(1.22, 1.1181405165)
		(1.23, 1.1158773732)
		(1.24, 1.1137134534)
		(1.25, 1.1116415299)
		(1.26, 1.1096551254)
		(1.27, 1.1077484119)
		(1.28, 1.1059161259)
		(1.29, 1.1041534977)
		(1.3, 1.1024561909)
		(1.31, 1.1008202517)
		(1.32, 1.0992420649)
		(1.33, 1.0977183163)
		(1.34, 1.0962459602)
		(1.35, 1.0948221912)
		(1.36, 1.0934444197)
		(1.37, 1.0921102501)
		(1.38, 1.0908174625)
		(1.39, 1.0895639955)
		(1.4, 1.0883479319)
		(1.41, 1.0871674852)
		(1.42, 1.0860209888)
		(1.43, 1.0849068849)
		(1.44, 1.0838237156)
		(1.45, 1.0827701148)
		(1.46, 1.0817448004)
		(1.47, 1.0807465679)
		(1.48, 1.079774284)
		(1.49, 1.0788268817)
		(1.5, 1.0779033547)
		(1.51, 1.0770027536)
		(1.52, 1.076124181)
		(1.53, 1.0752667885)
		(1.54, 1.074429773)
		(1.55, 1.0736123734)
		(1.56, 1.0728138682)
		(1.57, 1.0720335723)
		(1.58, 1.0712708352)
		(1.59, 1.0705250383)
		(1.6, 1.0697955934)
		(1.61, 1.06908194)
		(1.62, 1.0683835446)
		(1.63, 1.0676998982)
		(1.64, 1.0670305153)
		(1.65, 1.0663749323)
		(1.66, 1.0657327066)
		(1.67, 1.0651034148)
		(1.68, 1.064486652)
		(1.69, 1.063882031)
		(1.7, 1.0632891806)
		(1.71, 1.0627077456)
		(1.72, 1.0621373851)
		(1.73, 1.0615777724)
		(1.74, 1.061028594)
		(1.75, 1.0604895489)
		(1.76, 1.0599603478)
		(1.77, 1.059440713)
		(1.78, 1.0589303774)
		(1.79, 1.0584290839)
		(1.8, 1.0579365855)
		(1.81, 1.057452644)
		(1.82, 1.0569770304)
		(1.83, 1.0565095238)
		(1.84, 1.0560499114)
		(1.85, 1.0555979879)
		(1.86, 1.0551535556)
		(1.87, 1.0547164233)
		(1.88, 1.0542864067)
		(1.89, 1.0538633278)
		(1.9, 1.0534470148)
		(1.91, 1.0530373014)
		(1.92, 1.052634027)
		(1.93, 1.0522370364)
		(1.94, 1.0518461795)
		(1.95, 1.0514613108)
		(1.96, 1.0510822898)
		(1.97, 1.0507089804)
		(1.98, 1.0503412507)
		(1.99, 1.049978973)
		(2.0, 1.0496220236)
	};

	\addplot[blue , thick] coordinates {
		(1.0, 2.0)
		(1.01, 1.9387047513)
		(1.02, 1.8930929239)
		(1.03, 1.853849694)
		(1.04, 1.8188398378)
		(1.05, 1.7870069014)
		(1.06, 1.7577092525)
		(1.07, 1.7305126)
		(1.08, 1.705102307)
		(1.09, 1.6812394878)
		(1.1, 1.6587364406)
		(1.11, 1.6374417619)
		(1.12, 1.6172307723)
		(1.13, 1.5979990628)
		(1.14, 1.5796579792)
		(1.15, 1.5621313614)
		(1.16, 1.545353129)
		(1.17, 1.5292654531)
		(1.18, 1.5138173453)
		(1.19, 1.498963551)
		(1.2, 1.484663669)
		(1.21, 1.4708814415)
		(1.22, 1.457584176)
		(1.23, 1.4447422687)
		(1.24, 1.4323288084)
		(1.25, 1.4203192453)
		(1.26, 1.4086911107)
		(1.27, 1.3974237788)
		(1.28, 1.3864982635)
		(1.29, 1.3758970425)
		(1.3, 1.3656039065)
		(1.31, 1.3556038269)
		(1.32, 1.3458828408)
		(1.33, 1.3364279501)
		(1.34, 1.3272270326)
		(1.35, 1.3182687633)
		(1.36, 1.3095425448)
		(1.37, 1.3010384449)
		(1.38, 1.2927471418)
		(1.39, 1.2846598738)
		(1.4, 1.2767683951)
		(1.41, 1.2690649358)
		(1.42, 1.2615421652)
		(1.43, 1.2541931594)
		(1.44, 1.2470113718)
		(1.45, 1.2399906054)
		(1.46, 1.233124989)
		(1.47, 1.2264089541)
		(1.48, 1.2198372148)
		(1.49, 1.2134047491)
		(1.5, 1.2071067812)
		(1.51, 1.2009387662)
		(1.52, 1.1948963754)
		(1.53, 1.1889754825)
		(1.54, 1.1831721518)
		(1.55, 1.1774826264)
		(1.56, 1.1719033178)
		(1.57, 1.1664307958)
		(1.58, 1.1610617798)
		(1.59, 1.1557931299)
		(1.6, 1.1506218396)
		(1.61, 1.1455450278)
		(1.62, 1.1405599327)
		(1.63, 1.1356639047)
		(1.64, 1.1308544012)
		(1.65, 1.1261289803)
		(1.66, 1.1214852963)
		(1.67, 1.1169210943)
		(1.68, 1.112434206)
		(1.69, 1.1080225451)
		(1.7, 1.1036841036)
		(1.71, 1.0994169478)
		(1.72, 1.0952192149)
		(1.73, 1.0910891093)
		(1.74, 1.0870249001)
		(1.75, 1.0830249175)
		(1.76, 1.0790875504)
		(1.77, 1.0752112435)
		(1.78, 1.0713944949)
		(1.79, 1.067635854)
		(1.8, 1.0639339188)
		(1.81, 1.060287334)
		(1.82, 1.0566947891)
		(1.83, 1.0531550165)
		(1.84, 1.0496667894)
		(1.85, 1.0462289206)
		(1.86, 1.0428402603)
		(1.87, 1.0394996953)
		(1.88, 1.0362061466)
		(1.89, 1.0329585688)
		(1.9, 1.0297559486)
		(1.91, 1.0265973033)
		(1.92, 1.0234816797)
		(1.93, 1.0204081532)
		(1.94, 1.0173758263)
		(1.95, 1.0143838281)
		(1.96, 1.0114313128)
		(1.97, 1.008517459)
		(1.98, 1.0056414688)
		(1.99, 1.002802567)
		(2.0, 1.0)
	};

	\addlegendentry[no markers, teal]{brute}
	\addlegendentry[no markers, red]{$\alpha = 1.0$, $c = 1.363$}
	\addlegendentry[no markers, blue]{$\alpha = 2.0$, $c = 1.0$}

	\end{axis}
\end{tikzpicture}

%% file: plots/figure_vc_comp.tex
\begin{tikzpicture}[scale = 0.9]
	\begin{axis}[xmin = 1, xmax = 2, ymax = 1.29, xlabel = {approximation ratio}]

	\addplot[red , thick] coordinates {
		(1.0, 1.2663242847)
		(1.01, 1.2367469348)
		(1.02, 1.2200190205)
		(1.03, 1.2073852094)
		(1.04, 1.1970937089)
		(1.05, 1.1883673416)
		(1.06, 1.1807772732)
		(1.07, 1.1740569098)
		(1.08, 1.1680272319)
		(1.09, 1.1625612928)
		(1.1, 1.1575653247)
		(1.11, 1.1529678327)
		(1.12, 1.1487128983)
		(1.13, 1.1447558616)
		(1.14, 1.1410604257)
		(1.15, 1.1375966491)
		(1.16, 1.1343395157)
		(1.17, 1.131267891)
		(1.18, 1.128363745)
		(1.19, 1.1256115629)
		(1.2, 1.1229978911)
		(1.21, 1.1205109827)
		(1.22, 1.1181405165)
		(1.23, 1.1158773732)
		(1.24, 1.1137134534)
		(1.25, 1.1116415299)
		(1.26, 1.1096551254)
		(1.27, 1.1077484119)
		(1.28, 1.1059161259)
		(1.29, 1.1041534977)
		(1.3, 1.1024561909)
		(1.31, 1.1008202517)
		(1.32, 1.0992420649)
		(1.33, 1.0977183163)
		(1.34, 1.0962459602)
		(1.35, 1.0948221912)
		(1.36, 1.0934444197)
		(1.37, 1.0921102501)
		(1.38, 1.0908174625)
		(1.39, 1.0895639955)
		(1.4, 1.0883479319)
		(1.41, 1.0871674852)
		(1.42, 1.0860209888)
		(1.43, 1.0849068849)
		(1.44, 1.0838237156)
		(1.45, 1.0827701148)
		(1.46, 1.0817448004)
		(1.47, 1.0807465679)
		(1.48, 1.079774284)
		(1.49, 1.0788268817)
		(1.5, 1.0779033547)
		(1.51, 1.0770027536)
		(1.52, 1.076124181)
		(1.53, 1.0752667885)
		(1.54, 1.074429773)
		(1.55, 1.0736123734)
		(1.56, 1.0728138682)
		(1.57, 1.0720335723)
		(1.58, 1.0712708352)
		(1.59, 1.0705250383)
		(1.6, 1.0697955934)
		(1.61, 1.06908194)
		(1.62, 1.0683835446)
		(1.63, 1.0676998982)
		(1.64, 1.0670305153)
		(1.65, 1.0663749323)
		(1.66, 1.0657327066)
		(1.67, 1.0651034148)
		(1.68, 1.064486652)
		(1.69, 1.063882031)
		(1.7, 1.0632891806)
		(1.71, 1.0627077456)
		(1.72, 1.0621373851)
		(1.73, 1.0615777724)
		(1.74, 1.061028594)
		(1.75, 1.0604895489)
		(1.76, 1.0599603478)
		(1.77, 1.059440713)
		(1.78, 1.0589303774)
		(1.79, 1.0584290839)
		(1.8, 1.0579365855)
		(1.81, 1.057452644)
		(1.82, 1.0569770304)
		(1.83, 1.0565095238)
		(1.84, 1.0560499114)
		(1.85, 1.0555979879)
		(1.86, 1.0551535556)
		(1.87, 1.0547164233)
		(1.88, 1.0542864067)
		(1.89, 1.0538633278)
		(1.9, 1.0534470148)
		(1.91, 1.0530373014)
		(1.92, 1.052634027)
		(1.93, 1.0522370364)
		(1.94, 1.0518461795)
		(1.95, 1.0514613108)
		(1.96, 1.0510822898)
		(1.97, 1.0507089804)
		(1.98, 1.0503412507)
		(1.99, 1.049978973)
		(2.0, 1.0496220236)
	};

	\addplot[blue , thick] coordinates {
		(1.0, 2.0)
		(1.01, 1.9387047513)
		(1.02, 1.8930929239)
		(1.03, 1.853849694)
		(1.04, 1.8188398378)
		(1.05, 1.7870069014)
		(1.06, 1.7577092525)
		(1.07, 1.7305126)
		(1.08, 1.705102307)
		(1.09, 1.6812394878)
		(1.1, 1.6587364406)
		(1.11, 1.6374417619)
		(1.12, 1.6172307723)
		(1.13, 1.5979990628)
		(1.14, 1.5796579792)
		(1.15, 1.5621313614)
		(1.16, 1.545353129)
		(1.17, 1.5292654531)
		(1.18, 1.5138173453)
		(1.19, 1.498963551)
		(1.2, 1.484663669)
		(1.21, 1.4708814415)
		(1.22, 1.457584176)
		(1.23, 1.4447422687)
		(1.24, 1.4323288084)
		(1.25, 1.4203192453)
		(1.26, 1.4086911107)
		(1.27, 1.3974237788)
		(1.28, 1.3864982635)
		(1.29, 1.3758970425)
		(1.3, 1.3656039065)
		(1.31, 1.3556038269)
		(1.32, 1.3458828408)
		(1.33, 1.3364279501)
		(1.34, 1.3272270326)
		(1.35, 1.3182687633)
		(1.36, 1.3095425448)
		(1.37, 1.3010384449)
		(1.38, 1.2927471418)
		(1.39, 1.2846598738)
		(1.4, 1.2767683951)
		(1.41, 1.2690649358)
		(1.42, 1.2615421652)
		(1.43, 1.2541931594)
		(1.44, 1.2470113718)
		(1.45, 1.2399906054)
		(1.46, 1.233124989)
		(1.47, 1.2264089541)
		(1.48, 1.2198372148)
		(1.49, 1.2134047491)
		(1.5, 1.2071067812)
		(1.51, 1.2009387662)
		(1.52, 1.1948963754)
		(1.53, 1.1889754825)
		(1.54, 1.1831721518)
		(1.55, 1.1774826264)
		(1.56, 1.1719033178)
		(1.57, 1.1664307958)
		(1.58, 1.1610617798)
		(1.59, 1.1557931299)
		(1.6, 1.1506218396)
		(1.61, 1.1455450278)
		(1.62, 1.1405599327)
		(1.63, 1.1356639047)
		(1.64, 1.1308544012)
		(1.65, 1.1261289803)
		(1.66, 1.1214852963)
		(1.67, 1.1169210943)
		(1.68, 1.112434206)
		(1.69, 1.1080225451)
		(1.7, 1.1036841036)
		(1.71, 1.0994169478)
		(1.72, 1.0952192149)
		(1.73, 1.0910891093)
		(1.74, 1.0870249001)
		(1.75, 1.0830249175)
		(1.76, 1.0790875504)
		(1.77, 1.0752112435)
		(1.78, 1.0713944949)
		(1.79, 1.067635854)
		(1.8, 1.0639339188)
		(1.81, 1.060287334)
		(1.82, 1.0566947891)
		(1.83, 1.0531550165)
		(1.84, 1.0496667894)
		(1.85, 1.0462289206)
		(1.86, 1.0428402603)
		(1.87, 1.0394996953)
		(1.88, 1.0362061466)
		(1.89, 1.0329585688)
		(1.9, 1.0297559486)
		(1.91, 1.0265973033)
		(1.92, 1.0234816797)
		(1.93, 1.0204081532)
		(1.94, 1.0173758263)
		(1.95, 1.0143838281)
		(1.96, 1.0114313128)
		(1.97, 1.008517459)
		(1.98, 1.0056414688)
		(1.99, 1.002802567)
		(2.0, 1.0)
	};

	\addplot[orange, thick] coordinates {
		(1.0, 1.2149474014758987)
		(1.01, 1.18906187678848620217)
		(1.02, 1.17517639819838795677)
		(1.03, 1.16486283813320110251)
		(1.04, 1.1565359199017096358)
		(1.05, 1.14951427038467875476)
		(1.06, 1.14324496444728092378)
		(1.07, 1.13578717169771894528)
		(1.08, 1.12748551869511049934)
		(1.09, 1.11966731422367016877)
		(1.1, 1.11244222475317542976)
		(1.11, 1.10573867471306007846)
		(1.12, 1.09949699287917953601)
		(1.13, 1.09366733579764810475)
		(1.14, 1.08820779437497395815)
		(1.15, 1.08308283885290275186)
		(1.16, 1.07826292660034489221)
		(1.17, 1.0737229674451551924)
		(1.18, 1.06944051659504461074)
		(1.19, 1.06545889271475507657)
		(1.2, 1.06181887032363014463)
		(1.21, 1.05856413868863626872)
		(1.22, 1.05548631178338285524)
		(1.23, 1.05257228057004263772)
		(1.24, 1.0498103538923178446)
		(1.25, 1.04719005717206098577)
		(1.26, 1.04470159876920504277)
		(1.27, 1.04233584006755384937)
		(1.28, 1.04008501345072896235)
		(1.29, 1.03794208237903507086)
		(1.3, 1.03590064702952063888)
		(1.31, 1.0339548701230821537)
		(1.32, 1.03209941322778037936)
		(1.33, 1.03032938181569576814)
		(1.34, 1.0286402776709222974)
		(1.35, 1.02702795750076865868)
		(1.36, 1.02548859680556225758)
		(1.37, 1.02401865822581310039)
		(1.38, 1.02261538996744028955)
		(1.39, 1.02127636340112225955)
		(1.4, 1.0199986773370962229)
		(1.41, 1.0187796202793092292)
		(1.42, 1.01761665613497695446)
		(1.43, 1.01650740954673617692)
		(1.44, 1.01544965271449579417)
		(1.45, 1.01444129352596993583)
		(1.46, 1.0134803648400620398)
		(1.47, 1.01255959115197166637)
		(1.48, 1.01160585623202152032)
		(1.49, 1.01066283998030439129)
		(1.5, 1.0097938316435373252)
		(1.51, 1.00903332843582603709)
		(1.52, 1.00839681540847415694)
		(1.53, 1.00779056110135301423)
		(1.54, 1.00721353930610744252)
		(1.55, 1.0066555108246231531)
		(1.56, 1.00609173249576134161)
		(1.57, 1.00554949750961886022)
		(1.58, 1.0050711537857322363)
		(1.59, 1.00466701630189243981)
		(1.6, 1.00428364654988682732)
		(1.61, 1.00391190739587411862)
		(1.62, 1.0035432478302732036)
		(1.63, 1.00320813046511462222)
		(1.64, 1.00291597122675515817)
		(1.65, 1.0026484362259953341)
		(1.66, 1.00238597220440289203)
		(1.67, 1.00214013288323928125)
		(1.68, 1.00192467940422192377)
		(1.69, 1.00172653534298176978)
		(1.7, 1.00153542438835011632)
		(1.71, 1.00136717234851301925)
		(1.72, 1.00121402166926752089)
		(1.73, 1.00107025744788320591)
		(1.74, 1.00094322547674637004)
		(1.75, 1.00082410030373020685)
		(1.76, 1.00071959248398631429)
		(1.77, 1.00062307338813879387)
		(1.78, 1.00053795571309941829)
		(1.79, 1.0004610572608726894)
		(1.8, 1.0003918475388674973)
		(1.81, 1.00033138058719884933)
		(1.82, 1.00027792575608432877)
		(1.83, 1.00023080171600007264)
		(1.84, 1.00018960278121126699)
		(1.85, 1.0001538732201917791)
		(1.86, 1.00012346738960317224)
		(1.87, 1.0000974004279254247)
		(1.88, 1.00007552056638967602)
		(1.89, 1.00005738980392362766)
		(1.9, 1.00004250380303013558)
		(1.91, 1.00003056963139673563)
		(1.92, 1.00002117666933980112)
		(1.93, 1.00001398874524700761)
		(1.94, 1.00000868828727620586)
		(1.95, 1.00000495969237464556)
		(1.96, 1.00000250499070630283)
		(1.97, 1.00000104251505773278)
		(1.98, 1.00000030471899043561)
		(1.99, 1.00000003757135259284)
		(2.0, 1.0)
	};

	\addlegendentry[no markers, red]{$\alpha = 1.0$, $c = 1.363$}
	\addlegendentry[no markers, blue]{$\alpha = 2.0$, $c = 1.0$}
	\addlegendentry[no markers, orange]{\textsc{Unweighted VC}}

	\end{axis}
\end{tikzpicture}

%% file: plots/table_vc_extended.tex
\begin{tabular}{c|c|c|c|c|c|c|c|c|c|}
	 & $1.1$ & $1.2$ & $1.3$ & $1.4$ & $1.5$ & $1.6$ & $1.7$ & $1.8$ & $1.9$\\
	\hline
	$\brute$ & $1.716$ & $1.583$ & $1.496$ & $1.433$ & $1.385$ & $1.347$ & $1.317$ & $1.291$ & $1.269$\\
	\hline
	$(\alpha = 1, c = 1.363)$ & $1.158$ & $1.123$ & $1.103$ & $1.089$ & $1.078$ & $1.07$ & $1.064$ & $1.058$ & $1.054$\\
	\hline
	$(\alpha = 2,c = 1)$ & $1.659$ & $1.485$ & $1.366$ & $1.277$ & $1.208$ & $1.151$ & $1.104$ & $1.064$ & $1.03$\\
	\hline
	\textsc{Unweighted VC} & $1.113$ & $1.062$ & $1.036$ & $1.02$ & $1.01$ & $1.005$ & $1.002$ & $1.0004$ & $1.00005$\\
	\hline
\end{tabular}

%% file: plots/figure_fvs.tex
\begin{tikzpicture}[scale = 0.9]
	\begin{axis}[xmin = 1, xmax = 2, ymin = 0.9, ymax = 2.1, xlabel = {approximation ratio}]

	\addplot[teal, thick] coordinates {
		(1.0, 2.0)
		(1.01, 1.9454431345)
		(1.02, 1.9062508454)
		(1.03, 1.8731549013)
		(1.04, 1.8440491304)
		(1.05, 1.8178991111)
		(1.06, 1.794081097)
		(1.07, 1.7721758604)
		(1.08, 1.7518817491)
		(1.09, 1.7329712548)
		(1.1, 1.7152667656)
		(1.11, 1.698625911)
		(1.12, 1.6829321593)
		(1.13, 1.6680884977)
		(1.14, 1.6540130203)
		(1.15, 1.6406357531)
		(1.16, 1.6278963084)
		(1.17, 1.615742114)
		(1.18, 1.6041270506)
		(1.19, 1.5930103866)
		(1.2, 1.5823559323)
		(1.21, 1.5721313608)
		(1.22, 1.5623076557)
		(1.23, 1.552858658)
		(1.24, 1.5437606905)
		(1.25, 1.534992244)
		(1.26, 1.5265337131)
		(1.27, 1.5183671728)
		(1.28, 1.510476187)
		(1.29, 1.5028456449)
		(1.3, 1.4954616201)
		(1.31, 1.4883112476)
		(1.32, 1.4813826173)
		(1.33, 1.4746646809)
		(1.34, 1.4681471696)
		(1.35, 1.4618205222)
		(1.36, 1.4556758212)
		(1.37, 1.4497047356)
		(1.38, 1.443899471)
		(1.39, 1.4382527242)
		(1.4, 1.4327576428)
		(1.41, 1.427407789)
		(1.42, 1.4221971069)
		(1.43, 1.4171198929)
		(1.44, 1.4121707695)
		(1.45, 1.4073446605)
		(1.46, 1.4026367697)
		(1.47, 1.3980425604)
		(1.48, 1.3935577374)
		(1.49, 1.3891782307)
		(1.5, 1.3849001795)
		(1.51, 1.380719919)
		(1.52, 1.3766339674)
		(1.53, 1.3726390137)
		(1.54, 1.3687319076)
		(1.55, 1.3649096489)
		(1.56, 1.3611693784)
		(1.57, 1.3575083696)
		(1.58, 1.3539240206)
		(1.59, 1.3504138468)
		(1.6, 1.3469754742)
		(1.61, 1.343606633)
		(1.62, 1.3403051519)
		(1.63, 1.3370689522)
		(1.64, 1.3338960432)
		(1.65, 1.3307845174)
		(1.66, 1.3277325456)
		(1.67, 1.3247383734)
		(1.68, 1.3218003168)
		(1.69, 1.3189167589)
		(1.7, 1.3160861463)
		(1.71, 1.3133069861)
		(1.72, 1.3105778425)
		(1.73, 1.3078973347)
		(1.74, 1.3052641335)
		(1.75, 1.3026769593)
		(1.76, 1.3001345795)
		(1.77, 1.2976358065)
		(1.78, 1.2951794954)
		(1.79, 1.2927645423)
		(1.8, 1.2903898821)
		(1.81, 1.2880544871)
		(1.82, 1.2857573652)
		(1.83, 1.2834975581)
		(1.84, 1.2812741403)
		(1.85, 1.2790862173)
		(1.86, 1.2769329244)
		(1.87, 1.2748134255)
		(1.88, 1.2727269118)
		(1.89, 1.2706726008)
		(1.9, 1.2686497351)
		(1.91, 1.2666575813)
		(1.92, 1.2646954293)
		(1.93, 1.2627625911)
		(1.94, 1.2608584001)
		(1.95, 1.2589822102)
		(1.96, 1.2571333949)
		(1.97, 1.2553113469)
		(1.98, 1.2535154767)
		(1.99, 1.2517452127)
		(2.0, 1.25)
	};

	\addplot[ red , thick] coordinates {
		(1.0, 1.7236042012)
		(1.01, 1.6752823032)
		(1.02, 1.6419477534)
		(1.03, 1.6143949715)
		(1.04, 1.5905556226)
		(1.05, 1.569423424)
		(1.06, 1.5503971586)
		(1.07, 1.5330763699)
		(1.08, 1.5171754461)
		(1.09, 1.5024808262)
		(1.1, 1.488827183)
		(1.11, 1.4760830734)
		(1.12, 1.4641417633)
		(1.13, 1.452915081)
		(1.14, 1.4423291422)
		(1.15, 1.4323212842)
		(1.16, 1.4228378102)
		(1.17, 1.4138322926)
		(1.18, 1.405264273)
		(1.19, 1.3970982487)
		(1.2, 1.3893028732)
		(1.21, 1.3818503149)
		(1.22, 1.3747157396)
		(1.23, 1.3678768857)
		(1.24, 1.3613137153)
		(1.25, 1.355008123)
		(1.26, 1.3489436916)
		(1.27, 1.3431054873)
		(1.28, 1.3374798837)
		(1.29, 1.3320544132)
		(1.3, 1.3268176384)
		(1.31, 1.3217590414)
		(1.32, 1.3168689275)
		(1.33, 1.3121383418)
		(1.34, 1.3075589959)
		(1.35, 1.3031232033)
		(1.36, 1.2988238232)
		(1.37, 1.2946542099)
		(1.38, 1.2906081691)
		(1.39, 1.2866799173)
		(1.4, 1.2828640477)
		(1.41, 1.2791554977)
		(1.42, 1.275549521)
		(1.43, 1.2720416623)
		(1.44, 1.2686277339)
		(1.45, 1.2653037953)
		(1.46, 1.2620661346)
		(1.47, 1.2589112511)
		(1.48, 1.2558358398)
		(1.49, 1.252836778)
		(1.5, 1.2499111117)
		(1.51, 1.2470560442)
		(1.52, 1.2442689253)
		(1.53, 1.2415472418)
		(1.54, 1.2388886077)
		(1.55, 1.2362907569)
		(1.56, 1.2337515345)
		(1.57, 1.2312688908)
		(1.58, 1.2288408739)
		(1.59, 1.2264656244)
		(1.6, 1.2241413691)
		(1.61, 1.2218664168)
		(1.62, 1.2196391525)
		(1.63, 1.2174580339)
		(1.64, 1.2153215864)
		(1.65, 1.2132284001)
		(1.66, 1.2111771256)
		(1.67, 1.2091664708)
		(1.68, 1.2071951982)
		(1.69, 1.2052621214)
		(1.7, 1.2033661029)
		(1.71, 1.2015060511)
		(1.72, 1.1996809184)
		(1.73, 1.1978896984)
		(1.74, 1.1961314244)
		(1.75, 1.1944051671)
		(1.76, 1.1927100328)
		(1.77, 1.1910451617)
		(1.78, 1.1894097262)
		(1.79, 1.1878029295)
		(1.8, 1.186224004)
		(1.81, 1.1846722102)
		(1.82, 1.1831468348)
		(1.83, 1.1816471904)
		(1.84, 1.1801726137)
		(1.85, 1.1787224643)
		(1.86, 1.1772961242)
		(1.87, 1.1758929967)
		(1.88, 1.1745125051)
		(1.89, 1.1731540922)
		(1.9, 1.1718172194)
		(1.91, 1.1705013657)
		(1.92, 1.1692060275)
		(1.93, 1.167930717)
		(1.94, 1.1666749626)
		(1.95, 1.1654383075)
		(1.96, 1.1642203092)
		(1.97, 1.1630205391)
		(1.98, 1.1618385821)
		(1.99, 1.1606740357)
		(2.0, 1.1595265097)
	};

	\addplot[ blue , thick] coordinates {
		(1.0, 2.0)
		(1.01, 1.9387047513)
		(1.02, 1.8930929239)
		(1.03, 1.853849694)
		(1.04, 1.8188398378)
		(1.05, 1.7870069014)
		(1.06, 1.7577092525)
		(1.07, 1.7305126)
		(1.08, 1.705102307)
		(1.09, 1.6812394878)
		(1.1, 1.6587364406)
		(1.11, 1.6374417619)
		(1.12, 1.6172307723)
		(1.13, 1.5979990628)
		(1.14, 1.5796579792)
		(1.15, 1.5621313614)
		(1.16, 1.545353129)
		(1.17, 1.5292654531)
		(1.18, 1.5138173453)
		(1.19, 1.498963551)
		(1.2, 1.484663669)
		(1.21, 1.4708814415)
		(1.22, 1.457584176)
		(1.23, 1.4447422687)
		(1.24, 1.4323288084)
		(1.25, 1.4203192453)
		(1.26, 1.4086911107)
		(1.27, 1.3974237788)
		(1.28, 1.3864982635)
		(1.29, 1.3758970425)
		(1.3, 1.3656039065)
		(1.31, 1.3556038269)
		(1.32, 1.3458828408)
		(1.33, 1.3364279501)
		(1.34, 1.3272270326)
		(1.35, 1.3182687633)
		(1.36, 1.3095425448)
		(1.37, 1.3010384449)
		(1.38, 1.2927471418)
		(1.39, 1.2846598738)
		(1.4, 1.2767683951)
		(1.41, 1.2690649358)
		(1.42, 1.2615421652)
		(1.43, 1.2541931594)
		(1.44, 1.2470113718)
		(1.45, 1.2399906054)
		(1.46, 1.233124989)
		(1.47, 1.2264089541)
		(1.48, 1.2198372148)
		(1.49, 1.2134047491)
		(1.5, 1.2071067812)
		(1.51, 1.2009387662)
		(1.52, 1.1948963754)
		(1.53, 1.1889754825)
		(1.54, 1.1831721518)
		(1.55, 1.1774826264)
		(1.56, 1.1719033178)
		(1.57, 1.1664307958)
		(1.58, 1.1610617798)
		(1.59, 1.1557931299)
		(1.6, 1.1506218396)
		(1.61, 1.1455450278)
		(1.62, 1.1405599327)
		(1.63, 1.1356639047)
		(1.64, 1.1308544012)
		(1.65, 1.1261289803)
		(1.66, 1.1214852963)
		(1.67, 1.1169210943)
		(1.68, 1.112434206)
		(1.69, 1.1080225451)
		(1.7, 1.1036841036)
		(1.71, 1.0994169478)
		(1.72, 1.0952192149)
		(1.73, 1.0910891093)
		(1.74, 1.0870249001)
		(1.75, 1.0830249175)
		(1.76, 1.0790875504)
		(1.77, 1.0752112435)
		(1.78, 1.0713944949)
		(1.79, 1.067635854)
		(1.8, 1.0639339188)
		(1.81, 1.060287334)
		(1.82, 1.0566947891)
		(1.83, 1.0531550165)
		(1.84, 1.0496667894)
		(1.85, 1.0462289206)
		(1.86, 1.0428402603)
		(1.87, 1.0394996953)
		(1.88, 1.0362061466)
		(1.89, 1.0329585688)
		(1.9, 1.0297559486)
		(1.91, 1.0265973033)
		(1.92, 1.0234816797)
		(1.93, 1.0204081532)
		(1.94, 1.0173758263)
		(1.95, 1.0143838281)
		(1.96, 1.0114313128)
		(1.97, 1.008517459)
		(1.98, 1.0056414688)
		(1.99, 1.002802567)
		(2.0, 1.0)
	};

	\addlegendentry[no markers, teal]{brute}
	\addlegendentry[no markers, red]{$\alpha = 1.0$, $c = 3.618$}
	\addlegendentry[no markers, blue]{$\alpha = 2.0$, $c = 1.0$}

	\end{axis}
\end{tikzpicture}

%% file: plots/figure_3hs.tex
\begin{tikzpicture}[scale = 0.9]
	\begin{axis}[xmin = 1, xmax = 3, ymin = 0.9, ymax = 2.1, xlabel = {approximation ratio}]

	\addplot[teal, thick] coordinates {
		(1.0, 2.0)
		(1.02, 1.9062508454)
		(1.04, 1.8440491304)
		(1.06, 1.794081097)
		(1.08, 1.7518817491)
		(1.1, 1.7152667656)
		(1.12, 1.6829321593)
		(1.14, 1.6540130203)
		(1.16, 1.6278963084)
		(1.18, 1.6041270506)
		(1.2, 1.5823559323)
		(1.22, 1.5623076557)
		(1.24, 1.5437606905)
		(1.26, 1.5265337131)
		(1.28, 1.510476187)
		(1.3, 1.4954616201)
		(1.32, 1.4813826173)
		(1.34, 1.4681471696)
		(1.36, 1.4556758212)
		(1.38, 1.443899471)
		(1.4, 1.4327576428)
		(1.42, 1.4221971069)
		(1.44, 1.4121707695)
		(1.46, 1.4026367697)
		(1.48, 1.3935577374)
		(1.5, 1.3849001795)
		(1.52, 1.3766339674)
		(1.54, 1.3687319076)
		(1.56, 1.3611693784)
		(1.58, 1.3539240206)
		(1.6, 1.3469754742)
		(1.62, 1.3403051519)
		(1.64, 1.3338960432)
		(1.66, 1.3277325456)
		(1.68, 1.3218003168)
		(1.7, 1.3160861463)
		(1.72, 1.3105778425)
		(1.74, 1.3052641335)
		(1.76, 1.3001345795)
		(1.78, 1.2951794954)
		(1.8, 1.2903898821)
		(1.82, 1.2857573652)
		(1.84, 1.2812741403)
		(1.86, 1.2769329244)
		(1.88, 1.2727269118)
		(1.9, 1.2686497351)
		(1.92, 1.2646954293)
		(1.94, 1.2608584001)
		(1.96, 1.2571333949)
		(1.98, 1.2535154767)
		(2.0, 1.25)
		(2.02, 1.2465825894)
		(2.04, 1.24325912)
		(2.06, 1.2400256991)
		(2.08, 1.23687865)
		(2.1, 1.2338144969)
		(2.12, 1.2308299514)
		(2.14, 1.2279218994)
		(2.16, 1.2250873898)
		(2.18, 1.2223236237)
		(2.2, 1.2196279447)
		(2.22, 1.2169978296)
		(2.24, 1.2144308803)
		(2.26, 1.2119248158)
		(2.28, 1.2094774652)
		(2.3, 1.2070867609)
		(2.32, 1.2047507326)
		(2.34, 1.2024675014)
		(2.36, 1.2002352749)
		(2.38, 1.1980523416)
		(2.4, 1.1959170667)
		(2.42, 1.1938278879)
		(2.44, 1.1917833111)
		(2.46, 1.1897819069)
		(2.48, 1.1878223069)
		(2.5, 1.1859032006)
		(2.52, 1.1840233324)
		(2.54, 1.1821814985)
		(2.56, 1.1803765443)
		(2.58, 1.178607362)
		(2.6, 1.1768728882)
		(2.62, 1.1751721016)
		(2.64, 1.1735040209)
		(2.66, 1.171867703)
		(2.68, 1.170262241)
		(2.7, 1.1686867625)
		(2.72, 1.1671404279)
		(2.74, 1.1656224291)
		(2.76, 1.1641319876)
		(2.78, 1.1626683537)
		(2.8, 1.1612308047)
		(2.82, 1.1598186437)
		(2.84, 1.1584311989)
		(2.86, 1.157067822)
		(2.88, 1.1557278873)
		(2.9, 1.1544107911)
		(2.92, 1.1531159499)
		(2.94, 1.1518428004)
		(2.96, 1.1505907982)
		(2.98, 1.149359417)
		(3.0, 1.1481481481)
	};

	\addplot[ red , thick] coordinates {
		(1.0, 1.5387453875)
		(1.02, 1.4680176601)
		(1.04, 1.426140626)
		(1.06, 1.3943823623)
		(1.08, 1.3686503508)
		(1.1, 1.3470339938)
		(1.12, 1.3284407274)
		(1.14, 1.3121738284)
		(1.16, 1.2977568579)
		(1.18, 1.2848475816)
		(1.2, 1.2731909475)
		(1.22, 1.2625913358)
		(1.24, 1.2528951993)
		(1.26, 1.2439796939)
		(1.28, 1.2357449507)
		(1.3, 1.2281086608)
		(1.32, 1.2210021802)
		(1.34, 1.2143676657)
		(1.36, 1.2081559285)
		(1.38, 1.2023247984)
		(1.4, 1.1968378582)
		(1.42, 1.1916634517)
		(1.44, 1.1867738974)
		(1.46, 1.1821448572)
		(1.48, 1.177754827)
		(1.5, 1.1735847187)
		(1.52, 1.1696175184)
		(1.54, 1.1658380008)
		(1.56, 1.1622324923)
		(1.58, 1.1587886714)
		(1.6, 1.1554953994)
		(1.62, 1.1523425779)
		(1.64, 1.1493210257)
		(1.66, 1.1464223737)
		(1.68, 1.1436389749)
		(1.7, 1.1409638251)
		(1.72, 1.1383904956)
		(1.74, 1.1359130728)
		(1.76, 1.1335261069)
		(1.78, 1.1312245659)
		(1.8, 1.1290037951)
		(1.82, 1.1268594816)
		(1.84, 1.1247876228)
		(1.86, 1.1227844982)
		(1.88, 1.1208466445)
		(1.9, 1.1189708333)
		(1.92, 1.1171540506)
		(1.94, 1.1153934797)
		(1.96, 1.1136864843)
		(1.98, 1.1120305942)
		(2.0, 1.1104234923)
		(2.02, 1.1088630026)
		(2.04, 1.1073470792)
		(2.06, 1.1058737968)
		(2.08, 1.1044413417)
		(2.1, 1.1030480037)
		(2.12, 1.1016921686)
		(2.14, 1.1003723116)
		(2.16, 1.0990869908)
		(2.18, 1.0978348419)
		(2.2, 1.0966145727)
		(2.22, 1.0954249586)
		(2.24, 1.0942648377)
		(2.26, 1.0931331074)
		(2.28, 1.09202872)
		(2.3, 1.0909506796)
		(2.32, 1.0898980391)
		(2.34, 1.0888698965)
		(2.36, 1.0878653931)
		(2.38, 1.0868837101)
		(2.4, 1.0859240669)
		(2.42, 1.0849857184)
		(2.44, 1.0840679533)
		(2.46, 1.083170092)
		(2.48, 1.0822914849)
		(2.5, 1.0814315107)
		(2.52, 1.080589575)
		(2.54, 1.0797651086)
		(2.56, 1.0789575663)
		(2.58, 1.0781664259)
		(2.6, 1.0773911863)
		(2.62, 1.076631367)
		(2.64, 1.0758865072)
		(2.66, 1.0751561639)
		(2.68, 1.0744399122)
		(2.7, 1.0737373434)
		(2.72, 1.0730480647)
		(2.74, 1.0723716986)
		(2.76, 1.0717078815)
		(2.78, 1.0710562639)
		(2.8, 1.0704165089)
		(2.82, 1.0697882922)
		(2.84, 1.0691713014)
		(2.86, 1.0685652352)
		(2.88, 1.0679698033)
		(2.9, 1.0673847254)
		(2.92, 1.0668097314)
		(2.94, 1.0662445603)
		(2.96, 1.0656889603)
		(2.98, 1.0651426881)
		(3.0, 1.0646055087)
	};

	\addplot[ blue , thick] coordinates {
		(1.0, 2.0)
		(1.02, 1.9058941662)
		(1.04, 1.842764221)
		(1.06, 1.7914558434)
		(1.08, 1.7476145144)
		(1.1, 1.7091350108)
		(1.12, 1.6747714482)
		(1.14, 1.6437024677)
		(1.16, 1.6153481991)
		(1.18, 1.5892792572)
		(1.2, 1.5651662684)
		(1.22, 1.5427496004)
		(1.24, 1.521820106)
		(1.26, 1.5022062939)
		(1.28, 1.483765458)
		(1.3, 1.4663773513)
		(1.32, 1.4499395578)
		(1.34, 1.4343640312)
		(1.36, 1.4195744567)
		(1.38, 1.405504207)
		(1.4, 1.3920947355)
		(1.42, 1.3792942966)
		(1.44, 1.3670569155)
		(1.46, 1.3553415491)
		(1.48, 1.3441113975)
		(1.5, 1.3333333333)
		(1.52, 1.3229774253)
		(1.54, 1.3130165381)
		(1.56, 1.3034259933)
		(1.58, 1.2941832811)
		(1.6, 1.2852678138)
		(1.62, 1.276660713)
		(1.64, 1.2683446265)
		(1.66, 1.2603035687)
		(1.68, 1.2525227816)
		(1.7, 1.2449886134)
		(1.72, 1.2376884114)
		(1.74, 1.2306104278)
		(1.76, 1.2237437366)
		(1.78, 1.2170781592)
		(1.8, 1.2106041992)
		(1.82, 1.2043129832)
		(1.84, 1.1981962085)
		(1.86, 1.1922460957)
		(1.88, 1.186455347)
		(1.9, 1.180817107)
		(1.92, 1.1753249293)
		(1.94, 1.1699727446)
		(1.96, 1.1647548325)
		(1.98, 1.1596657961)
		(2.0, 1.1547005384)
		(2.02, 1.1498542408)
		(2.04, 1.145122344)
		(2.06, 1.1405005299)
		(2.08, 1.1359847055)
		(2.1, 1.1315709875)
		(2.12, 1.1272556893)
		(2.14, 1.1230353073)
		(2.16, 1.1189065102)
		(2.18, 1.1148661277)
		(2.2, 1.1109111406)
		(2.22, 1.1070386718)
		(2.24, 1.1032459776)
		(2.26, 1.0995304399)
		(2.28, 1.0958895588)
		(2.3, 1.0923209458)
		(2.32, 1.0888223176)
		(2.34, 1.08539149)
		(2.36, 1.0820263727)
		(2.38, 1.0787249639)
		(2.4, 1.0754853454)
		(2.42, 1.0723056786)
		(2.44, 1.0691842)
		(2.46, 1.0661192173)
		(2.48, 1.063109106)
		(2.5, 1.0601523054)
		(2.52, 1.0572473159)
		(2.54, 1.054392696)
		(2.56, 1.0515870589)
		(2.58, 1.0488290705)
		(2.6, 1.0461174463)
		(2.62, 1.0434509495)
		(2.64, 1.0408283884)
		(2.66, 1.0382486146)
		(2.68, 1.0357105208)
		(2.7, 1.0332130391)
		(2.72, 1.0307551389)
		(2.74, 1.0283358256)
		(2.76, 1.0259541389)
		(2.78, 1.0236091512)
		(2.8, 1.0212999661)
		(2.82, 1.0190257175)
		(2.84, 1.0167855678)
		(2.86, 1.014578707)
		(2.88, 1.0124043513)
		(2.9, 1.0102617425)
		(2.92, 1.0081501464)
		(2.94, 1.006068852)
		(2.96, 1.0040171708)
		(2.98, 1.0019944357)
		(3.0, 1.0)
	};

	\addlegendentry[no markers, teal]{brute}
	\addlegendentry[no markers, red]{$\alpha = 1.0$, $c = 2.168$}
	\addlegendentry[no markers, blue]{$\alpha = 3.0$, $c = 1.0$}

	\end{axis}
\end{tikzpicture}

%% file: membership_oracle.tex
In this section we prove \Cref{thm:brute_force}.
The algorithm is based on the notion of {\em covering families}.
\begin{definition}[Covering Family]
	Let $U$ be a finite set,  $\w\colon U\rightarrow \NN$ be a weight function and $\alpha\geq 1$.
	We say $\CC\subseteq 2^U$ is an {\em $\alpha$-covering family} of $U$ and $\w$ if for every $S\subseteq U$ there exists $T\in \CC$ such that $S\subseteq T$ and $\w(T)\leq \alpha \cdot \w(S)$.
\end{definition}
An $\alpha$-covering family $\CC$ can be easily used to attain an $\alpha$-approximation algorithm in the membership model as follows. 
The algorithm constructs the covering family $\CC$ and uses the membership oracle to compute $\CC\cap \CF$ using $\abs{\CC}$ queries. 
The algorithm then returns the set $Q\in \CC\cap \CF$ of minimum weight. To show correctness, consider a set $S\in \CF$ such that $\w(S)= \opt(U,\w,\CF)$.
Since $\CC$ is an $\alpha$-covering family there is $T\in \CC$ such that $S\subseteq T$ and $\w(T)\leq \alpha \cdot \w(S)= \alpha \cdot \opt(U,\w,\CF)$.
Since $\CF$ is monotone it holds that $T\in \CF$, thus $T\in \CC\cap \CF$, and we conclude that the algorithm returns a set of weight at most $\alpha \cdot \opt(U,\w,\CF)$. The running time of the algorithm, up to polynomial factors, is the construction time of $\CC$ plus~$\abs{\CC}$.

By the above argument, the proof of \Cref{thm:brute_force} boils down to the construction of  $\alpha$-covering families. To this end, we show the next result.
\begin{lemma}
	\label{lem:covering_family} 
	There exists an algorithm which given a finite set $U$, a weight function $\w:U\rightarrow \NN$ and $\alpha > 1$, constructs an $\alpha$-covering family $\CC$ of $U$ and $\w$ such that $\abs{\CC} \leq \left( \brute(\alpha)\right)^{n+o(n)}$. Furthermore, the running time of the algorithm is $\left( \brute(\alpha)\right)^{n+o(n)}$.
\end{lemma}

That is, \Cref{thm:brute_force} is an immediate consequence of \Cref{lem:covering_family}.
The construction of the covering family in \Cref{lem:covering_family} is based on a rounding of the weight function and a  reduction to covering families in the unweighted case. The construction of such families was implicitly given in \cite{EsmerKMNS22}.

\begin{lemma}[\cite{EsmerKMNS22}]
	\label{lem:covering_implicit}
	There exists an algorithm which given a finite set $U$ and $\alpha > 1$ returns an $\alpha$-covering family $\CC$ of $U$ and the uniform weight function $\w:U\rightarrow \{1\} $ such that $\abs{\CC}\leq  \left(\brute(\alpha)\right)^{n} \cdot n^{\OO(1)}$. Furthermore, the algorithm runs in time $\left(\brute(\alpha)\right)^{n} \cdot n^{\OO(1)}$.
\end{lemma}

The next auxiliary lemma is a weaker version of \Cref{lem:covering_family} and we will later show how to use it to prove \Cref{lem:covering_family}. Recall that $n$ denotes the size of the universe $U$.

\begin{lemma}
	\label{lem:covering_family_epsilon} 
	There exists an algorithm $\mathscr{A}$ which given a finite set $U$, a weight function $\w:U\rightarrow \NN$, $\beta > 1$ and $0 < \delta < 1$, constructs a $(1 + \delta) \cdot \beta$-covering family $\CC$ of $U$ and $\w$ such that $\abs{\CC} \leq \left(\brute(\beta)\right)^{n} \cdot n^{\OO\left( \frac{1}{\delta}  \log(\frac{n}{\delta}) \right) }$.
	Furthermore, the running time of the algorithm is $\left(\brute(\beta)\right)^{n} \cdot n^{\OO\left( \frac{1}{\delta}  \log(\frac{n}{\delta}) \right) }$.
\end{lemma}

\begin{proof}
	The algorithm $\mathscr{A}$ works as follows:
	\begin{itemize}
		\item Define $\gamma \coloneqq 1 + \frac{\delta}{2} > 1$.
		For $i \geq 0$ let
		\begin{equation}\label{eq:U_i_defn_covering}
			U_i \coloneqq \{u \in U \mid \gamma^i \leq \w(u) < \gamma^{i+1}\}
		\end{equation}
		and $n_i \coloneqq |U_i|$.
		Let $I \coloneqq \{i \in \ZZ_{\geq 0} \mid U_i \neq \emptyset\}$ denote the set of indices $i \geq 0$ for which $U_i$ is non-empty.
		Note that $|I| \leq n$.

		\item  For each $i \in I$ construct a $\beta$-covering family $\CC_i$ of the universe $U_i$ and the uniform weight function, using \Cref{lem:covering_implicit}.

		\item Define $d \coloneqq \lceil (2/\delta) \cdot \log(2n/\delta)\rceil$ and for each $k \in I$, let $I_k \coloneqq \{i \in I \mid k-d \leq i \leq k\}$ denote the indices in $I$ between $k-d$ and $k$.

		\item For every $k \in I$, let $r_k \coloneqq \abs{I_k}$ and define
		\begin{align*}
			W_k &\coloneqq \bigcup_{i \in \;I \colon 1 \leq i < k - d } U_i \qquad \text{and} \\
			\mathcal{Q}_k &\coloneqq \left\{W_k \cup E_1 \cup \ldots \cup E_{r_{k}} \;\middle| \; \left( E_1, \ldots, E_{r_k} \right) \in \prod_{i \in I_k} \CC_i \right\}.
		\end{align*}

		\item Return the set $\CC \coloneqq \bigcup_{k \in I} \mathcal{Q}_k$.
	\end{itemize}

	\begin{claim}\label{claim:covering}
		The algorithm $\mathscr{A}$ returns a $(1 + \delta) \cdot \beta$-covering family of $U$ and $\w$.
	\end{claim}

	\begin{claimproof}
		Let us pick a set $S \subseteq U$.
		We show that there exists  $T \in \CC$ such that $S \subseteq T$ and $\w(T) \leq (1 + \delta) \cdot \beta \cdot \w(S)$.

		By the definition of $\gamma$ and $d$, it holds that
		\begin{equation}\label{eq:gamma_d_lb_covering}
			\gamma^d \geq \left(1 + \frac{\delta}{2}\right)^{(2/\delta) \cdot \log(2n/\delta)} \geq 2^{\log(2n/\delta)} = \frac{2n}{\delta}
		\end{equation}
		using that $(1 + \frac{1}{x})^x \geq 2$ for all $x \geq 1$. Let $k \in I$ be the largest index such that $S \cap U_k \neq \emptyset$. It holds that
		\begin{equation}\label{eq:weight_W_k_S_covering}
			\w(W_k) < n \cdot \gamma^{k - d} \leq \frac{\delta}{2} \cdot \gamma^{k} \leq \frac{\delta}{2} \cdot \w(S)
		\end{equation}
		where the first inequality follows from the fact that $\abs{W_k} \leq n$ and each $u \in W_k$ belongs to a set $U_i$ where $i \leq k-d - 1$ and therefore $\w(u) < \gamma^{k - d - 1 + 1} = \gamma^{k - d}$ by \eqref{eq:U_i_defn_covering}. The second inequality follows from \eqref{eq:gamma_d_lb_covering} and finally the last inequality holds because by definition of $k$, there exists $u \in S \cap U_k$ such that $\w(u) \geq \gamma^{k}$ by \eqref{eq:U_i_defn_covering}.

		For every $i \in I_k$ define $S_i \coloneqq S \cap U_i$. Since $\CC_i$ is a $\beta$-covering family of $U_i$ and the \emph{uniform weight function}, for each $i \in I_k$ there exists $T_i \in \CC_i$ such that $S_i \subseteq T_i$ and $\abs{T_i} \leq \beta \cdot \abs{S_i}$. Therefore for all $i \in I_k$ it holds that
		\begin{equation}\label{eq:weight_T_i_ub}
			\w(T_i) \leq \gamma^{i + 1} \cdot \abs{T_i} \leq \gamma^{i + 1} \cdot \beta \cdot \abs{S_i} \leq \gamma \cdot \beta \cdot \w(S_i)
		\end{equation}
		where the first inequality follows from the fact that $T_i \subseteq U_i$ and \eqref{eq:U_i_defn_covering}, the second inequality follows from the definition of $T_i$ and finally the last one again follows from the fact that $S_i \subseteq U_i$ and~\eqref{eq:U_i_defn_covering}.

		Let $\overline{T} \coloneqq \bigcup_{i \in I_k} T_i$ and define
		\begin{equation*}
			T \coloneqq \biggl(W_k \cup \overline{T}\biggr) \in \mathcal{Q}_k \subseteq \CC.
		\end{equation*}
		Therefore we have
		\begin{align*}
			S = S \cap U = S \cap \left( \bigcup_{i \in I} U_i \right) &= \biggl(S \cap \Bigl( \cup_{i \in I : i < k - d} U_i\Bigr) \biggr) \cup \biggl(S \cap \Bigl( \cup_{i \in I_k} U_i \Bigr) \biggr)\\
										&= \left( S \cap W_k \right) \cup \left( \bigcup_{i \in I_k}  \left( S\cap U_i\right) \right)\\
										&\subseteq W_k \cup \left( \bigcup_{i \in I_k} T_i\right)\\
										&= T.
		\end{align*}
		Finally, it also holds that
		\begin{align*}
			\w(T) &= \w(W_k) + \sum_{i \in I_k} \w(T_i) \\
			&\leq \frac{\delta}{2} \cdot \w(S) + \sum_{i \in I_k} \gamma \cdot \beta \cdot \w(S_i) &&\text{by }\eqref{eq:weight_W_k_S_covering} \text{ and } \eqref{eq:weight_T_i_ub}  \\
			&\leq \frac{\delta}{2} \cdot \beta \cdot \w(S) + \gamma \cdot \beta \cdot \w(S)\\
			&\leq (1 + \delta) \cdot \beta \cdot \w(S).
		\end{align*}
		This shows us that $\CC$ is a $(1 + \delta)\cdot\beta$-covering family of $U$ and $\w$.
	\end{claimproof}

	\begin{claim}\label{claim:size_CC}
		\begin{equation*}
			\abs{\CC} \leq \left(\brute(\beta)\right)^{n} \cdot n^{\OO\left( \frac{1}{\delta} \cdot \log(\frac{n}{\delta}) \right) }.
		\end{equation*}
	\end{claim}

	\begin{claimproof}
		By \Cref{lem:covering_implicit}, for each $i \in I$ it holds that
		\begin{equation}\label{eq:C_i_size}
			\abs{\CC_i} \leq \left(\brute(\beta)\right)^{n_i} \cdot n_i^{c}
		\end{equation}
		for some $c > 0$.
		Thus, for every $k \in I$,
		\begin{align*}
			\abs{\mathcal{Q}_k} \leq \abs{\prod_{i \in I_k} \CC_i} = \prod_{i \in I_k} \abs{\CC_i} &\leq \prod_{i \in I_k} \left(\brute(\beta)\right)^{n_i}  \cdot n_i^{c}\\
				&= \left(\brute(\beta)\right)^{\sum_{i \in I_k} n_i} \cdot n_i^{c \cdot \abs{I_k}}\\
				&\leq \left(\brute(\beta)\right)^{n} \cdot n^{c \cdot (d + 1)}\\
				&= \left(\brute(\beta)\right)^{n} \cdot n^{\OO\left( \frac{1}{\delta} \cdot \log(\frac{n}{\delta}) \right) }.
		\end{align*}
		Finally, we have that
		\begin{equation*}
			\abs{\CC} = \abs{\bigcup_{k \in I} \mathcal{Q}_k} \leq \sum_{k \in I} \abs{\mathcal{Q}_k} = \left(\brute(\beta)\right)^{n} \cdot n^{\OO\left( \frac{1}{\delta} \cdot \log(\frac{n}{\delta}) \right) }
		\end{equation*}
		since $\abs{I} \leq n$.
	\end{claimproof}

	\begin{claim}\label{claim:cover_rt}
		The running time of $\mathscr{A}$ is $\left(\brute(\beta)\right)^{n} \cdot n^{\OO\left( \frac{1}{\delta} \cdot \log(\frac{n}{\delta}) \right) }$.
	\end{claim}

	\begin{claimproof}
		The construction of $\{\CC_i\}_{i \in I}$ takes
		\begin{align}\label{eq:membership_runtime_1}
			\sum_{i \in I} \left(\brute(\beta)\right)^{n_i} \cdot n_i^{\OO(1)} \leq \left(\brute(\beta)\right)^{n} \cdot n^{\OO(1)}
		\end{align}
		which follows from \Cref{lem:covering_implicit}.
		Finally, the construction of $\CC$ takes time proportional to the size of $\CC$ where we have
		\begin{align*}
			\abs{\CC} = \left(\brute(\beta)\right)^{n} \cdot n^{\OO\left( \frac{1}{\delta} \cdot \log(\frac{n}{\delta}) \right) }
		\end{align*}
		by \Cref{claim:size_CC}.
		All in all, the running time of $\mathscr{A}$ is upper bounded by
		\begin{equation*}
			\left(\brute(\beta)\right)^{n} \cdot n^{\OO\left( \frac{1}{\delta} \cdot \log(\frac{n}{\delta}) \right) }.\qedhere
		\end{equation*}
	\end{claimproof}

	The lemma follows from \Cref{claim:covering,claim:size_CC,claim:cover_rt}.
\end{proof}

To complete the proof of \Cref{lem:covering_family} we use \Cref{lem:covering_family_epsilon} together with the following technical lemma.  

\begin{restatable}{lemma}{epssubexponential}
	\label{lem:eps_subexponential}
	Let $f \colon I \to \mathbb{R}$ be a continuous function on an open interval $I\subseteq \mathbb{R}$ and let $\alpha > 1$ such that $\alpha \in I$. 
	Define $\beta(n) \coloneqq \alpha - \frac{1}{\log(n)}$ and $\delta(n) \coloneqq \frac{\alpha}{\beta(n)} - 1$ for all $n \in \mathbb{N}$.
	Then it holds that
	\begin{equation*}
		f\bigl(\beta(n)\bigr)^{n} \cdot n^{\OO\left( \frac{1}{\delta(n)} \cdot \log\left(\frac{n}{\delta(n)}\right) \right) } = f(\alpha)^{n + o(n)}.
	\end{equation*}
\end{restatable}

The proof of \Cref{lem:eps_subexponential} is given in \Cref{sec:eps_subexponential}.

\begin{proof}[Proof of \Cref{lem:covering_family}]
	We claim that the algorithm $\mathscr{A}$ from \Cref{lem:covering_family_epsilon} with $\beta \coloneqq \alpha - \frac{1}{\log(n)}$ and $\delta \coloneqq \frac{\alpha}{\beta} - 1$ satisfies the properties listed in \Cref{lem:covering_family}.
	Note that $\beta$ and $\delta$ are functions of $n$, however we write $\beta$ and $\delta$ instead of $\beta(n)$ and $\delta(n)$ for the sake of readability.
	
	Observe that we have $(1 + \delta) \cdot \beta = \frac{\alpha}{\beta} \cdot \beta = \alpha$.
	Hence, by \Cref{lem:covering_family_epsilon}, the set returned by $\mathscr{A}$ is an $\alpha = (1 + \delta) \cdot \beta$-covering family $\CC$ of $U$ and $\w$ such that $\abs{\CC} \leq \left(\brute(\beta)\right)^{n} \cdot n^{\OO\left( \frac{1}{\delta}  \log(\frac{n}{\delta}) \right) }$.
	The running time of the algorithm is also bounded by $\left(\brute(\beta)\right)^{n} \cdot n^{\OO\left( \frac{1}{\delta} \cdot \log(\frac{n}{\delta}) \right) }$.

	By \eqref{eq:brute_defn}, $\brute(x)$ is a continuous function of $x$ because the entropy function is continuous. Therefore by \Cref{lem:eps_subexponential} it holds that
	\begin{equation*}
		\left(\brute(\beta)\right)^{n} \cdot n^{\OO\left( \frac{1}{\delta}  \log(\frac{n}{\delta}) \right) } = \left(\brute(\alpha)\right)^{n + o(n)}
	\end{equation*}
	which proves the lemma.
\end{proof}

%% file: extension_oracles.tex
In this section we prove \Cref{thm:weighted_amls}.
The proof presented here follows the outline of the proof of \Cref{thm:brute_force} in \Cref{sec:brute_force}, using the concept of \emph{$(\alpha,\beta)$-extension family}, an adaptation of the term $\alpha$-covering family presented in \Cref{sec:brute_force} to the setting of extension oracles.
While the items in a covering family are queries to a membership oracle, and hence are subsets of~$U$, the items in an extension family represent queries to the extension oracle, and thus are pairs $(T,\ell)$ of a subset $T$ of $U$ and a non-negative integer $\ell$.
A reduction to a construction from \cite{EsmerKMNS23} is used to build the extension family, and the extension family itself can be trivially used to attain a $\beta$-approximation algorithm for $\WSM$.

\begin{definition}[Extension Family]
	Let $U$ be a finite set and $\w:U\rightarrow \NN$ be a weight function.
	Furthermore, let $\alpha,\beta\geq 1$.
	We say $\CE\subseteq 2^U\times \NN$ is an \emph{$(\alpha,\beta)$-extension family} of $U$ and~$\w$ if for every $S\subseteq U$ there exists $(T,\ell) \in \CE$ which satisfies the following:
	\begin{align}\label{eq:defn_extension}
		\begin{split}
			\abs{S\setminus T} &\leq \ell,\\
			\w(T) + \alpha \cdot \w(S\setminus T) &\leq \beta \cdot \w(S).
		\end{split}
	\end{align}
\end{definition}

Let $c\geq 1$. The \emph{$c$-cost} of an $(\alpha,\beta)$-extension family $\CE$ of $U$ and $\w$ is defined as $\cost_c(\CE) \coloneqq \sum_{(T,\ell)\in \CE} c^{\ell}$.
The proof of \Cref{thm:weighted_amls} relies on the following lemma.

\begin{lemma}
	\label{lem:extension_construction}
	Let $\alpha,c\geq 1$, $\beta>1$ and $\eps > 0$.
	Then there is an algorithm which given a finite set $U$ and a weight function $\w:U\rightarrow \NN$ returns an $(\alpha,\beta)$-extension family $\CE$ of $U$ and~$\w$ such that $\cost_c(\CE) \leq \OO\left(\Bigl( \amlsbound(\alpha,c,\beta) + \eps\Bigr)^{n}\right)$.
	Furthermore, the running time of the algorithm is $\OO\left(\Bigl( \amlsbound(\alpha,c,\beta) + \eps\Bigr)^{n}\right)$.
\end{lemma}

Before heading to the proof of \Cref{lem:extension_construction}, we show how the lemma can be used to prove \Cref{thm:weighted_amls}.

\begin{proof}[Proof of \Cref{thm:weighted_amls}]
	Let $\alpha, c \geq 1$, $\beta > 1$ and $\eps > 0$. Consider the following algorithm $\mathscr{A}$:

	\begin{itemize}
		\item Given the  input $(U, \w, \mathcal{F})$, use \Cref{lem:extension_construction} with $\alpha, \beta, c$ and $\eps$ to construct an $(\alpha, \beta)$-extension family $\CE$ of $U$ and $\w$. 
		\item Let $\oracle$ be the $\alpha$-extension oracle. For each $(T_i, \ell_i) \in \CE$, use $\oracle$ to compute $X_i \coloneqq \oracle(T_i, \ell_i)$.
		\item Define $\mathcal{T} \coloneqq \{T_i \cup X_i \mid (T_i, \ell_i) \in \CE\}$ and return a set in $\mathcal{T}$ with the minimum weight, i.e. a set $T \in \mathcal{T}$ such that $\w(T) = \min \{\w(Y) \mid Y \in \mathcal{T}\}$. 
	\end{itemize}

	The algorithm $\mathscr{A}$ first creates an $(\alpha, \beta)$-extension family $\CE$.
	Then it goes over all elements $(T_i, \ell_i) \in \CE$ and queries the oracle with $(T_i, \ell_i)$.
	The running time of this algorithm in the $(\alpha,c)$-extension model is therefore equal to the running time of the algorithm from \Cref{lem:extension_construction} plus $c^{\ell_i}$ for every query $(T_i, \ell_i)$, i.e., the cost of the extension family $\CE$.
	By \Cref{lem:extension_construction} this value is at most
	\begin{align*}
		\left( \amlsbound(\alpha,c,\beta)  + \eps\right)^{n} + \cost_c(\CE) &\leq \OO\Bigl(\left( \amlsbound(\alpha,c,\beta) + \eps\right)^{n}\Bigr) + \OO\Bigl(\left( \amlsbound(\alpha,c,\beta) + \eps\right)^{n}\Bigr)\\
		&= \OO\Bigl(\left( \amlsbound(\alpha,c,\beta) + \eps\right)^{n}\Bigr)
	\end{align*}

	\begin{claim}
		The algorithm $\mathscr{A}$ is a deterministic $\beta$-approximation for $\WSM$ in the $(\alpha,c)$-extension model.
	\end{claim}

	\begin{claimproof}
		Let $R \in \mathcal{T}$ be the set returned by the algorithm. Note that by definition of $\mathcal{T}$, $R = T_j \cup X_j$ for some $(T_j, \ell_j) \in \CE$ and $X_j = \oracle(T_j, \ell_j)$.
		Also observe that $X_j = \oracle(T_j, \ell_j)$ satisfies  $R = X_j \cup T_j \in \mathcal{F}$ (\Cref{def:extension}).
		
		Let $S \in \mathcal{F}$ be a set with minimum weight in $\mathcal{F}$, i.e. $\w(S) = \min \{\w(Y) \mid Y \in \mathcal{F}\} = \opt(U,\w,\CF)$.
		Since $\CE$ is an $(\alpha, \beta)$-extension family, there exists $(T_i, \ell_i) \in \CE$ such that $S, T_i$ and $\ell_i$ satisfy \eqref{eq:defn_extension} for $T = T_i$ and $\ell = \ell_i$.
		Moreover, observe that 
		\begin{equation*}
			S \setminus T_i \in \{X \subseteq U \mid \abs{X} \leq \ell_i, X \cup T_i \in \mathcal{F}\} 
		\end{equation*}
		because $S \subseteq \left( S \setminus T_i \right) \cup T_i \in \mathcal{F}$. Hence by the definiton of $\oracle$ and $(T_i, \ell_i)$ it follows that
		\begin{align*}
			\w(X_i) = \w\left(\oracle (T_i,\ell_i)\right) &\leq \alpha \cdot \min \left\{ \w(X)~|~X\subseteq U,~\abs{X}\leq \ell_i,~X\cup T_i\in \CF \right\}\\
			&\leq \alpha \cdot \w\left( S \setminus T_i \right)\\
			       &\leq \beta \cdot \w(S) - \w(T_i).
		\end{align*}
		Finally, we have
		\begin{align*}
			\w(R) &= \min \{\w(Y) \mid Y \in \mathcal{T}\}\\
			     &\leq \w(T_i \cup X_i) \\
			     &\leq \w(T_i) + \w(X_i) \\
			     &\leq \beta \cdot \w(S)\\
			     &= \beta \cdot \opt(U,\w,\CF)
		\end{align*}
		which proves the claim.
	\end{claimproof}

	This shows that the algorithm $\mathscr{A}$ is a deterministic $\beta$-approximation for $\WSM$ in the $(\alpha,c)$-extension model with the running time given in \Cref{thm:weighted_amls}.
\end{proof}

The proof of \Cref{lem:extension_construction} relies on the following construction for extension families in the unweighted case which was implicitly given in \cite{EsmerKMNS23}.

\begin{lemma}[\cite{EsmerKMNS23}]\label{lem:extension_implicit}
	\label{lem:unweighted_extension}
	Let $\alpha,c\geq 1$ and $\beta>1$. Then there is a deterministic algorithm which given a finite set $U$, returns  an $(\alpha,\beta)$-extension family $\CE$ of $U$ and the uniform weight function $\w:U\rightarrow \{1\}$ such that $\cost_c(\CE) \leq \Bigl( \amlsbound(\alpha,c,\beta)\Bigr)^{n + o(n)}$. Furthermore, the running time of the algorithm is $\Bigl( \amlsbound(\alpha,c,\beta)\Bigr)^{n + o(n)}$.
\end{lemma}

In a manner analogous to \Cref{sec:brute_force}, we begin by introducing a related lemma, \Cref{lem:extension_family_epsilon}, which presents a slightly weaker form of \Cref{lem:extension_construction}. Within \Cref{lem:extension_family_epsilon}, $\betaext$ takes the place of the previously mentioned $\beta$ from \Cref{lem:extension_construction}. Subsequently, in the proof of \Cref{lem:extension_construction}, we set the value of $\zeta$ as a function of $\beta$.

\begin{lemma}\label{lem:extension_family_epsilon}
	Let $\alpha,c\geq 1$, $\betaext>1$ and $ 0 < \delta < 1$. Then there is an algorithm which given a finite set $U$ and a weight function $\w:U\rightarrow \NN$ returns  an $(\alpha,(1 + \delta) \cdot \betaext)$-extension family $\CE$ of $U$ and $\w$ such that $\cost_c(\CE) \leq \Bigl( \amlsbound(\alpha,c, \betaext)\Bigr)^{n + o(n)}$. Furthermore, the running time of the algorithm is $\Bigl( \amlsbound(\alpha,c,\betaext)\Bigr)^{n+o(n)}$.
\end{lemma}

The proof of \Cref{lem:extension_family_epsilon} contains repeated arguments from the proof of \Cref{lem:covering_family_epsilon}. To enhance the overall readability and reduce the notational burden, we keep the common arguments in both proofs. Finally, we need the following technical lemma whose proof is given in \Cref{sec:o_notation}.

\begin{restatable}{lemma}{onotation}
	\label{lem:o_notation}
	Let $g,d\colon \NN \to \NN$ be two functions such that $g \in n + o(n)$ and $d \in o(n)$.
	
	We define $f\colon \NN \to \NN$ via
	\[f(n) = \max_{n = n_1 + \dots + n_{d(n)}} \sum_{i=1}^{d(n)} g(n_i).\]
	Then $f \in n + o(n)$.
\end{restatable}

\begin{proof}[Proof of \Cref{lem:extension_family_epsilon}]
We claim that \Cref{algo:extension} satisfies the conditions stated in the lemma.
\begin{algorithm}
	\begin{algorithmic}[1]
		\Configuration $\alpha \geq 1$, $c \geq 1$, $\betaext > 1$ and $0 < \delta < 1$
		\Input A universe $U$, weight function $\w \colon U \to \mathbb{N}$
		\State Define $\gamma \coloneqq 1 + \frac{\delta}{2} > 1$.
	For $i \geq 0$ let
	\begin{equation}\label{eq:U_i_defn_extension}
		U_i \coloneqq \{u \in U \mid \gamma^i \leq \w(u) < \gamma^{i+1}\}
	\end{equation}
	and $n_i \coloneqq |U_i|$.  Let $I \coloneqq \{i \in \ZZ_{\geq 0} \mid U_i \neq \emptyset\}$ denote the set of indices $i \geq 0$ for which $U_i$ is non-empty. 
	Note that $|I| \leq n$.

	\State For each $i \in I$ construct an $\Bigl( \alpha, \betaext \Bigr)$-extension family $\CE_i$ of the universe $U_i$ and the uniform weight function, using \Cref{lem:extension_implicit}, with respect to $\alpha$, $c$ and $\betaext$.

		\State Define $d \coloneqq \lceil (2/\delta) \cdot \log(2n/\delta)\rceil$ and for each $k \in I$, let $I_k \coloneqq \{i \in I \mid k-d \leq i \leq k\}$ denote the indices in $I$ between $k-d$ and $k$.
		\State \label{step_algo_defn} For every $k \in I$, let $r_k \coloneqq \abs{I_k}$ and define 
	\begin{align*}
		W_k &\coloneqq \bigcup_{i \in I \colon 1 \leq i < k - d} U_i \qquad \text{and} \\
		\mathcal{Q}_k &\coloneqq \left\{ \Bigl( W_k \cup E_1 \cup \ldots \cup E_{r_{k}}, \ell_1 + \ldots + \ell_{r_k} \Bigr)  \;\middle| \; \Bigl( \left( E_1, \ell_1 \right), \ldots, \left( E_{r_k}, \ell_{r_k} \right)  \Bigr) \in \prod_{i \in I_k} \CE_i \right\}.
	\end{align*}

	\State Return the set $\CE \coloneqq \bigcup_{k \in I} \mathcal{Q}_k$. 	
	\end{algorithmic}
	\caption{Extension Family for Arbitrary Weight Functions}\label{algo:extension}
\end{algorithm}

\begin{claim}\label{claim:extension_correctness}
	The set $\CE$ returned by \Cref{algo:extension} is an $\bigl( \alpha, (1 + \delta) \cdot \betaext \bigr)$-extension family of $U$ and $\w$.
\end{claim}

\begin{claimproof}
	Let  $S \subseteq U$ be a set. We will show that there exists $(T,\ell) \in \CE$ such that $S,T$ and $\ell$ satisfy \eqref{eq:defn_extension}. 
	By the definition of $\gamma$ and $d$, it holds that 
	\begin{equation}\label{eq:gamma_d_lb_extension}
		\gamma^d \geq \left(1 + \frac{\delta}{2}\right)^{(2/\delta) \cdot \log(2n/\delta)} \geq 2^{\log(2n/\delta)} = \frac{2n}{\delta}
	\end{equation}
	using that $(1 + \frac{1}{x})^x \geq 2$ for all $x \geq 1$. Let $k \in I$ be the largest index such that $S \cap U_k \neq \emptyset$. It holds that
	\begin{equation}\label{eq:weight_W_k_S_extension}
		\w(W_k) < n \cdot \gamma^{k - d} \leq \frac{\delta}{2} \cdot \gamma^{k} \leq \frac{\delta}{2} \cdot \w(S)
	\end{equation}
	where the first inequality follows from the fact that $\abs{W_k} \leq n$ and each $u \in W_k$ belongs to a set $U_i$ where $i \leq k-d - 1$ and therefore $\w(u) < \gamma^{k - d - 1 + 1} = \gamma^{k - d}$ by \eqref{eq:U_i_defn_extension}. The second inequality follows from \eqref{eq:gamma_d_lb_extension} and finally the last inequality holds because by the definition of $k$, there exists $u \in S \cap U_k$ such that $\w(S) \geq \w(u) \geq \gamma^{k}$ by \eqref{eq:U_i_defn_extension}.

	For $i \in I_k$ let $S_i \coloneqq S \cap U_i$. Since $\CE_i$ is an $(\alpha, \betaext)$-extension family of $U_i$ and the \emph{uniform weight function}, for each $i \in I_k$ there exists $(T_i, \ell_i) \in \CE_i$ such that
	\begin{align}\label{eq:extension_family_uniform}
		\begin{split}
			\abs{S_i \setminus T_i} &\leq \ell_i\\
			\abs{T_i} + \alpha \cdot \abs{S_i \setminus T_i} &\leq \betaext \cdot \abs{S_i}. 
		\end{split}
	\end{align}

	Let $\overline{T} \coloneqq \bigcup_{i \in I_k} T_i$ and define
	\begin{align*}
		\begin{split}
			T &\coloneqq W_k \cup \overline{T}\\
			\ell &\coloneqq \sum_{i \in I_k} \ell_i.
		\end{split}
	\end{align*}
	By definition of $\mathcal{Q}_k$ in \Cref{algo:extension} it holds that $(T,\ell) \in \mathcal{Q}_k \subseteq \CE$. Observe that
	\begin{align}\label{eq:ST_I_k}
		\begin{split}
		S \setminus T &= \left( S \setminus T \right) \cap U\\
			      &= \left( S \setminus T \right) \cap \left( W_k \cup \biggl( \;\bigcup_{i \in I_k} U_i \biggr)  \right) \\
			      &= \Bigl( (S\setminus T) \cap W_k \Bigr) \cup \left( (S \setminus T) \cap \biggl( \;\bigcup_{i \in I_k} U_i \biggr)  \right)\\
			      &=  (S \setminus T) \cap \biggl( \;\bigcup_{i \in I_k} U_i \biggr)\\
			      &= \bigcup_{i \in I_k} \left(U_i \cap (S \setminus T)\right)\\
			      &= \bigcup_{i \in I_k} S_i \setminus T_i			
		\end{split}
	\end{align}
	where the second equality follows from $S \subseteq \bigcup_{1 \leq i \leq k} U_i = W_k \cup \left( \bigcup_{k-d \leq i \leq k} U_i \right)$ and the fourth inequality holds because $W_k \subseteq T$ which further implies $(S \setminus T) \cap W_k = \emptyset$. Therefore it holds that
	\begin{align*}
		\abs{S \setminus T} = \abs{\bigcup_{i \in I_k} S_i \setminus T_i} \leq \sum_{i \in I_k} \abs{S_i \setminus T_i} \leq \sum_{i \in I_k} \ell_i = \ell.
	\end{align*}
	where the second inequality follows from \eqref{eq:extension_family_uniform}.

	For all $i \in I_k$ we also have that 
	\begin{align}\label{eq:weight_ineq_T_i}
		\begin{split}
		\w(T_i) + \alpha \cdot \w(S_i \setminus T_i) &\leq \abs{T_i}\cdot \gamma^{i+1} + \alpha \cdot \gamma^{i+1} \cdot \abs{S_i \setminus T_i}\\
							   &\leq \gamma^{i+1} \cdot \left( \abs{T_i} + \alpha \cdot \abs{S_i \setminus T_i} \right)\\
							   &\leq \gamma^{i+1} \cdot \betaext \cdot \abs{S_i}\\
							   &\leq \betaext \cdot \gamma \cdot \w(S_i)			
		\end{split}
	\end{align}
	where the third inequality follows from \eqref{eq:extension_family_uniform}. Therefore we have that 
	\begin{align*}
		\w(T) + \alpha\cdot \w(S\setminus T) &= \w(W_k) + \w\left( \bigcup_{i \in I_k} T_i \right) + \alpha \cdot \w\left( S \setminus T \right) \\
						   &\leq \w(W_k) + \sum_{i \in I_k} \Biggl( \w(T_i) + \alpha \cdot \w(S_i \setminus T_i) \Biggr) &&\text{ by } \eqref{eq:ST_I_k}\\
						   &< \frac{\delta}{2} \cdot \w(S)  + \sum_{i \in I_k} \Biggl( \w(T_i) + \alpha \cdot \w(S_i \setminus T_i) \Biggr)&& \text{ by } \eqref{eq:weight_W_k_S_extension} \\
						   & \leq \frac{\delta}{2} \cdot \w(S) +  \sum_{i \in I_k}  \betaext \cdot \gamma \cdot \w(S_i) && \text{ by } \eqref{eq:weight_ineq_T_i} \\
		& < \betaext \cdot \frac{\delta}{2} \cdot \w(S) +  \sum_{i \in I_k}  \betaext \cdot \gamma \cdot \w(S_i) &&\text{ since $\betaext > 1$} \\
		& \leq \betaext \cdot \frac{\delta}{2} \cdot \w(S) + \betaext \cdot \gamma \cdot \w(S)\\
		&\leq \left( 1 + \delta \right) \cdot \betaext \cdot \w(S)
	\end{align*}
which proves the claim.
\end{claimproof}
	
\begin{claim}\label{claim:extension_cost}
	\begin{equation*}
		\cost_c(\CE) \leq \Bigl( \amlsbound(\alpha,c,\betaext)\Bigr)^{n+o(n)}.
	\end{equation*}
\end{claim}

\begin{claimproof}
	By \Cref{lem:extension_implicit}, for each $i \in I$ it holds that
	\begin{equation}\label{eq:extension_implicit_ineq}
		\cost_c(\CE_i) \leq \Bigl( \amlsbound(\alpha,c,\betaext)\Bigr)^{n_i + o(n_i)}.
	\end{equation}
	
	For $i \in \mathbb{Z}_{\geq 0} \setminus I$, we define $\CE_i = \{(\emptyset, 0)\}$. By doing so, we can make the assumption, without loss of generality, that $\CE_i$ is nonempty for all $i \in \mathbb{Z}_{\geq 0}$. Note that this assumption does not have any effect on the value of the $\cost_c$ and it simplifies the following analysis. Also let $\mathbb{E}$ denote the Cartesian product of $\CE_{k-d}$ up to $\CE_k$, i.e. $\mathbb{E} \coloneqq \prod_{i \in \{k - d, \ldots, k\} } \CE_i$. 

	With this assumption, for all $k \in I$ we have that
	\begin{align*}
		\cost_c(\mathcal{Q}_k) &= \sum_{\bigl( (E_{k - d}, \ell_{k - d}), \ldots, (E_{k - 1}, \ell_{k - 1}), (E_{k}, \ell_{k}) \bigr) \in \mathbb{E}} c^{\ell_{k - d} + \ldots + \ell_{k - 1} + \ell_{k}}\\
				       &= \sum_{\bigl( (E_{k - d}, \ell_{k - d}), \ldots, (E_{k - 1}, \ell_{k - 1}), (E_{k}, \ell_{k}) \bigr) \in \mathbb{E}} c^{\ell_{k - d}} \cdot \ldots \cdot c^{\ell_{k - 1}} \cdot c^{\ell_k}\\
		 &= \prod_{j = k -d}^{k} \;\sum_{(E,\ell) \in \CE_j}c^{\ell}\\
		 &= \prod_{j = k - d}^{k} \;\cost_c(\CE_j)\\
		 &= \prod_{j \in I_k}\; \cost_c(\CE_j)
	\end{align*}
	
	By  \eqref{eq:extension_implicit_ineq} we obtain
	\begin{align*}
		\cost_c(\mathcal{Q}_k) &= \prod_{i \in I_k} \cost_c\left( \CE_i \right)\\
				       &\leq \prod_{i \in I_k} \Bigl( \amlsbound(\alpha,c,\betaext)\Bigr)^{n_i + o(n_i)}\\
				       &= \Bigl( \amlsbound(\alpha,c,\betaext)\Bigr)^{n + o(n)}\\
	\end{align*}
	where the last step follows from \Cref{lem:o_notation}. 

	Finally, it holds that
	\begin{align*}
		\cost_c(\CE) = \cost_c\left( \bigcup_{k \in I} \mathcal{Q}_k \right) \leq \sum_{k \in I} \;\cost_c\left( \mathcal{Q}_k \right) = \Bigl( \amlsbound(\alpha,c,\betaext)\Bigr)^{n + o(n)}
	\end{align*}
	since $\abs{I} \leq n$.
\end{claimproof}

\begin{claim}\label{claim:extension_running_time}
	The running time of \Cref{algo:extension} is $\Bigl( \amlsbound(\alpha,c,\betaext)\Bigr)^{n+o(n)}$.
\end{claim}

\begin{claimproof}
	 The construction of $\{\CE_i\}_{i \in I}$ takes
	\begin{equation*}
		\sum_{i \in I} \Bigl( \amlsbound(\alpha,c,\betaext)\Bigr)^{n_i + o(n_i)}  \leq \Bigl( \amlsbound(\alpha,c,\betaext)\Bigr)^{n+o(n)}
	\end{equation*}
	by \Cref{lem:extension_implicit}.

	Finally, the construction of each $\mathcal{Q}_k$ takes time proportional to its size, which is upper bounded by $\cost_c(\mathcal{Q}_k)$. Therefore, the construction of $\CE$ takes at most $\cost_c(\CE)$ time, where we have
	\begin{equation*}
		\cost_c(\CE) \leq \Bigl( \amlsbound(\alpha,c,\betaext)\Bigr)^{n+o(n)} 
	\end{equation*}
	by \Cref{claim:extension_cost}.
	All in all, the running time of \Cref{algo:extension} is upper bounded by
	\begin{equation*}
		\Bigl(\amlsbound(\alpha,c,\betaext)\Bigr)^{n+o(n)}.\qedhere
	\end{equation*}
\end{claimproof}
The lemma follows from \Cref{claim:extension_correctness,claim:extension_cost,claim:extension_running_time}.
\end{proof}

Lastly, we use the following property of the function $\amls$.

\begin{restatable}{lemma}{amlscontinuity}
	\label{lem:amls_continuity}
	For every fixed $\alpha > 1$ and $c > 1$, $\amls(\alpha, c, x)$ is a continuous function of $x$ on the interval $(1, \infty)$.
\end{restatable}

The proof of \Cref{lem:amls_continuity} is given \Cref{sec:eps_subexponential}.

\begin{proof}[Proof of \Cref{lem:extension_construction}.]
	Let $\alpha, c \geq 1$, $\beta > 1$ and $\eps > 0$. Since $\amlsbound(\alpha, c, \betaext)$ is continuous in $\betaext$ by \Cref{lem:amls_continuity}, there exists a $\betaext' \in\left( 1,  \beta\right)$ such that $\amlsbound(\alpha, c, \betaext') < \amlsbound(\alpha, c, \beta) + \frac{\eps}{2}$. To prove the lemma, we use \Cref{algo:extension} with $\betaext \coloneqq \betaext'$ and $\delta \coloneqq \beta / \betaext - 1$. 

	Note that we have $(1 + \delta) \cdot \betaext = \left( \beta / \betaext \right)  \cdot \betaext = \beta$. Hence, by \Cref{lem:extension_family_epsilon}, the set $\CE$ returned by \Cref{algo:extension} is an $\bigl( \alpha, \beta \bigr)$-extension family of $U$ and $\w$ such that 
	\begin{equation*}
		\cost_c(\CE) \leq \Bigl(\amlsbound(\alpha,c,\betaext')\Bigr)^{n+o(n)} \leq \left( \amls(\alpha, c, \beta) + \frac{\eps}{2} \right)^{n + o(n)} = \OO\Bigl( \left( \amls(\alpha, c, \beta) + \eps \right)^{n} \Bigr)
	\end{equation*}

	Finally, the running time of the algorithm is also bounded by
	\[\Bigl(\amlsbound(\alpha,c,\betaext')\Bigr)^{n+o(n)} = \OO\Bigl( \left( \amls(\alpha, c, \beta) + \eps \right)^{n} \Bigr)\]
	which proves the lemma.
\end{proof}

%% file: discussion.tex
In this paper, we study weighted monotone subset minimization problems, where given a universe $U$ with $n$ elements and a weight function $\w\colon U \to \NN$,
the goal is to find a subset $S \subseteq U$ which satisfies a certain fixed property and has a minimum weight.
For such problems, we show that the Approximate Monotone Local Search framework of Esmer et al.~\cite{EsmerKMNS23} can be extended to the weighted setting. 
In particular, given a parameterized $\alpha$-approximate extension algorithm, that runs in time $c^k \cdot n^{\OO(1)}$ and outputs a solution whose weight is at most $\beta \cdot \w(\OPT)$ where $\OPT$ is a solution of size at most $k$ and minimum weight, one can design an exponential $\beta$-approximation algorithm that runs faster than the proposed (natural) brute-force algorithm.

Note that for most of our applications, the parameterized approximation algorithms that are available in the literature~\cite{ShachnaiZ17,AgrawalKLS16, FominGKLS10, LokshtanovMRSZ21} provide  bi-objective guarantees which are stronger than the requirements from the $\alpha$-approximate extension algorithm.
In particular, these algorithms run in time $c^k \cdot n^{\OO(1)}$ and output a solution of size at most $\gamma \cdot k$ and weight at most $\beta \cdot W$, if a solution of size at most $k$ and weight at most $W$ exists.
That is, they (approximately) optimize the size and weight of the output solution \emph{simultaneously}.

This leads to the following natural question. Consider more restrictive weighted monotone  subset minimization problems where given a universe $U$ on $n$ vertices, a weight function $\w$ on the elements of the universe, the goal is to find a subset of the universe of size at most~$k$ that minimizes the weight and satisfies a certain fixed property. 
What is the analogue of brute-force in this setting?
Can bi-objective parameterized approximation algorithms for weighted monotone subset minimization problems be used to design faster (than brute force) exponential approximation algorithms in this restrictive setting?
What happens if we extend this setting to a bi-criteria optimization setting?

%% file: o_notation.tex
In this section, we give a proof of \Cref{lem:o_notation}.

\onotation*

\begin{proof}
	Let $\varepsilon > 0$.
	Then there is some $N_g \in \NN$ such that
	\[g(n) \leq \left(1 + \frac{\varepsilon}{2}\right) \cdot n\]
	for every $n \geq N_g$.
	Let $C_g \coloneqq \max_{n < N_g} g(n)$ and define $\Gamma(n) = \{1 \leq i \leq d(n)\} $.
	Then
	\begin{align*}
		f(n) &= \max_{n = n_1 + \dots + n_{d(n)}} \sum_{i=1}^{d(n)} g(n_i)\\
		&= \max_{n = n_1 + \dots + n_{d(n)}} \sum_{i \in \Gamma(n)\colon n_i \geq N_g} g(n_i) + \sum_{i \in \Gamma(n)\colon n_i < N_g} g(n_i)\\
		&\leq \max_{n = n_1 + \dots + n_{d(n)}} \sum_{i \in \Gamma(n)\colon n_i \geq N_g} \left(1 + \frac{\varepsilon}{2}\right) \cdot n_i + \sum_{i \in \Gamma(n) \colon n_i < N_g} C_g\\
		&\leq \left(1 + \frac{\varepsilon}{2}\right) \cdot n + d(n) \cdot C_g
	\end{align*}
	for every $n \in \NN$.
	Since $d \in o(n)$ we conclude that there is some $N_d$ such that
	\[d(n) \leq \frac{\varepsilon}{2 \cdot C_g} \cdot n\]
	for all $n \geq N_d$.
	So
	\begin{align*}
		f(n) &\leq \left(1 + \frac{\varepsilon}{2}\right) \cdot n + d(n) \cdot C_g\\
		&\leq \left(1 + \frac{\varepsilon}{2}\right) \cdot n + \frac{\varepsilon}{2 \cdot C_g} \cdot n \cdot C_g\\
		&=    \left(1 + \varepsilon\right) \cdot n
	\end{align*}
	for every $n \geq N_d$.
\end{proof}

%% file: amls.tex
The definition of $\amlsbound$ relies on several auxiliary functions.
For every $\alpha,\beta,c\geq 1$ define
\begin{alignat*}{2}
	&\delta_{\alpha,\beta}(\kappa, \tau) &&=~ \begin{cases} \frac{\frac{\beta}{\alpha} \kappa -\frac{\tau}{\alpha}}{1-\tau} = \frac{\frac{\beta}{\alpha} \kappa -\frac{1}{\alpha}}{1-\tau}+\frac{1}{\alpha}& \tau \neq 1\\
		\frac{1}{\alpha} & \tau = 1
	\end{cases}\\
	&\gamma_{\alpha,\beta}(\kappa, \tau) &&=~ \begin{cases} \left(1-\frac{\beta}{\alpha}\right)\frac{\kappa}{\tau} +\frac{1}{\alpha}.
		& \tau \neq 0\\
		\frac{1}{\alpha}  & \tau=0
	\end{cases}\\
	&g_{\alpha,\beta,c}(\kappa,\tau) &&= ~\frac{\beta \kappa - \tau}{\alpha} \ln c  -
	\tau\cdot \entropy\left(\gamma_{\alpha,\beta}(\kappa,\tau)\right) -(1-\tau)\cdot \entropy\left(\delta_{\alpha,\beta} (\kappa,\tau)\right) + \entropy\left(\kappa\right) \\
	&M_{\alpha, \beta}(\kappa) &&= ~\begin{cases}
		\frac{\beta - \alpha}{1-\alpha \cdot \kappa} \cdot \kappa &\text{if } \alpha < \beta\\
		0 &\text{if } \alpha = \beta\\
		\frac{\alpha - \beta}{\alpha - 1}\cdot  \kappa &\text{if } \alpha > \beta
	\end{cases}
\end{alignat*}
We follow the standard notation in which $0 \ln 0 = 0$ and $\entropy(0) = \entropy(1)=0$.
We define $\amls(\alpha,c,\beta)$ by
\begin{equation*}
	\amlsbound(\alpha,c,\beta) \coloneqq \exp\left( \max_{~0\leq \kappa \leq \frac{1}{\beta} ~} ~\min_{~ M_{\alpha,\beta}(\kappa)  \leq  \tau \leq \beta \kappa~}~ g_{\alpha,\beta,c}(\kappa,\tau)\right).
\end{equation*}
Though $\amlsbound$ is defined using a $\max$-$\min$ optimization problem, it is shown in \cite{EsmerKMNS23} that $g_{\alpha,\beta,c}(\kappa,\tau)$ is \emph{convex} in $\tau$ for every fixed $0\leq \kappa\leq \frac{1}{\beta}$.
Thus,
\[g_{\alpha,\beta,c}^*(\kappa) = \min_{~M_{\alpha,\beta}(\kappa)  \leq  \tau \leq \beta \kappa~}~ g_{\alpha,\beta,c}(\kappa,\tau)\]
can be efficiently computed to arbitrary precision.
It is also shown in \cite{EsmerKMNS23} that $g_{\alpha,\beta,c}^*(\kappa) $ is \emph{concave}.
So its maximum can be efficiently computed to arbitrary precision as well.
Together, the convexity and concavity properties allow for an efficient computation of $\amlsbound$.

%% file: eps_subexponential.tex
\epssubexponential*

\begin{proof}
	We have that
	\begin{align*}
		f\Bigl(\beta(n)\Bigr)^{n} &= f\left( \alpha \right) ^{n} \cdot \left( \frac{f\Bigl(\beta(n)\Bigr)}{f\left( \alpha \right) } \right)^{n} \\
			     &= f\left( \alpha \right) ^{n} \cdot 2^{E(n)}
	\end{align*}
	where $E(n) \coloneqq n \cdot \log\left(  \frac{f(\beta(n))}{f(\alpha) }\right)$.

	\begin{claim}
		\begin{equation*}
			E(n) = o(n).
		\end{equation*}
	\end{claim}

	\begin{claimproof}
		To prove the claim, let us calculate
		\begin{equation*}
			\lim_{n \to \infty} \frac{E(n)}{n} = \lim_{n \to \infty} \log\left( \frac{f\Bigl(\beta(n)\Bigr)}{f\left( \alpha \right) } \right) = \log\left( \frac{\lim_{n \to \infty} f\Bigl(\beta(n)\Bigr)}{f(\alpha)} \right)  = \log\left( \frac{f(\alpha)}{f(\alpha)} \right) = 0
		\end{equation*}
		where the second and second to last step follows from the continuity of the functions $\log$ and $f$, respectively.
	\end{claimproof}

	Observe that $\frac{1}{\delta(n)} = \frac{\beta(n)}{\alpha - \beta(n)} < \frac{\alpha}{\alpha - \beta(n)} = \frac{\alpha}{1 / \log(n)} = \alpha \cdot \log(n)$ and therefore
	\begin{align*}
		n^{\frac{1}{\delta(n)} \cdot \log(\frac{n}{\delta(n)})} &< n^{\alpha \cdot \log(n) \cdot \log\left( \alpha \cdot n \cdot \log(n) \right) }\\
							      &= 2^{\alpha \cdot \log(n) \cdot \log\left( \alpha \cdot n \cdot \log(n) \right) \cdot \log(n)}\\
							      &\leq 2^{c \cdot \log^{3}\left( n \cdot \log(n) \right) }\\
							      &= f(\alpha)^{o(n)}
	\end{align*}
	where the $c$ in the second to last step is a positive constant.
	As before, the last step holds because $f(\alpha)$ is a constant depending on $\alpha$.
	All in all, we obtain that
	\begin{equation*}
		f(\beta)^{n} \cdot n^{\OO\left( \frac{1}{\delta(n)} \cdot \log(\frac{n}{\delta(n)}) \right) } = f(\alpha)^{n} \cdot 2^{E(n)} \cdot f(\alpha)^{o(n)} = f(\alpha)^{n + o(n)}.\qedhere
	\end{equation*}
\end{proof}

In the following, we state some definitions and claims given in \cite{Still18} which we use in the proof of \Cref{lem:amls_continuity}.

Let $T \subseteq \mathbb{R}^{p}$ be an open set.
For each $t \in T$, let $P(t)$ be the following optimization problem
\begin{equation*}
	P(t) \colon \min_{x \in \mathbb{R}^{n}} f(x,t) \text{ s.t. } x \in F(t) \coloneqq \{ x \in \mathbb{R}^{n} \mid g_j(x,t) \leq 0, j \in \{1, \ldots, m\} \},
\end{equation*}
for continuous functions $f$ and $g_j$, for some $m > 0$.
Let $v(t) = \min_{x \in F(t)} f(x,t)$ denote the global minimum of $P(t)$ and let $S(t)$ be the set of global minimizers of $P(t)$, i.e., $S(t) = \{x \in F(t) \mid f(x,t) = v(t)\}$.

In \cite{Still18}, two sufficients conditions are given to show that $v(t)$ is upper semi-continuous and lower semi-continuous, respectively.

\begin{lemma}[\cite{Still18}]\label{lem:v_lsc}
	Let $t \in T$.
	If there exists a $\delta_{t} > 0$ and a compact set $C_t$ such that
	\begin{equation*}
		\bigcup_{\norm{t' - t} < \delta_t} F(t') \subseteq C_{t},
	\end{equation*}
	then $v$ is lower semi-continuous at $t$.
\end{lemma}

\begin{lemma}[\cite{Still18}]\label{lem:v_usc}
	Let $t \in T$. If there exists $x \in S(t)$ such that there is a sequence $\left( x_{\ell} \right)_{\ell \in \mathbb{N}} \to x$ which satisfies
	\begin{equation*}
		g_j\left( x_\ell, t \right) < 0  \quad \forall j \in \{1, \ldots, m\}, \;\forall \ell \in \mathbb{N},
	\end{equation*}
	then $v$ is upper semi-continuous at $t$.
\end{lemma}

Now we are ready to prove \Cref{lem:amls_continuity}.

\amlscontinuity*

\begin{proof}
	Define $T_1 = \{(\beta, \kappa) \mid \beta > 1, 0 < \kappa < 2\}$.
	Moreover, also define the functions $M_{\alpha}(\beta, \kappa) \coloneqq M_{\alpha, \beta}(\kappa)$ and
	\begin{align*}
		h_1(\beta, \kappa) &= \begin{cases}
			\frac{1}{\beta} & \text{ if } \kappa > \frac{1}{\beta}\\
			\kappa &\text{ otherwise} 
		\end{cases}\\
			h_2(\beta, \kappa, \tau) &= \begin{cases}
				M_{\alpha}(\beta, \kappa) & \text{ if } \tau < M_\alpha(\beta, \kappa)\\
				\beta \cdot \kappa & \text{ if } \tau > \beta \cdot \kappa\\
				\tau  &\text{ otherwise}. 
			\end{cases}
	\end{align*}

	For each $(\beta, \kappa) \in T_1$ write
	\begin{equation*}
		P_1(\beta, \kappa) \colon \min_{\tau \in \mathbb{R}} \tilde{g}(\beta, \kappa, \tau) \; \text{ s.t. } \tau \in F_1(\beta, \kappa) \coloneqq \{\tau \in \mathbb{R} \mid \tau - \beta \cdot \kappa \leq 0 , -t + M_{\alpha}(\beta, \kappa) \leq 0\}
	\end{equation*}
	where
	$\tilde{g}\left(\beta, \kappa, \tau\right) = g_{\alpha, \beta, c} \Bigl(h_1\left( \beta, \kappa \right) , h_2\left( \beta, \kappa, \tau \right) \Bigr) $.
	Recall that $M_{\alpha, \beta}(\kappa)$ and $g_{\alpha, \beta, c}(\kappa,\tau)$ are both defined in \Cref{sec:amls_def}.

Note that $\tilde{g}$ is well defined for all $(\beta, \kappa) \in T_1$ and $\tau \in \mathbb{R}$.
Moreover, since the entropy function is continuous and $g_{\alpha, \beta, c}(\kappa,\tau)$ consists of elementary arithmetic operations and composition of continuous functions, $g_{\alpha, \beta, c}(\kappa,\tau)$ is a continuous function of $\beta > 1$, $0 < \kappa < \frac{1}{\beta}$ and $\tau \in [M_\alpha(\beta, \kappa), \beta \cdot \kappa]$.
Finally, since $h_1$ and $h_2$ are continuous functions, their composition with $g$, i.e., the function $\tilde{g}$, is also continuous.
Finally, the functions $\tau - \beta \cdot \kappa$ and $-\tau + M_{\alpha}(\beta, \kappa)$ are also continuous since they consist of elementary arithmetic operations.

Let $v_1$ be the function that maps $(\beta, \kappa) \in T_1$ to the global minimum of $P_1(\beta, \kappa)$.

\begin{claim}\label{claim:v_1_lsc}
	$v_1$ is lower semi-continuous at $(\beta,\kappa)$ for all $(\beta, \kappa) \in T_1$.
\end{claim}

\begin{claimproof}
	Let $(\beta,\kappa) \in T_1$. Observe that for all $(\beta',\kappa') \in T_1$, it holds that $F_1(\beta',\kappa') \subseteq [0,1]$. Let~$\delta$ be an arbitrarily small number. We have that
	\begin{equation*}
		\bigcup_{\norm{(\beta, \kappa) - (\beta', \kappa')} < \delta} F_1(\beta', \kappa') \subseteq [0,1],
	\end{equation*}
	therefore by \Cref{lem:v_lsc}, $v_1$ is lower semi-continuous at $(\beta, \kappa)$.  
\end{claimproof}

\begin{claim}\label{claim:v_1_usc}
	$v_1$ is upper semi-continuous at $(\beta,\kappa)$ for all $(\beta, \kappa) \in T_1$.	
\end{claim}

\begin{claimproof}
	Let $(\beta,\kappa) \in T_1$. Observe that any global minimizer $\tau$ of $P_1(\beta, \kappa)$ lies in the interval $[M_{\alpha}(\beta, \kappa), \beta \cdot \kappa]$.
	Suppose $\tau \neq M_{\alpha}(\beta, \kappa)$, in that case the sequence $x_\ell \coloneqq \tau - \frac{1}{\ell}$ converges to $\tau$ and satisfies $M_\alpha(\beta, \kappa) < x_\ell < \beta \cdot \kappa$ for large enough $\ell$.
	If $\tau = M_{\alpha}(\beta, \kappa)$, then the same argument works for $x_\ell \coloneqq \tau + \frac{1}{\ell}$.
	Therefore by \Cref{lem:v_usc} it holds that $v_1$ is upper semi-continuous at $(\beta,\kappa)$.
\end{claimproof}

\Cref{claim:v_1_usc,claim:v_1_lsc} together imply that $v_1$ is continuous on $T_1$. 

\begin{claim}
	$v_1$ is also continuous at $(\beta, 0)$ for $\beta > 1$.
\end{claim}

\begin{claimproof}
	Let $\left(\beta_\ell, \kappa_\ell \right)_{\ell \in \mathbb{N}}$ be a sequence
	such that $ \beta_\ell > 1$, $\kappa_\ell \geq 0$ and $\left( \beta_\ell, \kappa_\ell \right) \to \left( \beta, 0 \right)$.
	For each $\ell \in \mathbb{N}$, let $\tau^{*}_{\ell}$ be the value such that
	$v_1(\beta_\ell, \kappa_\ell) = \tilde{g}(\beta_{\ell}, \kappa_\ell,
	\tau^{*}_{\ell})$. Note that since $M_\alpha(\beta_\ell, \kappa_\ell) \leq
	\tau^{*}_{\ell} \leq  \beta_\ell \cdot \kappa_\ell$ and
	\begin{equation*}
		\lim_{\ell \to \infty} M_\alpha(\beta_\ell, \kappa_\ell) = 0 = \lim_{\ell \to \infty} \beta_\ell \cdot \kappa_\ell,
	\end{equation*}
	it holds that $\lim_{\ell \to \infty} \tau^{*}_{\ell} = 0$ by the sandwich theorem.
	Therefore, by using the continuity of the function $\tilde{g}$, we find that
	\begin{equation*}
		\lim_{\ell \to \infty} v_1(\beta_\ell, \kappa_\ell) = \lim_{\ell \to \infty} \tilde{g}(\beta_\ell, \kappa_\ell, \tau^{*}_{\ell} ) = \tilde{g}(\beta, 0, 0) = v_1(\beta, 0),
	\end{equation*}
	where the last equality holds because $F_1(\beta, 0) = \{0\}$.
\end{claimproof}
In particular, we get that $v_1$ is continuous on the set $\{(\beta, \kappa) \mid \beta > 1, 0 \leq \kappa \leq \frac{1}{\beta}\}$.

In a similar manner, let us now define $T_2 \coloneqq \{\beta \in \mathbb{R} \mid \beta > 1\}$ and for each $\beta \in T_2$ write
\begin{equation*}
	P_2(\beta) \colon \max_{\kappa \in \mathbb{R}}\; v_1(\beta, \kappa) \text{ s.t. } \kappa \in F_2(\beta) \coloneqq \left\{\kappa \in \mathbb{R} \mid -\kappa \leq 0,\; \kappa - \frac{1}{\beta} \leq 0\right\}. 
\end{equation*}

Recall that the function that is maximized, that is $v_1$, is already shown to be continuous.
Moreover, it can be verified that the constraint functions are also continuous.

Let $v_2$ be the function that maps $\beta \in T_2$ to the global maximum of $P_2(\beta)$.
Note that for $\alpha > 1$ and $c > 1$, by the definition of $\amls$ it holds that $v_2(\beta) = \amls(\alpha, c, \beta)$.

\begin{claim}\label{claim:v_2_lsc}
	$v_2$ is lower semi-continuous at $\beta$ for all $\beta \in T_2$.
\end{claim}

\begin{claimproof}
	Let $\beta \in T_2$. Observe that for all $\beta' \in T_2$ it holds that $F_2(\beta') \subseteq [0,1]$. Let $\delta$ be an arbitrarily small number. We have that
	\begin{equation*}
		\bigcup_{\norm{\beta - \beta'} < \delta} F_2(\beta') \subseteq [0,1],
	\end{equation*}
	therefore by \Cref{lem:v_lsc}, $v_2$ is lower semi-continuous at $\beta$.
\end{claimproof}

\begin{claim}\label{claim:v_2_usc}
	$v_2$ is upper semi-continuous at $\beta$ for all $\beta \in T_2$.
\end{claim}
\begin{claimproof}
	Let $(\beta,\kappa) \in T_1$. Observe that any global minimizer $\kappa$ of $P_2(\beta)$ lies in the interval $[0, \frac{1}{\beta}]$. Suppose $\kappa > 0$, in that case the sequence $x_\ell \coloneqq \kappa - \frac{1}{\ell}$ converges to $\kappa$ and satisfies $0 < x_\ell < \kappa \leq \frac{1}{\beta}$ for large enough $\ell$. If $\kappa = 0$, then the same argument works for $x_\ell \coloneqq \kappa + \frac{1}{\ell}$. Therefore by \Cref{lem:v_usc} it holds that $v_2$ is upper semi-continuous at $\kappa$. 
\end{claimproof}

\Cref{claim:v_2_lsc,claim:v_2_usc} together imply that $v_2$ is continuous on $T_2$, which proves the lemma.
\end{proof}

%% file: problem_definitions.tex
In this section, we give the problem definitions of all the problems discussed in the paper.
For simplicity, we define the problems in their unweighted version.
In the weighted version, the vertices are equipped with weights and we are looking for a solution $S$ of minimum weight.

\medskip

\defproblem{{\sc Vertex Cover ({\sc VC})}}{An undirected graph $G$.}{Find a minimum set $S$ of vertices of $G$ such that $G-S$ has no edges.}

\defproblem{{\sc Partial Vertex Cover ({\sc PVC})}}{An undirected graph $G$ and an integer $t \geq 0$.}{Find a minimum set $S$ of vertices of $G$ such that $G-S$ has at most $|E(G)| - t$ many edges.}

\defproblem{{\sc $d$-Hitting Set ({\sc $d$-HS})}}{A universe $U$ and set family $\mathcal{F} \subseteq \binom{U}{\leq d}$.}{Find a minimum set $S \subseteq U$ such that for each $F \in \mathcal{F}$, $S \cap F \neq \emptyset$.}

\defproblem{{\sc Feedback Vertex Set ({\sc FVS})}}{An undirected graph $G$.}{Find a minimum set $S$ of vertices of $G$ such that $G-S$ is an acyclic graph.}

\defproblem{{\sc Subset Feedback Vertex Set ({\sc Subset FVS})}}{An undirected graph $G$ and a set $T \subseteq V(G)$.}{Find a minimum set $S$ of vertices of $G$ such that $G-S$ has no cycle that contains at least one vertex of $T$.}

\defproblem{{\sc Tournament Feedback Vertex Set ({\sc TFVS})}}{A tournament graph $G$.}{Find a minimum set $S$ of vertices of $G$ such that $G-S$ is an acyclic tournament.}

\defproblem{{\sc Directed Feedback Vertex Set ({\sc DFVS})}}{A directed graph $G$.}{Find a minimum set $S$ of vertices of $G$ such that $G-S$ is a directed acyclic graph.}

\defproblem{{\sc Directed Subset Feedback Vertex Set ({\sc Subset DFVS})}}{A directed graph $G$ and a set $T \subseteq V(G)$.}{Find a minimum set $S$ of vertices of $G$ such that $G-S$ has no directed cycle that contains at least one vertex of $T$.}

\defproblem{{\sc Directed Odd Cycle Transversal ({\sc DOCT})}}{A directed graph $G$.}{Find a minimum set $S$ of vertices of $G$ such that $G-S$ has no directed cycle of odd length.}

\defproblem{{\sc Multicut}}{An undirected graph $G$ and a set $\mathcal{P} \subseteq V(G) \times V(G)$.}{Find a minimum set $S$ of vertices of $G$ such that $G-S$ has no path from $u$ to $v$ for any $(u,v) \in \mathcal{P}$}

\medskip

For the next problems, we require some additional definitions.
A graph $G$ is \emph{cluster graph} if every connected component of $G$ is a complete graph.
A \emph{cograph} is a graph $G$ which does not contain $P_4$ (a path on $4$ vertices) is an induced subgraph.
Finally, a graph $G$ is a \emph{split graph} if the vertex set can be partitioned into two sets $V(G) = I \uplus C$ such that $I$ is an independent set and $C$ is a clique in $G$.

\medskip

\defproblem{{\sc Cluster Graph Vertex Deletion}}{An undirected graph $G$.}{Find a minimum set $S$ of vertices of $G$ such that $G-S$ is a cluster graph.}

\defproblem{{\sc Cograph Vertex Deletion}}{An undirected graph $G$.}{Find a minimum set $S$ of vertices of $G$ such that $G-S$ is a cograph.}

\defproblem{{\sc Split Vertex Deletion}}{An undirected graph $G$.}{Find a minimum set $S$ of vertices of $G$ such that $G-S$ is a split graph.}

%% file: running_times.tex
We provide extensive data sets on the running times for the obtained exponential approximation algorothms for the problems listed in Section \ref{sec:applications}.
Table \ref{table:runtimes-appendix-1} contains the running times for selected approximation ratios,
and graphical visualizations can be found in Figures \ref{fig:runtimes_1} and \ref{fig:runtimes_2}.

\begin{table}[H]
 \small
 \centering
 {\sc Vertex Cover}
 \medskip

 \input{plots/table_vc}

 \medskip
 {\sc Feedback Vertex Set}
 \medskip

 \input{plots/table_fvs}

 \medskip
 {\sc Tournament Feedback Vertex Set}
 \medskip

 \input{plots/table_tfvs}

 \medskip
 {\sc $3$-Hitting Set}
 \medskip

 \input{plots/table_3hs}

 \medskip
 {\sc $4$-Hitting Set}
 \medskip

 \input{plots/table_4hs}

 \medskip
 {\sc $5$-Hitting Set}
 \medskip

 \input{plots/table_5hs}

 \medskip
 {\sc Subset FVS}
 \medskip

 \input{plots/table_subset_fvs}

 \medskip
 {\sc Partial Vertex Cover}
 \medskip

 \input{plots/table_pvc}
 \caption{An entry $d$ in column $\beta$ means that the respective algorithm outputs a $\beta$-approximation in time $\CO^*(d^n)$.}
 \label{table:runtimes-appendix-1}
\end{table}

\begin{figure}[H]
 \centering
 \begin{subfigure}{.5\textwidth}
  \centering
  \caption{{\sc Vertex Cover}}
  \input{plots/figure_vc.tex}
 \end{subfigure}%
 \begin{subfigure}{.5\textwidth}
  \centering
  \caption{{\sc Feedback Vertex Set}}
  \input{plots/figure_fvs.tex}
 \end{subfigure}%

 \begin{subfigure}{.5\textwidth}
  \centering
  \caption{{\sc Tournament Feedback Vertex Set}}
  \input{plots/figure_tfvs.tex}
 \end{subfigure}%
 \begin{subfigure}{.5\textwidth}
  \centering
  \caption{{\sc $3$-Hittng Set}}
  \input{plots/figure_3hs.tex}
 \end{subfigure}%

 \begin{subfigure}{.5\textwidth}
  \centering
  \caption{{\sc $4$-Hitting Set}}
  \input{plots/figure_4hs.tex}
 \end{subfigure}%
 \begin{subfigure}{.5\textwidth}
  \centering
  \caption{{\sc $5$-Hitting Set}}
  \input{plots/figure_5hs.tex}
 \end{subfigure}%
 \caption{The figure shows running times for \textsc{Vertex Cover}, \textsc{Feedback Vertex Set}, \textsc{Tournament Feedback Vertex Set}, \textsc{$3$-Hitting Set}, \textsc{$4$-Hitting Set} and \textsc{$5$-Hitting Set}.
  A dot at $(\beta,d)$ means that the respective algorithm outputs an $\beta$-approximation in time $O^*(d^n)$.}
 \label{fig:runtimes_1}
\end{figure}
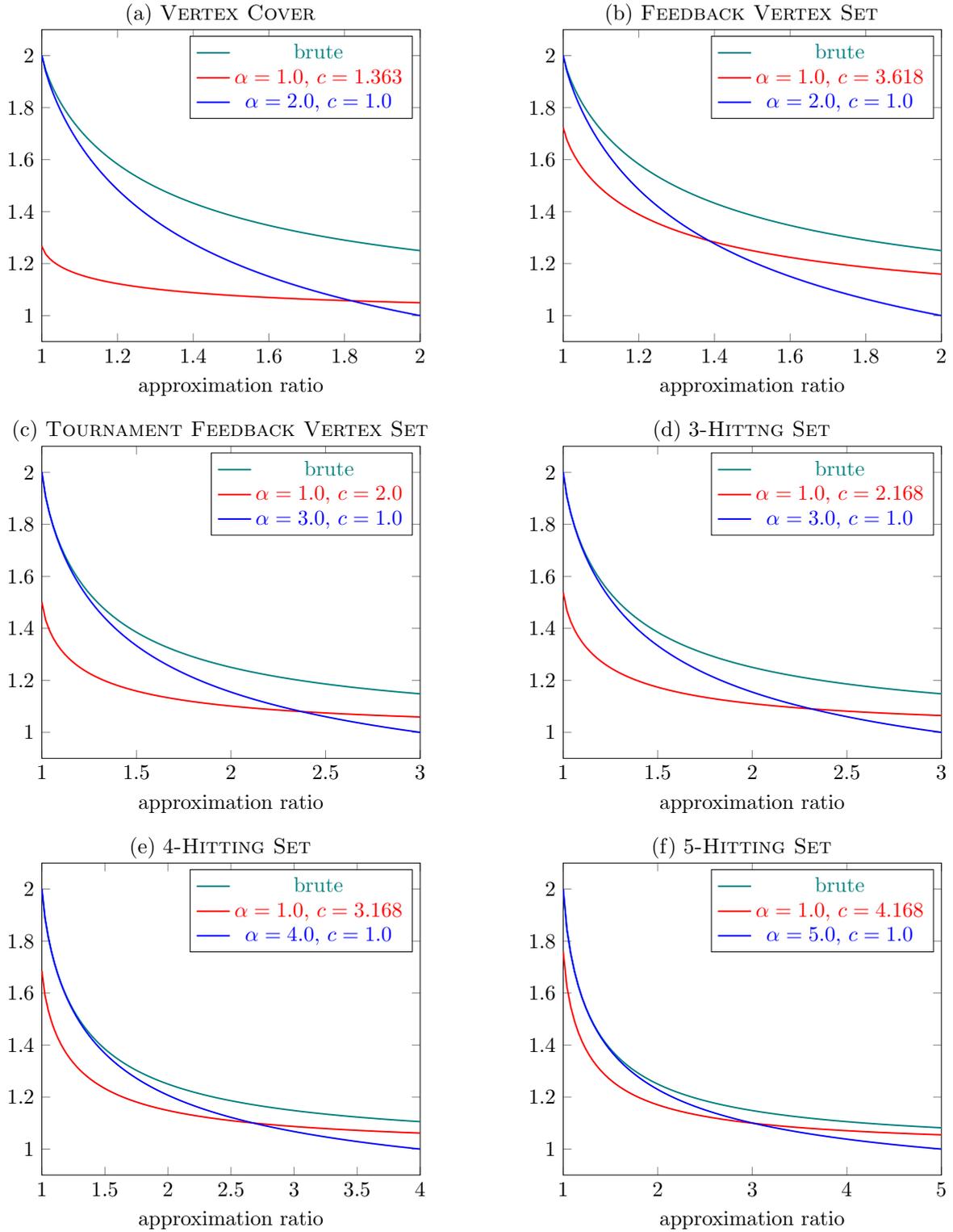

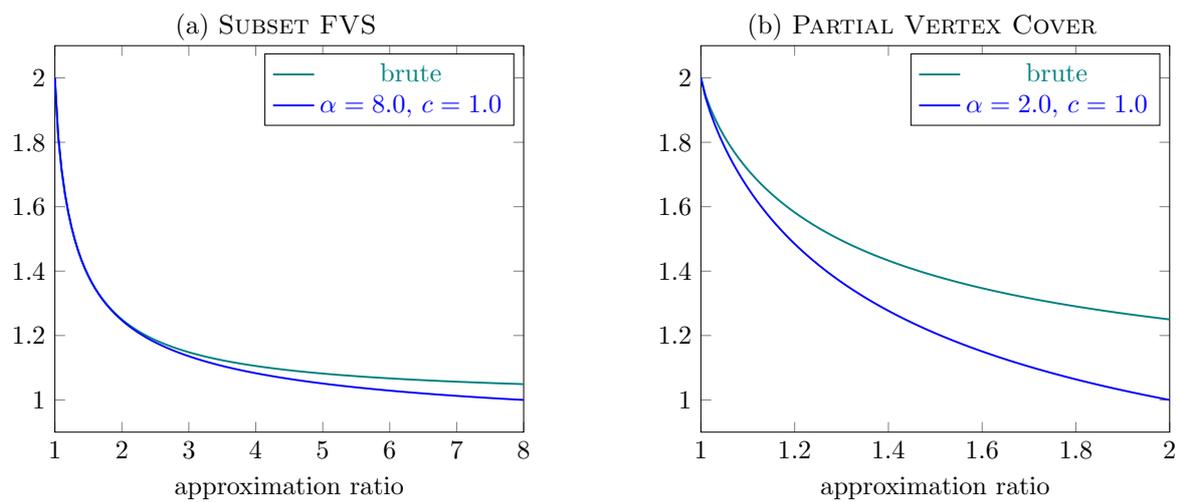
\begin{figure}[H]
 \centering
 \begin{subfigure}{.5\textwidth}
  \centering
  \caption{{\sc Subset FVS}}
  \input{plots/figure_subset_fvs.tex}
 \end{subfigure}%
 \begin{subfigure}{.5\textwidth}
  \centering
  \caption{{\sc Partial Vertex Cover}}
  \input{plots/figure_pvc.tex}
 \end{subfigure}%
  \caption{The figure shows running times for \textsc{Subset FVS} and \textsc{Partial Vertex Cover}.
  A dot at $(\beta,d)$ means that the respective algorithm outputs an $\beta$-approximation in time $O^*(d^n)$.}
 \label{fig:runtimes_2}
\end{figure}

%% file: plots/table_vc.tex
\begin{tabular}{c|c|c|c|c|c|c|c|c|c|}
	 & $1.1$ & $1.2$ & $1.3$ & $1.4$ & $1.5$ & $1.6$ & $1.7$ & $1.8$ & $1.9$\\
	\hline
	$\brute$ & $1.716$ & $1.583$ & $1.496$ & $1.433$ & $1.385$ & $1.347$ & $1.317$ & $1.291$ & $1.269$\\
	\hline
	$(\alpha = 1, c = 1.363)$ & $1.158$ & $1.123$ & $1.103$ & $1.089$ & $1.078$ & $1.07$ & $1.064$ & $1.058$ & $1.054$\\
	\hline
	$(\alpha = 2,c = 1)$ & $1.659$ & $1.485$ & $1.366$ & $1.277$ & $1.208$ & $1.151$ & $1.104$ & $1.064$ & $1.03$\\
	\hline
\end{tabular}

%% file: plots/table_fvs.tex
\begin{tabular}{c|c|c|c|c|c|c|c|c|c|}
	 & $1.1$ & $1.2$ & $1.3$ & $1.4$ & $1.5$ & $1.6$ & $1.7$ & $1.8$ & $1.9$\\
	\hline
	$\brute$ & $1.716$ & $1.583$ & $1.496$ & $1.433$ & $1.385$ & $1.347$ & $1.317$ & $1.291$ & $1.269$\\
	\hline
	$(\alpha = 1, c = 3.618)$ & $1.489$ & $1.39$ & $1.327$ & $1.283$ & $1.25$ & $1.225$ & $1.204$ & $1.187$ & $1.172$\\
	\hline
	$(\alpha = 2,c = 1)$ & $1.659$ & $1.485$ & $1.366$ & $1.277$ & $1.208$ & $1.151$ & $1.104$ & $1.064$ & $1.03$\\
	\hline
\end{tabular}

%% file: plots/table_tfvs.tex
\begin{tabular}{c|c|c|c|c|c|c|c|c|c|}
	 & $1.2$ & $1.4$ & $1.6$ & $1.8$ & $2.0$ & $2.2$ & $2.4$ & $2.6$ & $2.8$\\
	\hline
	$\brute$ & $1.583$ & $1.433$ & $1.347$ & $1.291$ & $1.25$ & $1.22$ & $1.196$ & $1.177$ & $1.162$\\
	\hline
	$(\alpha = 1, c = 2)$ & $1.251$ & $1.181$ & $1.143$ & $1.119$ & $1.102$ & $1.089$ & $1.079$ & $1.071$ & $1.065$\\
	\hline
	$(\alpha = 3,c = 1)$ & $1.566$ & $1.393$ & $1.286$ & $1.211$ & $1.155$ & $1.111$ & $1.076$ & $1.047$ & $1.022$\\
	\hline
\end{tabular}

%% file: plots/table_3hs.tex
\begin{tabular}{c|c|c|c|c|c|c|c|c|c|}
	 & $1.2$ & $1.4$ & $1.6$ & $1.8$ & $2.0$ & $2.2$ & $2.4$ & $2.6$ & $2.8$\\
	\hline
	$\brute$ & $1.583$ & $1.433$ & $1.347$ & $1.291$ & $1.25$ & $1.22$ & $1.196$ & $1.177$ & $1.162$\\
	\hline
	$(\alpha = 1, c = 2.168)$ & $1.274$ & $1.197$ & $1.156$ & $1.13$ & $1.111$ & $1.097$ & $1.086$ & $1.078$ & $1.071$\\
	\hline
	$(\alpha = 3,c = 1)$ & $1.566$ & $1.393$ & $1.286$ & $1.211$ & $1.155$ & $1.111$ & $1.076$ & $1.047$ & $1.022$\\
	\hline
\end{tabular}

%% file: plots/table_4hs.tex
\begin{tabular}{c|c|c|c|c|c|c|c|c|c|}
	 & $1.3$ & $1.6$ & $1.9$ & $2.2$ & $2.5$ & $2.8$ & $3.1$ & $3.4$ & $3.7$\\
	\hline
	$\brute$ & $1.496$ & $1.347$ & $1.269$ & $1.22$ & $1.186$ & $1.162$ & $1.143$ & $1.128$ & $1.116$\\
	\hline
	$(\alpha = 1, c = 3.168)$ & $1.305$ & $1.209$ & $1.16$ & $1.13$ & $1.11$ & $1.095$ & $1.084$ & $1.075$ & $1.068$\\
	\hline
	$(\alpha = 4,c = 1)$ & $1.489$ & $1.325$ & $1.231$ & $1.168$ & $1.122$ & $1.087$ & $1.059$ & $1.036$ & $1.017$\\
	\hline
\end{tabular}

%% file: plots/table_5hs.tex
\begin{tabular}{c|c|c|c|c|c|c|c|c|c|}
	 & $1.4$ & $1.8$ & $2.2$ & $2.6$ & $3.0$ & $3.4$ & $3.8$ & $4.2$ & $4.6$\\
	\hline
	$\brute$ & $1.433$ & $1.291$ & $1.22$ & $1.177$ & $1.149$ & $1.128$ & $1.112$ & $1.1$ & $1.09$\\
	\hline
	$(\alpha = 1, c = 4.168)$ & $1.302$ & $1.199$ & $1.149$ & $1.12$ & $1.1$ & $1.086$ & $1.076$ & $1.067$ & $1.061$\\
	\hline
	$(\alpha = 5,c = 1)$ & $1.43$ & $1.276$ & $1.193$ & $1.139$ & $1.101$ & $1.072$ & $1.048$ & $1.03$ & $1.014$\\
	\hline
\end{tabular}

%% file: plots/table_subset_fvs.tex
\begin{tabular}{c|c|c|c|c|c|c|c|c|c|}
	 & $1.7$ & $2.4$ & $3.1$ & $3.8$ & $4.5$ & $5.2$ & $5.9$ & $6.6$ & $7.3$\\
	\hline
	$\brute$ & $1.317$ & $1.196$ & $1.143$ & $1.112$ & $1.093$ & $1.079$ & $1.069$ & $1.061$ & $1.055$\\
	\hline
	$(\alpha = 8, c= 1)$ & $1.316$ & $1.19$ & $1.129$ & $1.092$ & $1.066$ & $1.046$ & $1.031$ & $1.024$ & $1.009$\\
\end{tabular}

%% file: plots/table_pvc.tex
\begin{tabular}{c|c|c|c|c|c|c|c|c|c|}
	 & $1.1$ & $1.2$ & $1.3$ & $1.4$ & $1.5$ & $1.6$ & $1.7$ & $1.8$ & $1.9$\\
	\hline
	$\brute$ & $1.716$ & $1.583$ & $1.496$ & $1.433$ & $1.385$ & $1.347$ & $1.317$ & $1.291$ & $1.269$\\
	\hline
	$(\alpha = 2,c = 1)$ & $1.659$ & $1.485$ & $1.366$ & $1.277$ & $1.208$ & $1.151$ & $1.104$ & $1.07$ & $1.03$\\
	\hline
\end{tabular}

%% file: plots/figure_tfvs.tex
\begin{tikzpicture}[scale = 0.9]
	\begin{axis}[xmin = 1, xmax = 3, ymin = 0.9, ymax = 2.1, xlabel = {approximation ratio}]

	\addplot[teal, thick] coordinates {
		(1.0, 2.0)
		(1.02, 1.9062508454)
		(1.04, 1.8440491304)
		(1.06, 1.794081097)
		(1.08, 1.7518817491)
		(1.1, 1.7152667656)
		(1.12, 1.6829321593)
		(1.14, 1.6540130203)
		(1.16, 1.6278963084)
		(1.18, 1.6041270506)
		(1.2, 1.5823559323)
		(1.22, 1.5623076557)
		(1.24, 1.5437606905)
		(1.26, 1.5265337131)
		(1.28, 1.510476187)
		(1.3, 1.4954616201)
		(1.32, 1.4813826173)
		(1.34, 1.4681471696)
		(1.36, 1.4556758212)
		(1.38, 1.443899471)
		(1.4, 1.4327576428)
		(1.42, 1.4221971069)
		(1.44, 1.4121707695)
		(1.46, 1.4026367697)
		(1.48, 1.3935577374)
		(1.5, 1.3849001795)
		(1.52, 1.3766339674)
		(1.54, 1.3687319076)
		(1.56, 1.3611693784)
		(1.58, 1.3539240206)
		(1.6, 1.3469754742)
		(1.62, 1.3403051519)
		(1.64, 1.3338960432)
		(1.66, 1.3277325456)
		(1.68, 1.3218003168)
		(1.7, 1.3160861463)
		(1.72, 1.3105778425)
		(1.74, 1.3052641335)
		(1.76, 1.3001345795)
		(1.78, 1.2951794954)
		(1.8, 1.2903898821)
		(1.82, 1.2857573652)
		(1.84, 1.2812741403)
		(1.86, 1.2769329244)
		(1.88, 1.2727269118)
		(1.9, 1.2686497351)
		(1.92, 1.2646954293)
		(1.94, 1.2608584001)
		(1.96, 1.2571333949)
		(1.98, 1.2535154767)
		(2.0, 1.25)
		(2.02, 1.2465825894)
		(2.04, 1.24325912)
		(2.06, 1.2400256991)
		(2.08, 1.23687865)
		(2.1, 1.2338144969)
		(2.12, 1.2308299514)
		(2.14, 1.2279218994)
		(2.16, 1.2250873898)
		(2.18, 1.2223236237)
		(2.2, 1.2196279447)
		(2.22, 1.2169978296)
		(2.24, 1.2144308803)
		(2.26, 1.2119248158)
		(2.28, 1.2094774652)
		(2.3, 1.2070867609)
		(2.32, 1.2047507326)
		(2.34, 1.2024675014)
		(2.36, 1.2002352749)
		(2.38, 1.1980523416)
		(2.4, 1.1959170667)
		(2.42, 1.1938278879)
		(2.44, 1.1917833111)
		(2.46, 1.1897819069)
		(2.48, 1.1878223069)
		(2.5, 1.1859032006)
		(2.52, 1.1840233324)
		(2.54, 1.1821814985)
		(2.56, 1.1803765443)
		(2.58, 1.178607362)
		(2.6, 1.1768728882)
		(2.62, 1.1751721016)
		(2.64, 1.1735040209)
		(2.66, 1.171867703)
		(2.68, 1.170262241)
		(2.7, 1.1686867625)
		(2.72, 1.1671404279)
		(2.74, 1.1656224291)
		(2.76, 1.1641319876)
		(2.78, 1.1626683537)
		(2.8, 1.1612308047)
		(2.82, 1.1598186437)
		(2.84, 1.1584311989)
		(2.86, 1.157067822)
		(2.88, 1.1557278873)
		(2.9, 1.1544107911)
		(2.92, 1.1531159499)
		(2.94, 1.1518428004)
		(2.96, 1.1505907982)
		(2.98, 1.149359417)
		(3.0, 1.1481481481)
	};

	\addplot[ red , thick] coordinates {
		(1.0, 1.5)
		(1.02, 1.4320002296)
		(1.04, 1.3924270233)
		(1.06, 1.3626484542)
		(1.08, 1.338641818)
		(1.1, 1.3185468406)
		(1.12, 1.3013082254)
		(1.14, 1.28625752)
		(1.16, 1.2729401877)
		(1.18, 1.2610312763)
		(1.2, 1.2502895893)
		(1.22, 1.2405307721)
		(1.24, 1.2316105494)
		(1.26, 1.2234137922)
		(1.28, 1.2158471143)
		(1.3, 1.2088337011)
		(1.32, 1.2023096021)
		(1.34, 1.196221012)
		(1.36, 1.1905222388)
		(1.38, 1.1851741577)
		(1.4, 1.1801430189)
		(1.42, 1.1753995138)
		(1.44, 1.1709180366)
		(1.46, 1.1666760926)
		(1.48, 1.1626538207)
		(1.5, 1.1588336037)
		(1.52, 1.1551997476)
		(1.54, 1.1517382168)
		(1.56, 1.1484364122)
		(1.58, 1.145282986)
		(1.6, 1.1422676845)
		(1.62, 1.139381215)
		(1.64, 1.1366151321)
		(1.66, 1.1339617402)
		(1.68, 1.1314140096)
		(1.7, 1.128965504)
		(1.72, 1.1266103169)
		(1.74, 1.124343017)
		(1.76, 1.1221585999)
		(1.78, 1.1200524458)
		(1.8, 1.1180202825)
		(1.82, 1.116058152)
		(1.84, 1.114162382)
		(1.86, 1.1123295593)
		(1.88, 1.1105565075)
		(1.9, 1.1088402657)
		(1.92, 1.1071780705)
		(1.94, 1.1055673394)
		(1.96, 1.1040056556)
		(1.98, 1.1024907552)
		(2.0, 1.1010205144)
		(2.02, 1.0995929392)
		(2.04, 1.0982061549)
		(2.06, 1.0968583977)
		(2.08, 1.0955480058)
		(2.1, 1.0942734127)
		(2.12, 1.0930331398)
		(2.14, 1.0918257904)
		(2.16, 1.0906500442)
		(2.18, 1.0895046516)
		(2.2, 1.0883884294)
		(2.22, 1.0873002563)
		(2.24, 1.0862390684)
		(2.26, 1.0852038562)
		(2.28, 1.0841936606)
		(2.3, 1.0832075699)
		(2.32, 1.082244717)
		(2.34, 1.0813042763)
		(2.36, 1.0803854618)
		(2.38, 1.079487524)
		(2.4, 1.0786097485)
		(2.42, 1.0777514534)
		(2.44, 1.0769119876)
		(2.46, 1.0760907293)
		(2.48, 1.0752870839)
		(2.5, 1.074500483)
		(2.52, 1.0737303825)
		(2.54, 1.0729762616)
		(2.56, 1.0722376216)
		(2.58, 1.0715139844)
		(2.6, 1.0708048918)
		(2.62, 1.070109904)
		(2.64, 1.0694285992)
		(2.66, 1.0687605726)
		(2.68, 1.0681054349)
		(2.7, 1.0674628126)
		(2.72, 1.0668323462)
		(2.74, 1.0662136902)
		(2.76, 1.0656065122)
		(2.78, 1.0650104923)
		(2.8, 1.0644253224)
		(2.82, 1.0638507058)
		(2.84, 1.0632863567)
		(2.86, 1.0627319996)
		(2.88, 1.0621873687)
		(2.9, 1.0616522079)
		(2.92, 1.0611262699)
		(2.94, 1.0606093161)
		(2.96, 1.060101116)
		(2.98, 1.0596014473)
		(3.0, 1.0591100949)
	};

	\addplot[ blue , thick] coordinates {
		(1.0, 2.0)
		(1.02, 1.9058941662)
		(1.04, 1.842764221)
		(1.06, 1.7914558434)
		(1.08, 1.7476145144)
		(1.1, 1.7091350108)
		(1.12, 1.6747714482)
		(1.14, 1.6437024677)
		(1.16, 1.6153481991)
		(1.18, 1.5892792572)
		(1.2, 1.5651662684)
		(1.22, 1.5427496004)
		(1.24, 1.521820106)
		(1.26, 1.5022062939)
		(1.28, 1.483765458)
		(1.3, 1.4663773513)
		(1.32, 1.4499395578)
		(1.34, 1.4343640312)
		(1.36, 1.4195744567)
		(1.38, 1.405504207)
		(1.4, 1.3920947355)
		(1.42, 1.3792942966)
		(1.44, 1.3670569155)
		(1.46, 1.3553415491)
		(1.48, 1.3441113975)
		(1.5, 1.3333333333)
		(1.52, 1.3229774253)
		(1.54, 1.3130165381)
		(1.56, 1.3034259933)
		(1.58, 1.2941832811)
		(1.6, 1.2852678138)
		(1.62, 1.276660713)
		(1.64, 1.2683446265)
		(1.66, 1.2603035687)
		(1.68, 1.2525227816)
		(1.7, 1.2449886134)
		(1.72, 1.2376884114)
		(1.74, 1.2306104278)
		(1.76, 1.2237437366)
		(1.78, 1.2170781592)
		(1.8, 1.2106041992)
		(1.82, 1.2043129832)
		(1.84, 1.1981962085)
		(1.86, 1.1922460957)
		(1.88, 1.186455347)
		(1.9, 1.180817107)
		(1.92, 1.1753249293)
		(1.94, 1.1699727446)
		(1.96, 1.1647548325)
		(1.98, 1.1596657961)
		(2.0, 1.1547005384)
		(2.02, 1.1498542408)
		(2.04, 1.145122344)
		(2.06, 1.1405005299)
		(2.08, 1.1359847055)
		(2.1, 1.1315709875)
		(2.12, 1.1272556893)
		(2.14, 1.1230353073)
		(2.16, 1.1189065102)
		(2.18, 1.1148661277)
		(2.2, 1.1109111406)
		(2.22, 1.1070386718)
		(2.24, 1.1032459776)
		(2.26, 1.0995304399)
		(2.28, 1.0958895588)
		(2.3, 1.0923209458)
		(2.32, 1.0888223176)
		(2.34, 1.08539149)
		(2.36, 1.0820263727)
		(2.38, 1.0787249639)
		(2.4, 1.0754853454)
		(2.42, 1.0723056786)
		(2.44, 1.0691842)
		(2.46, 1.0661192173)
		(2.48, 1.063109106)
		(2.5, 1.0601523054)
		(2.52, 1.0572473159)
		(2.54, 1.054392696)
		(2.56, 1.0515870589)
		(2.58, 1.0488290705)
		(2.6, 1.0461174463)
		(2.62, 1.0434509495)
		(2.64, 1.0408283884)
		(2.66, 1.0382486146)
		(2.68, 1.0357105208)
		(2.7, 1.0332130391)
		(2.72, 1.0307551389)
		(2.74, 1.0283358256)
		(2.76, 1.0259541389)
		(2.78, 1.0236091512)
		(2.8, 1.0212999661)
		(2.82, 1.0190257175)
		(2.84, 1.0167855678)
		(2.86, 1.014578707)
		(2.88, 1.0124043513)
		(2.9, 1.0102617425)
		(2.92, 1.0081501464)
		(2.94, 1.006068852)
		(2.96, 1.0040171708)
		(2.98, 1.0019944357)
		(3.0, 1.0)
	};

	\addlegendentry[no markers, teal]{brute}
	\addlegendentry[no markers, red]{$\alpha = 1.0$, $c = 2.0$}
	\addlegendentry[no markers, blue]{$\alpha = 3.0$, $c = 1.0$}

	\end{axis}
\end{tikzpicture}

%% file: plots/figure_4hs.tex
\begin{tikzpicture}[scale = 0.9]
	\begin{axis}[xmin = 1, xmax = 4, ymin = 0.9, ymax = 2.1, xlabel = {approximation ratio}]

	\addplot[teal, thick] coordinates {
		(1.0, 2.0)
		(1.025, 1.8891121359)
		(1.05, 1.8178991111)
		(1.075, 1.7618430249)
		(1.1, 1.7152667656)
		(1.125, 1.6754094984)
		(1.15, 1.6406357531)
		(1.175, 1.6098698521)
		(1.2, 1.5823559323)
		(1.225, 1.5575378525)
		(1.25, 1.534992244)
		(1.275, 1.5143882093)
		(1.3, 1.4954616201)
		(1.325, 1.4779979715)
		(1.35, 1.4618205222)
		(1.375, 1.4467818431)
		(1.4, 1.4327576428)
		(1.425, 1.4196421598)
		(1.45, 1.4073446605)
		(1.475, 1.3957867351)
		(1.5, 1.3849001795)
		(1.525, 1.3746253163)
		(1.55, 1.3649096489)
		(1.575, 1.3557067705)
		(1.6, 1.3469754742)
		(1.625, 1.3386790183)
		(1.65, 1.3307845174)
		(1.675, 1.3232624332)
		(1.7, 1.3160861463)
		(1.725, 1.3092315934)
		(1.75, 1.3026769593)
		(1.775, 1.2964024131)
		(1.8, 1.2903898821)
		(1.825, 1.2846228559)
		(1.85, 1.2790862173)
		(1.875, 1.273766095)
		(1.9, 1.2686497351)
		(1.925, 1.2637253881)
		(1.95, 1.2589822102)
		(1.975, 1.2544101757)
		(2.0, 1.25)
		(2.025, 1.245743071)
		(2.05, 1.2416313881)
		(2.075, 1.2376575082)
		(2.1, 1.2338144969)
		(2.125, 1.2300958852)
		(2.15, 1.2264956302)
		(2.175, 1.2230080807)
		(2.2, 1.2196279447)
		(2.225, 1.2163502616)
		(2.25, 1.2131703756)
		(2.275, 1.2100839132)
		(2.3, 1.2070867609)
		(2.325, 1.2041750465)
		(2.35, 1.2013451214)
		(2.375, 1.1985935443)
		(2.4, 1.1959170667)
		(2.425, 1.1933126192)
		(2.45, 1.1907772993)
		(2.475, 1.1883083603)
		(2.5, 1.1859032006)
		(2.525, 1.1835593541)
		(2.55, 1.1812744817)
		(2.575, 1.1790463628)
		(2.6, 1.1768728882)
		(2.625, 1.1747520527)
		(2.65, 1.1726819491)
		(2.675, 1.170660762)
		(2.7, 1.1686867625)
		(2.725, 1.1667583026)
		(2.75, 1.1648738112)
		(2.775, 1.163031789)
		(2.8, 1.1612308047)
		(2.825, 1.1594694911)
		(2.85, 1.1577465416)
		(2.875, 1.1560607067)
		(2.9, 1.1544107911)
		(2.925, 1.1527956505)
		(2.95, 1.1512141892)
		(2.975, 1.1496653572)
		(3.0, 1.1481481481)
		(3.025, 1.1466615969)
		(3.05, 1.1452047773)
		(3.075, 1.1437768006)
		(3.1, 1.1423768132)
		(3.125, 1.1410039951)
		(3.15, 1.1396575582)
		(3.175, 1.1383367447)
		(3.2, 1.1370408259)
		(3.225, 1.1357691004)
		(3.25, 1.1345208933)
		(3.275, 1.1332955545)
		(3.3, 1.1320924579)
		(3.325, 1.130911)
		(3.35, 1.1297505994)
		(3.375, 1.128610695)
		(3.4, 1.1274907459)
		(3.425, 1.1263902301)
		(3.45, 1.1253086435)
		(3.475, 1.1242454997)
		(3.5, 1.1232003287)
		(3.525, 1.1221726764)
		(3.55, 1.1211621041)
		(3.575, 1.1201681876)
		(3.6, 1.1191905166)
		(3.625, 1.1182286943)
		(3.65, 1.1172823368)
		(3.675, 1.1163510726)
		(3.7, 1.115434542)
		(3.725, 1.1145323966)
		(3.75, 1.113644299)
		(3.775, 1.1127699224)
		(3.8, 1.1119089499)
		(3.825, 1.1110610746)
		(3.85, 1.1102259987)
		(3.875, 1.1094034335)
		(3.9, 1.1085930989)
		(3.925, 1.1077947231)
		(3.95, 1.1070080425)
		(3.975, 1.1062328009)
		(4.0, 1.10546875)
	};

	\addplot[red , thick] coordinates {
		(1.0, 1.6843434343434343)
		(1.025, 1.5909150915)
		(1.05, 1.5349235534)
		(1.075, 1.4923922962)
		(1.1, 1.4579649251)
		(1.125, 1.4291082309)
		(1.15, 1.404359197)
		(1.175, 1.3827776982)
		(1.2, 1.3637170727)
		(1.225, 1.3467107592)
		(1.25, 1.3314099234)
		(1.275, 1.3175464064)
		(1.3, 1.3049094079)
		(1.325, 1.293330134)
		(1.35, 1.2826713146)
		(1.375, 1.2728198269)
		(1.4, 1.2636813716)
		(1.425, 1.2551765467)
		(1.45, 1.2472378967)
		(1.475, 1.2398076602)
		(1.5, 1.2328360231)
		(1.525, 1.2262797484)
		(1.55, 1.2201010878)
		(1.575, 1.2142669091)
		(1.6, 1.2087479887)
		(1.625, 1.2035184342)
		(1.65, 1.198555209)
		(1.675, 1.1938377378)
		(1.7, 1.1893475776)
		(1.725, 1.1850681411)
		(1.75, 1.180984463)
		(1.775, 1.1770830017)
		(1.8, 1.1733514696)
		(1.825, 1.1697786884)
		(1.85, 1.1663544633)
		(1.875, 1.1630694753)
		(1.9, 1.159915187)
		(1.925, 1.1568837604)
		(1.95, 1.1539679859)
		(1.975, 1.1511612187)
		(2.0, 1.1484573239)
		(2.025, 1.1458506278)
		(2.05, 1.143335874)
		(2.075, 1.1409081858)
		(2.1, 1.1385630313)
		(2.125, 1.1362961934)
		(2.15, 1.1341037424)
		(2.175, 1.1319820114)
		(2.2, 1.1299275746)
		(2.225, 1.1279372273)
		(2.25, 1.1260079684)
		(2.275, 1.1241369837)
		(2.3, 1.1223216321)
		(2.325, 1.1205594318)
		(2.35, 1.1188480484)
		(2.375, 1.1171852842)
		(2.4, 1.115569068)
		(2.425, 1.113997446)
		(2.45, 1.1124685737)
		(2.475, 1.1109807082)
		(2.5, 1.109532201)
		(2.525, 1.1081214921)
		(2.55, 1.1067471035)
		(2.575, 1.1054076343)
		(2.6, 1.1041017554)
		(2.625, 1.1028282051)
		(2.65, 1.1015857847)
		(2.675, 1.1003733546)
		(2.7, 1.0991898307)
		(2.725, 1.0980341811)
		(2.75, 1.096905423)
		(2.775, 1.0958026194)
		(2.8, 1.0947248769)
		(2.825, 1.0936713431)
		(2.85, 1.0926412039)
		(2.875, 1.0916336818)
		(2.9, 1.0906480337)
		(2.925, 1.0896835487)
		(2.95, 1.088739547)
		(2.975, 1.0878153774)
		(3.0, 1.0869104167)
		(3.025, 1.0860240673)
		(3.05, 1.0851557566)
		(3.075, 1.0843049353)
		(3.1, 1.0834710764)
		(3.125, 1.082653674)
		(3.15, 1.0818522422)
		(3.175, 1.0810663142)
		(3.2, 1.0802954414)
		(3.225, 1.0795391923)
		(3.25, 1.0787971521)
		(3.275, 1.0780689213)
		(3.3, 1.0773541156)
		(3.325, 1.0766523649)
		(3.35, 1.0759633124)
		(3.375, 1.0752866147)
		(3.4, 1.0746219402)
		(3.425, 1.0739689696)
		(3.45, 1.0733273947)
		(3.475, 1.072696918)
		(3.5, 1.0720772524)
		(3.525, 1.0714681207)
		(3.55, 1.0708692551)
		(3.575, 1.0702803969)
		(3.6, 1.0697012959)
		(3.625, 1.0691317105)
		(3.65, 1.0685714068)
		(3.675, 1.0680201587)
		(3.7, 1.0674777472)
		(3.725, 1.0669439606)
		(3.75, 1.0664185939)
		(3.775, 1.0659014483)
		(3.8, 1.0653923317)
		(3.825, 1.0648910574)
		(3.85, 1.0643974451)
		(3.875, 1.0639113194)
		(3.9, 1.0634325107)
		(3.925, 1.0629608542)
		(3.95, 1.0624961903)
		(3.975, 1.0620383639)
		(4.0, 1.0615872246)
	};

	\addplot[blue , thick] coordinates {
		(1.0, 2.0)
		(1.025, 1.8890983316)
		(1.05, 1.8178014188)
		(1.075, 1.7615508842)
		(1.1, 1.7146518149)
		(1.125, 1.6743399655)
		(1.15, 1.6389849419)
		(1.175, 1.6075206948)
		(1.2, 1.5792030575)
		(1.225, 1.5534881571)
		(1.25, 1.5299646094)
		(1.275, 1.5083127616)
		(1.3, 1.4882787783)
		(1.325, 1.4696574264)
		(1.35, 1.4522802276)
		(1.375, 1.4360070632)
		(1.4, 1.4207200765)
		(1.425, 1.4063191494)
		(1.45, 1.3927184847)
		(1.475, 1.3798439817)
		(1.5, 1.3676311926)
		(1.525, 1.3560237094)
		(1.55, 1.344971878)
		(1.575, 1.3344317601)
		(1.6, 1.3243642892)
		(1.625, 1.3147345774)
		(1.65, 1.305511342)
		(1.675, 1.2966664273)
		(1.7, 1.2881744037)
		(1.725, 1.2800122292)
		(1.75, 1.2721589617)
		(1.775, 1.2645955137)
		(1.8, 1.257304442)
		(1.825, 1.2502697655)
		(1.85, 1.2434768091)
		(1.875, 1.236912066)
		(1.9, 1.2305630792)
		(1.925, 1.2244183369)
		(1.95, 1.2184671804)
		(1.975, 1.2126997237)
		(2.0, 1.2071067812)
		(2.025, 1.2016798048)
		(2.05, 1.1964108268)
		(2.075, 1.1912924101)
		(2.1, 1.1863176024)
		(2.125, 1.1814798961)
		(2.15, 1.1767731921)
		(2.175, 1.1721917669)
		(2.2, 1.1677302427)
		(2.225, 1.1633835614)
		(2.25, 1.1591469599)
		(2.275, 1.1550159482)
		(2.3, 1.1509862893)
		(2.325, 1.1470539815)
		(2.35, 1.1432152411)
		(2.375, 1.1394664875)
		(2.4, 1.1358043294)
		(2.425, 1.1322255519)
		(2.45, 1.1287271049)
		(2.475, 1.1253060923)
		(2.5, 1.1219597619)
		(2.525, 1.1186854965)
		(2.55, 1.1154808055)
		(2.575, 1.112343317)
		(2.6, 1.1092707704)
		(2.625, 1.1062610101)
		(2.65, 1.103311979)
		(2.675, 1.100421713)
		(2.7, 1.0975883355)
		(2.725, 1.0948100523)
		(2.75, 1.0920851475)
		(2.775, 1.0894119784)
		(2.8, 1.0867889723)
		(2.825, 1.084214622)
		(2.85, 1.0816874829)
		(2.875, 1.0792061692)
		(2.9, 1.0767693515)
		(2.925, 1.074375753)
		(2.95, 1.0720241477)
		(2.975, 1.0697133572)
		(3.0, 1.0674422489)
		(3.025, 1.065209733)
		(3.05, 1.0630147614)
		(3.075, 1.0608563247)
		(3.1, 1.0587334512)
		(3.125, 1.0566452044)
		(3.15, 1.0545906819)
		(3.175, 1.0525690138)
		(3.2, 1.0505793608)
		(3.225, 1.0486209131)
		(3.25, 1.0466928892)
		(3.275, 1.0447945344)
		(3.3, 1.0429251198)
		(3.325, 1.0410839412)
		(3.35, 1.0392703177)
		(3.375, 1.0374835911)
		(3.4, 1.035723125)
		(3.425, 1.0339883033)
		(3.45, 1.0322785302)
		(3.475, 1.0305932286)
		(3.5, 1.0289318397)
		(3.525, 1.0272938225)
		(3.55, 1.0256786525)
		(3.575, 1.0240858215)
		(3.6, 1.0225148369)
		(3.625, 1.0209652209)
		(3.65, 1.0194365101)
		(3.675, 1.0179282549)
		(3.7, 1.0164400191)
		(3.725, 1.014971379)
		(3.75, 1.0135219236)
		(3.775, 1.0120912533)
		(3.8, 1.0106789804)
		(3.825, 1.0092847277)
		(3.85, 1.0079081291)
		(3.875, 1.0065488283)
		(3.9, 1.0052064791)
		(3.925, 1.0038807449)
		(3.95, 1.0025712979)
		(3.975, 1.0012778197)
		(4.0, 1.0)
	};

	\addlegendentry[no markers, teal]{brute}
	\addlegendentry[no markers, red]{$\alpha = 1.0$, $c = 3.168$}
	\addlegendentry[no markers, blue]{$\alpha = 4.0$, $c = 1.0$}

	\end{axis}
\end{tikzpicture}

%% file: plots/figure_5hs.tex
\begin{tikzpicture}[scale = 0.9]
	\begin{axis}[xmin = 1, xmax = 5, ymin = 0.9, ymax = 2.1, xlabel = {approximation ratio}]

	\addplot[teal, thick] coordinates {
		(1.0, 2.0)
		(1.04, 1.8440491304)
		(1.08, 1.7518817491)
		(1.12, 1.6829321593)
		(1.16, 1.6278963084)
		(1.2, 1.5823559323)
		(1.24, 1.5437606905)
		(1.28, 1.510476187)
		(1.32, 1.4813826173)
		(1.36, 1.4556758212)
		(1.4, 1.4327576428)
		(1.44, 1.4121707695)
		(1.48, 1.3935577374)
		(1.52, 1.3766339674)
		(1.56, 1.3611693784)
		(1.6, 1.3469754742)
		(1.64, 1.3338960432)
		(1.68, 1.3218003168)
		(1.72, 1.3105778425)
		(1.76, 1.3001345795)
		(1.8, 1.2903898821)
		(1.84, 1.2812741403)
		(1.88, 1.2727269118)
		(1.92, 1.2646954293)
		(1.96, 1.2571333949)
		(2.0, 1.25)
		(2.04, 1.24325912)
		(2.08, 1.23687865)
		(2.12, 1.2308299514)
		(2.16, 1.2250873898)
		(2.2, 1.2196279447)
		(2.24, 1.2144308803)
		(2.28, 1.2094774652)
		(2.32, 1.2047507326)
		(2.36, 1.2002352749)
		(2.4, 1.1959170667)
		(2.44, 1.1917833111)
		(2.48, 1.1878223069)
		(2.52, 1.1840233324)
		(2.56, 1.1803765443)
		(2.6, 1.1768728882)
		(2.64, 1.1735040209)
		(2.68, 1.170262241)
		(2.72, 1.1671404279)
		(2.76, 1.1641319876)
		(2.8, 1.1612308047)
		(2.84, 1.1584311989)
		(2.88, 1.1557278873)
		(2.92, 1.1531159499)
		(2.96, 1.1505907982)
		(3.0, 1.1481481481)
		(3.04, 1.145783995)
		(3.08, 1.1434945907)
		(3.12, 1.1412764235)
		(3.16, 1.1391261994)
		(3.2, 1.1370408259)
		(3.24, 1.1350173961)
		(3.28, 1.1330531754)
		(3.32, 1.1311455887)
		(3.36, 1.1292922088)
		(3.4, 1.1274907459)
		(3.44, 1.1257390381)
		(3.48, 1.1240350424)
		(3.52, 1.1223768266)
		(3.56, 1.1207625621)
		(3.6, 1.1191905166)
		(3.64, 1.1176590481)
		(3.68, 1.1161665991)
		(3.72, 1.1147116912)
		(3.76, 1.1132929198)
		(3.8, 1.1119089499)
		(3.84, 1.1105585118)
		(3.88, 1.1092403969)
		(3.92, 1.1079534543)
		(3.96, 1.1066965872)
		(4.0, 1.10546875)
		(4.04, 1.1042689449)
		(4.08, 1.1030962198)
		(4.12, 1.1019496651)
		(4.16, 1.1008284116)
		(4.2, 1.0997316286)
		(4.24, 1.0986585211)
		(4.28, 1.0976083286)
		(4.32, 1.0965803226)
		(4.36, 1.0955738055)
		(4.4, 1.0945881086)
		(4.44, 1.0936225908)
		(4.48, 1.0926766371)
		(4.52, 1.0917496572)
		(4.56, 1.0908410845)
		(4.6, 1.0899503747)
		(4.64, 1.089077005)
		(4.68, 1.0882204727)
		(4.72, 1.0873802942)
		(4.76, 1.0865560047)
		(4.8, 1.0857471566)
		(4.84, 1.0849533191)
		(4.88, 1.0841740772)
		(4.92, 1.0834090312)
		(4.96, 1.082657796)
		(5.0, 1.08192)
	};

	\addplot[red , thick] coordinates {
		(1.0, 1.7600767754318618)
		(1.04, 1.6235631811)
		(1.08, 1.5473858372)
		(1.12, 1.4920228156)
		(1.16, 1.4487414411)
		(1.2, 1.4135031529)
		(1.24, 1.3840293092)
		(1.28, 1.3588885597)
		(1.32, 1.3371174434)
		(1.36, 1.3180349163)
		(1.4, 1.3011415166)
		(1.44, 1.2860602181)
		(1.48, 1.2724996865)
		(1.52, 1.2602303978)
		(1.56, 1.2490685398)
		(1.6, 1.2388648259)
		(1.64, 1.2294965206)
		(1.68, 1.2208616296)
		(1.72, 1.212874585)
		(1.76, 1.2054629873)
		(1.8, 1.1985651086)
		(1.84, 1.1921279545)
		(1.88, 1.1861057405)
		(1.92, 1.1804586825)
		(1.96, 1.1751520272)
		(2.0, 1.170155267)
		(2.04, 1.1654415005)
		(2.08, 1.1609869059)
		(2.12, 1.1567703062)
		(2.16, 1.1527728066)
		(2.2, 1.1489774909)
		(2.24, 1.1453691656)
		(2.28, 1.1419341432)
		(2.32, 1.138660058)
		(2.36, 1.135535708)
		(2.4, 1.1325509201)
		(2.44, 1.129696433)
		(2.48, 1.1269637962)
		(2.52, 1.1243452825)
		(2.56, 1.1218338112)
		(2.6, 1.1194228815)
		(2.64, 1.1171065135)
		(2.68, 1.1148791967)
		(2.72, 1.1127358446)
		(2.76, 1.1106717541)
		(2.8, 1.1086825699)
		(2.84, 1.1067642527)
		(2.88, 1.1049130507)
		(2.92, 1.1031254744)
		(2.96, 1.1013982738)
		(3.0, 1.0997284182)
		(3.04, 1.0981130779)
		(3.08, 1.0965496075)
		(3.12, 1.0950355311)
		(3.16, 1.093568529)
		(3.2, 1.0921464255)
		(3.24, 1.0907671777)
		(3.28, 1.0894288653)
		(3.32, 1.0881296819)
		(3.36, 1.0868679265)
		(3.4, 1.0856419955)
		(3.44, 1.0844503765)
		(3.48, 1.0832916413)
		(3.52, 1.0821644401)
		(3.56, 1.0810674966)
		(3.6, 1.0799996027)
		(3.64, 1.078959614)
		(3.68, 1.0779464457)
		(3.72, 1.076959069)
		(3.76, 1.0759965069)
		(3.8, 1.0750578316)
		(3.84, 1.0741421612)
		(3.88, 1.0732486566)
		(3.92, 1.0723765194)
		(3.96, 1.0715249892)
		(4.0, 1.0706933414)
		(4.04, 1.069880885)
		(4.08, 1.0690869608)
		(4.12, 1.0683109397)
		(4.16, 1.0675522208)
		(4.2, 1.0668102296)
		(4.24, 1.0660844173)
		(4.28, 1.0653742585)
		(4.32, 1.0646792508)
		(4.36, 1.0639989127)
		(4.4, 1.0633327833)
		(4.44, 1.0626804207)
		(4.48, 1.0620414011)
		(4.52, 1.0614153183)
		(4.56, 1.0608017821)
		(4.6, 1.0602004181)
		(4.64, 1.0596108665)
		(4.68, 1.0590327819)
		(4.72, 1.058465832)
		(4.76, 1.0579096972)
		(4.8, 1.0573640704)
		(4.84, 1.0568286556)
		(4.88, 1.0563031682)
		(4.92, 1.0557873339)
		(4.96, 1.0552808885)
		(5.0, 1.0547835774)
	};

	\addplot[blue , thick] coordinates {
		(1.0, 2.0)
		(1.04, 1.8440470982)
		(1.08, 1.7518556271)
		(1.12, 1.6828250339)
		(1.16, 1.6276201431)
		(1.2, 1.5818024139)
		(1.24, 1.5428124067)
		(1.28, 1.5090154788)
		(1.32, 1.4792975785)
		(1.36, 1.4528638316)
		(1.4, 1.4291270783)
		(1.44, 1.4076414882)
		(1.48, 1.3880608105)
		(1.52, 1.3701109813)
		(1.56, 1.3535715262)
		(1.6, 1.3382625715)
		(1.64, 1.3240355499)
		(1.68, 1.3107664107)
		(1.72, 1.2983505665)
		(1.76, 1.2866990723)
		(1.8, 1.2757356928)
		(1.84, 1.2653946226)
		(1.88, 1.255618693)
		(1.92, 1.2463579455)
		(1.96, 1.2375684868)
		(2.0, 1.2292115606)
		(2.04, 1.2212527892)
		(2.08, 1.2136615494)
		(2.12, 1.2064104547)
		(2.16, 1.1994749228)
		(2.2, 1.1928328131)
		(2.24, 1.1864641201)
		(2.28, 1.1803507126)
		(2.32, 1.1744761119)
		(2.36, 1.1688253012)
		(2.4, 1.1633845614)
		(2.44, 1.1581413292)
		(2.48, 1.1530840738)
		(2.52, 1.1482021894)
		(2.56, 1.1434859013)
		(2.6, 1.1389261834)
		(2.64, 1.1345146856)
		(2.68, 1.1302436698)
		(2.72, 1.1261059525)
		(2.76, 1.1220948553)
		(2.8, 1.1182041592)
		(2.84, 1.1144280648)
		(2.88, 1.1107611567)
		(2.92, 1.1071983706)
		(2.96, 1.1037349652)
		(3.0, 1.1003664956)
		(3.04, 1.0970887899)
		(3.08, 1.0938979281)
		(3.12, 1.0907902227)
		(3.16, 1.0877622013)
		(3.2, 1.0848105905)
		(3.24, 1.0819323016)
		(3.28, 1.0791244174)
		(3.32, 1.0763841802)
		(3.36, 1.0737089804)
		(3.4, 1.0710963468)
		(3.44, 1.0685439371)
		(3.48, 1.0660495294)
		(3.52, 1.0636110143)
		(3.56, 1.0612263879)
		(3.6, 1.0588937448)
		(3.64, 1.0566112721)
		(3.68, 1.0543772439)
		(3.72, 1.0521900155)
		(3.76, 1.0500480192)
		(3.8, 1.0479497592)
		(3.84, 1.0458938075)
		(3.88, 1.0438788005)
		(3.92, 1.0419034347)
		(3.96, 1.0399664634)
		(4.0, 1.038066694)
		(4.04, 1.0362029847)
		(4.08, 1.0343742417)
		(4.12, 1.0325794169)
		(4.16, 1.0308175053)
		(4.2, 1.0290875425)
		(4.24, 1.0273886034)
		(4.28, 1.0257197993)
		(4.32, 1.0240802764)
		(4.36, 1.0224692144)
		(4.4, 1.0208858244)
		(4.44, 1.0193293475)
		(4.48, 1.0177990534)
		(4.52, 1.0162942389)
		(4.56, 1.0148142269)
		(4.6, 1.0133583648)
		(4.64, 1.0119260239)
		(4.68, 1.0105165975)
		(4.72, 1.0091295008)
		(4.76, 1.0077641692)
		(4.8, 1.0064200579)
		(4.84, 1.0050966406)
		(4.88, 1.0037934092)
		(4.92, 1.0025098725)
		(4.96, 1.0012455558)
		(5.0, 1.0)
	};

	\addlegendentry[no markers, teal]{brute}
	\addlegendentry[no markers, red]{$\alpha = 1.0$, $c = 4.168$}
	\addlegendentry[no markers, blue]{$\alpha = 5.0$, $c = 1.0$}

	\end{axis}
\end{tikzpicture}

%% file: plots/figure_subset_fvs.tex
\begin{tikzpicture}[scale = 0.9]
	\begin{axis}[xmin = 1, xmax = 8, ymin = 0.9, ymax = 2.1, xlabel = {approximation ratio}]

	\addplot[teal, thick] coordinates {
		(1.0, 2.0)
		(1.05, 1.8178991111)
		(1.1, 1.7152667656)
		(1.15, 1.6406357531)
		(1.2, 1.5823559323)
		(1.25, 1.534992244)
		(1.3, 1.4954616201)
		(1.35, 1.4618205222)
		(1.4, 1.4327576428)
		(1.45, 1.4073446605)
		(1.5, 1.3849001795)
		(1.55, 1.3649096489)
		(1.6, 1.3469754742)
		(1.65, 1.3307845174)
		(1.7, 1.3160861463)
		(1.75, 1.3026769593)
		(1.8, 1.2903898821)
		(1.85, 1.2790862173)
		(1.9, 1.2686497351)
		(1.95, 1.2589822102)
		(2.0, 1.25)
		(2.05, 1.2416313881)
		(2.1, 1.2338144969)
		(2.15, 1.2264956302)
		(2.2, 1.2196279447)
		(2.25, 1.2131703756)
		(2.3, 1.2070867609)
		(2.35, 1.2013451214)
		(2.4, 1.1959170667)
		(2.45, 1.1907772993)
		(2.5, 1.1859032006)
		(2.55, 1.1812744817)
		(2.6, 1.1768728882)
		(2.65, 1.1726819491)
		(2.7, 1.1686867625)
		(2.75, 1.1648738112)
		(2.8, 1.1612308047)
		(2.85, 1.1577465416)
		(2.9, 1.1544107911)
		(2.95, 1.1512141892)
		(3.0, 1.1481481481)
		(3.05, 1.1452047773)
		(3.1, 1.1423768132)
		(3.15, 1.1396575582)
		(3.2, 1.1370408259)
		(3.25, 1.1345208933)
		(3.3, 1.1320924579)
		(3.35, 1.1297505994)
		(3.4, 1.1274907459)
		(3.45, 1.1253086435)
		(3.5, 1.1232003287)
		(3.55, 1.1211621041)
		(3.6, 1.1191905166)
		(3.65, 1.1172823368)
		(3.7, 1.115434542)
		(3.75, 1.113644299)
		(3.8, 1.1119089499)
		(3.85, 1.1102259987)
		(3.9, 1.1085930989)
		(3.95, 1.1070080425)
		(4.0, 1.10546875)
		(4.05, 1.1039732612)
		(4.1, 1.1025197264)
		(4.15, 1.1011063991)
		(4.2, 1.0997316286)
		(4.25, 1.0983938537)
		(4.3, 1.0970915966)
		(4.35, 1.0958234572)
		(4.4, 1.0945881086)
		(4.45, 1.0933842917)
		(4.5, 1.0922108113)
		(4.55, 1.0910665321)
		(4.6, 1.0899503747)
		(4.65, 1.0888613127)
		(4.7, 1.0877983687)
		(4.75, 1.0867606123)
		(4.8, 1.0857471566)
		(4.85, 1.0847571561)
		(4.9, 1.0837898041)
		(4.95, 1.0828443306)
		(5.0, 1.08192)
		(5.05, 1.0810161096)
		(5.1, 1.0801319875)
		(5.15, 1.0792669909)
		(5.2, 1.0784205047)
		(5.25, 1.0775919399)
		(5.3, 1.0767807323)
		(5.35, 1.0759863413)
		(5.4, 1.0752082483)
		(5.45, 1.074445956)
		(5.5, 1.0736989872)
		(5.55, 1.0729668837)
		(5.6, 1.0722492053)
		(5.65, 1.0715455293)
		(5.7, 1.0708554491)
		(5.75, 1.0701785739)
		(5.8, 1.0695145278)
		(5.85, 1.0688629488)
		(5.9, 1.0682234886)
		(5.95, 1.0675958119)
		(6.0, 1.0669795953)
		(6.05, 1.0663745276)
		(6.1, 1.0657803083)
		(6.15, 1.065196648)
		(6.2, 1.0646232673)
		(6.25, 1.0640598968)
		(6.3, 1.0635062761)
		(6.35, 1.0629621541)
		(6.4, 1.0624272882)
		(6.45, 1.0619014437)
		(6.5, 1.0613843942)
		(6.55, 1.0608759205)
		(6.6, 1.0603758109)
		(6.65, 1.0598838604)
		(6.7, 1.0593998706)
		(6.75, 1.0589236499)
		(6.8, 1.0584550123)
		(6.85, 1.057993778)
		(6.9, 1.0575397728)
		(6.95, 1.0570928278)
		(7.0, 1.0566527795)
		(7.05, 1.0562194693)
		(7.1, 1.0557927435)
		(7.15, 1.055372453)
		(7.2, 1.0549584531)
		(7.25, 1.0545506036)
		(7.3, 1.0541487683)
		(7.35, 1.0537528151)
		(7.4, 1.0533626155)
		(7.45, 1.0529780449)
		(7.5, 1.0525989825)
		(7.55, 1.0522253104)
		(7.6, 1.0518569146)
		(7.65, 1.0514936839)
		(7.7, 1.0511355103)
		(7.75, 1.0507822889)
		(7.8, 1.0504339177)
		(7.85, 1.0500902973)
		(7.9, 1.0497513311)
		(7.95, 1.0494169252)
		(8.0, 1.049086988)
	};

	\addplot[blue , thick] coordinates {
		(1.0, 2.0)
		(1.05, 1.8178991106)
		(1.1, 1.7152667223)
		(1.15, 1.6406352483)
		(1.2, 1.5823533032)
		(1.25, 1.5349833836)
		(1.3, 1.4954388509)
		(1.35, 1.4617718511)
		(1.4, 1.4326665206)
		(1.45, 1.4071902191)
		(1.5, 1.3846578328)
		(1.55, 1.364551915)
		(1.6, 1.3464728839)
		(1.65, 1.3301065047)
		(1.7, 1.3152018512)
		(1.75, 1.3015559042)
		(1.8, 1.2890025109)
		(1.85, 1.2774042993)
		(1.9, 1.2666466492)
		(1.95, 1.2566331271)
		(2.0, 1.2472819815)
		(2.05, 1.2385234214)
		(2.1, 1.2302974808)
		(2.15, 1.2225523252)
		(2.2, 1.2152428985)
		(2.25, 1.2083298329)
		(2.3, 1.2017785626)
		(2.35, 1.1955585999)
		(2.4, 1.1896429397)
		(2.45, 1.184007565)
		(2.5, 1.1786310363)
		(2.55, 1.1734941468)
		(2.6, 1.1685796325)
		(2.65, 1.1638719265)
		(2.7, 1.1593569507)
		(2.75, 1.155021937)
		(2.8, 1.1508552747)
		(2.85, 1.146846379)
		(2.9, 1.1429855767)
		(2.95, 1.139264008)
		(3.0, 1.1356735398)
		(3.05, 1.1322066914)
		(3.1, 1.1288565679)
		(3.15, 1.1256168026)
		(3.2, 1.1224815063)
		(3.25, 1.1194452214)
		(3.3, 1.1165028825)
		(3.35, 1.1136497806)
		(3.4, 1.1108815314)
		(3.45, 1.1081940471)
		(3.5, 1.1055835106)
		(3.55, 1.1030463537)
		(3.6, 1.1005792355)
		(3.65, 1.0981790248)
		(3.7, 1.0958427833)
		(3.75, 1.0935677503)
		(3.8, 1.0913513293)
		(3.85, 1.0891910755)
		(3.9, 1.0870846848)
		(3.95, 1.0850299831)
		(4.0, 1.0830249175)
		(4.05, 1.0810675474)
		(4.1, 1.0791560367)
		(4.15, 1.0772886469)
		(4.2, 1.0754637304)
		(4.25, 1.0736797241)
		(4.3, 1.0719351447)
		(4.35, 1.0702285827)
		(4.4, 1.0685586983)
		(4.45, 1.0669242166)
		(4.5, 1.0653239242)
		(4.55, 1.0637566646)
		(4.6, 1.0622213356)
		(4.65, 1.0607168857)
		(4.7, 1.0592423109)
		(4.75, 1.0577966526)
		(4.8, 1.0563789943)
		(4.85, 1.0549884598)
		(4.9, 1.0536242107)
		(4.95, 1.0522854442)
		(5.0, 1.0509713915)
		(5.05, 1.0496813158)
		(5.1, 1.0484145104)
		(5.15, 1.0471702976)
		(5.2, 1.0459480267)
		(5.25, 1.0447470731)
		(5.3, 1.0435668365)
		(5.35, 1.0424067401)
		(5.4, 1.0412662292)
		(5.45, 1.0401447701)
		(5.5, 1.0390418493)
		(5.55, 1.0379569723)
		(5.6, 1.0368896628)
		(5.65, 1.035839462)
		(5.7, 1.0348059274)
		(5.75, 1.0337886325)
		(5.8, 1.0327871657)
		(5.85, 1.0318011299)
		(5.9, 1.0308301417)
		(5.95, 1.0298738309)
		(6.0, 1.0289318397)
		(6.05, 1.0280038227)
		(6.1, 1.0270894455)
		(6.15, 1.0261883851)
		(6.2, 1.025300329)
		(6.25, 1.0244249747)
		(6.3, 1.0235620295)
		(6.35, 1.0227112101)
		(6.4, 1.0218722418)
		(6.45, 1.0210448588)
		(6.5, 1.0202288035)
		(6.55, 1.019423826)
		(6.6, 1.0186296842)
		(6.65, 1.0178461432)
		(6.7, 1.0170729752)
		(6.75, 1.016309959)
		(6.8, 1.0155568802)
		(6.85, 1.0148135304)
		(6.9, 1.0140797072)
		(6.95, 1.0133552143)
		(7.0, 1.0126398606)
		(7.05, 1.0119334606)
		(7.1, 1.0112358341)
		(7.15, 1.0105468058)
		(7.2, 1.0098662051)
		(7.25, 1.0091938662)
		(7.3, 1.0085296278)
		(7.35, 1.007873333)
		(7.4, 1.007224829)
		(7.45, 1.006583967)
		(7.5, 1.0059506023)
		(7.55, 1.0053245939)
		(7.6, 1.0047058043)
		(7.65, 1.0040940997)
		(7.7, 1.0034893499)
		(7.75, 1.0028914276)
		(7.8, 1.0023002089)
		(7.85, 1.0017155732)
		(7.9, 1.0011374026)
		(7.95, 1.0005655822)
		(8.0, 1.0)
	};

	\addlegendentry[no markers, teal]{brute}
	\addlegendentry[no markers, blue]{$\alpha = 8.0$, $c = 1.0$}

	\end{axis}
\end{tikzpicture}

%% file: plots/figure_pvc.tex
\begin{tikzpicture}[scale = 0.9]
	\begin{axis}[xmin = 1, xmax = 2, ymin = 0.9, ymax = 2.1, xlabel = {approximation ratio}]

	\addplot[teal, thick] coordinates {
		(1.0, 2.0)
		(1.01, 1.9454431345)
		(1.02, 1.9062508454)
		(1.03, 1.8731549013)
		(1.04, 1.8440491304)
		(1.05, 1.8178991111)
		(1.06, 1.794081097)
		(1.07, 1.7721758604)
		(1.08, 1.7518817491)
		(1.09, 1.7329712548)
		(1.1, 1.7152667656)
		(1.11, 1.698625911)
		(1.12, 1.6829321593)
		(1.13, 1.6680884977)
		(1.14, 1.6540130203)
		(1.15, 1.6406357531)
		(1.16, 1.6278963084)
		(1.17, 1.615742114)
		(1.18, 1.6041270506)
		(1.19, 1.5930103866)
		(1.2, 1.5823559323)
		(1.21, 1.5721313608)
		(1.22, 1.5623076557)
		(1.23, 1.552858658)
		(1.24, 1.5437606905)
		(1.25, 1.534992244)
		(1.26, 1.5265337131)
		(1.27, 1.5183671728)
		(1.28, 1.510476187)
		(1.29, 1.5028456449)
		(1.3, 1.4954616201)
		(1.31, 1.4883112476)
		(1.32, 1.4813826173)
		(1.33, 1.4746646809)
		(1.34, 1.4681471696)
		(1.35, 1.4618205222)
		(1.36, 1.4556758212)
		(1.37, 1.4497047356)
		(1.38, 1.443899471)
		(1.39, 1.4382527242)
		(1.4, 1.4327576428)
		(1.41, 1.427407789)
		(1.42, 1.4221971069)
		(1.43, 1.4171198929)
		(1.44, 1.4121707695)
		(1.45, 1.4073446605)
		(1.46, 1.4026367697)
		(1.47, 1.3980425604)
		(1.48, 1.3935577374)
		(1.49, 1.3891782307)
		(1.5, 1.3849001795)
		(1.51, 1.380719919)
		(1.52, 1.3766339674)
		(1.53, 1.3726390137)
		(1.54, 1.3687319076)
		(1.55, 1.3649096489)
		(1.56, 1.3611693784)
		(1.57, 1.3575083696)
		(1.58, 1.3539240206)
		(1.59, 1.3504138468)
		(1.6, 1.3469754742)
		(1.61, 1.343606633)
		(1.62, 1.3403051519)
		(1.63, 1.3370689522)
		(1.64, 1.3338960432)
		(1.65, 1.3307845174)
		(1.66, 1.3277325456)
		(1.67, 1.3247383734)
		(1.68, 1.3218003168)
		(1.69, 1.3189167589)
		(1.7, 1.3160861463)
		(1.71, 1.3133069861)
		(1.72, 1.3105778425)
		(1.73, 1.3078973347)
		(1.74, 1.3052641335)
		(1.75, 1.3026769593)
		(1.76, 1.3001345795)
		(1.77, 1.2976358065)
		(1.78, 1.2951794954)
		(1.79, 1.2927645423)
		(1.8, 1.2903898821)
		(1.81, 1.2880544871)
		(1.82, 1.2857573652)
		(1.83, 1.2834975581)
		(1.84, 1.2812741403)
		(1.85, 1.2790862173)
		(1.86, 1.2769329244)
		(1.87, 1.2748134255)
		(1.88, 1.2727269118)
		(1.89, 1.2706726008)
		(1.9, 1.2686497351)
		(1.91, 1.2666575813)
		(1.92, 1.2646954293)
		(1.93, 1.2627625911)
		(1.94, 1.2608584001)
		(1.95, 1.2589822102)
		(1.96, 1.2571333949)
		(1.97, 1.2553113469)
		(1.98, 1.2535154767)
		(1.99, 1.2517452127)
		(2.0, 1.25)
	};

	\addplot[blue, thick] coordinates {
		(1.0, 2.0)
		(1.01, 1.9387047513)
		(1.02, 1.8930929239)
		(1.03, 1.853849694)
		(1.04, 1.8188398378)
		(1.05, 1.7870069014)
		(1.06, 1.7577092525)
		(1.07, 1.7305126)
		(1.08, 1.705102307)
		(1.09, 1.6812394878)
		(1.1, 1.6587364406)
		(1.11, 1.6374417619)
		(1.12, 1.6172307723)
		(1.13, 1.5979990628)
		(1.14, 1.5796579792)
		(1.15, 1.5621313614)
		(1.16, 1.545353129)
		(1.17, 1.5292654531)
		(1.18, 1.5138173453)
		(1.19, 1.498963551)
		(1.2, 1.484663669)
		(1.21, 1.4708814415)
		(1.22, 1.457584176)
		(1.23, 1.4447422687)
		(1.24, 1.4323288084)
		(1.25, 1.4203192453)
		(1.26, 1.4086911107)
		(1.27, 1.3974237788)
		(1.28, 1.3864982635)
		(1.29, 1.3758970425)
		(1.3, 1.3656039065)
		(1.31, 1.3556038269)
		(1.32, 1.3458828408)
		(1.33, 1.3364279501)
		(1.34, 1.3272270326)
		(1.35, 1.3182687633)
		(1.36, 1.3095425448)
		(1.37, 1.3010384449)
		(1.38, 1.2927471418)
		(1.39, 1.2846598738)
		(1.4, 1.2767683951)
		(1.41, 1.2690649358)
		(1.42, 1.2615421652)
		(1.43, 1.2541931594)
		(1.44, 1.2470113718)
		(1.45, 1.2399906054)
		(1.46, 1.233124989)
		(1.47, 1.2264089541)
		(1.48, 1.2198372148)
		(1.49, 1.2134047491)
		(1.5, 1.2071067812)
		(1.51, 1.2009387662)
		(1.52, 1.1948963754)
		(1.53, 1.1889754825)
		(1.54, 1.1831721518)
		(1.55, 1.1774826264)
		(1.56, 1.1719033178)
		(1.57, 1.1664307958)
		(1.58, 1.1610617798)
		(1.59, 1.1557931299)
		(1.6, 1.1506218396)
		(1.61, 1.1455450278)
		(1.62, 1.1405599327)
		(1.63, 1.1356639047)
		(1.64, 1.1308544012)
		(1.65, 1.1261289803)
		(1.66, 1.1214852963)
		(1.67, 1.1169210943)
		(1.68, 1.112434206)
		(1.69, 1.1080225451)
		(1.7, 1.1036841036)
		(1.71, 1.0994169478)
		(1.72, 1.0952192149)
		(1.73, 1.0910891093)
		(1.74, 1.0870249001)
		(1.75, 1.0830249175)
		(1.76, 1.0790875504)
		(1.77, 1.0752112435)
		(1.78, 1.0713944949)
		(1.79, 1.067635854)
		(1.8, 1.0639339188)
		(1.81, 1.060287334)
		(1.82, 1.0566947891)
		(1.83, 1.0531550165)
		(1.84, 1.0496667894)
		(1.85, 1.0462289206)
		(1.86, 1.0428402603)
		(1.87, 1.0394996953)
		(1.88, 1.0362061466)
		(1.89, 1.0329585688)
		(1.9, 1.0297559486)
		(1.91, 1.0265973033)
		(1.92, 1.0234816797)
		(1.93, 1.0204081532)
		(1.94, 1.0173758263)
		(1.95, 1.0143838281)
		(1.96, 1.0114313128)
		(1.97, 1.008517459)
		(1.98, 1.0056414688)
		(1.99, 1.002802567)
		(2.0, 1.0)
	};

	\addlegendentry[no markers, teal]{brute}
	\addlegendentry[no markers, blue]{$\alpha = 2.0$, $c = 1.0$}

	\end{axis}
\end{tikzpicture}